\documentclass[11pt]{article}

\usepackage[margin=1in]{geometry}
\usepackage[utf8]{inputenc}
\usepackage[T1]{fontenc}
\usepackage{amsmath,amssymb,amsthm,mathtools}
\usepackage{booktabs,threeparttable,array,tabularx,multirow,longtable}
\usepackage{enumitem}
\usepackage{graphicx}
\usepackage{subcaption}
\usepackage{tikz}
\usetikzlibrary{arrows.meta,positioning,decorations.pathreplacing,calc,fit,patterns}
\usepackage[backend=biber,style=authoryear,natbib=true,maxcitenames=2,maxbibnames=99,uniquelist=false,uniquename=false]{biblatex}
\usepackage[hyperfootnotes=false,hypertexnames=false]{hyperref}
\usepackage{xcolor}
\usepackage{setspace}
\usepackage{caption}
\usepackage{microtype}
\usepackage{placeins}
\usepackage{fancyhdr}

\hypersetup{colorlinks=true,citecolor=blue!65!black,linkcolor=blue!65!black,urlcolor=blue!65!black}
\graphicspath{{figures/}}
\allowdisplaybreaks
\setlist[itemize]{leftmargin=1.8em,itemsep=0.18em,topsep=0.25em}
\setlist[enumerate]{leftmargin=1.9em,itemsep=0.18em,topsep=0.25em}
\theoremstyle{plain}
\newtheorem{theorem}{Theorem}[section]
\newtheorem{lemma}{Lemma}[section]
\newtheorem{proposition}{Proposition}[section]
\newtheorem{corollary}{Corollary}[section]
\theoremstyle{definition}
\newtheorem{assumption}{Assumption}[section]
\newtheorem{definition}{Definition}[section]
\newtheorem{procedure}{Procedure}[section]
\newtheorem{remark}{Remark}[section]

\newcommand{\rideshareexample}{%
  \paragraph{Running example (rideshare pricing).}%
}
\newcommand{\W}{\mathbf W}
\newcommand{\w}{\mathbf w}
\newcommand{\Y}{\mathbf Y}
\newcommand{\Wobs}{\mathbf W^{\mathrm{obs}}}
\newcommand{\Yobs}{\mathbf Y^{\mathrm{obs}}}
\newcommand{\one}{\mathbf 1}
\newcommand{\zero}{\mathbf 0}
\newcommand{\calW}{\mathcal W}
\newcommand{\calU}{\mathcal U}
\newcommand{\calS}{\mathcal S}
\newcommand{\calJ}{\mathcal J}
\newcommand{\E}{\mathbb E}
\newcommand{\pr}{\mathbb P}
\newcommand{\Var}{\operatorname{Var}}
\newcommand{\ind}[1]{\mathbf 1\{#1\}}
\newcommand{\OA}{Online Appendix}
\newcommand{\notes}[1]{\linespread{0.2}\vspace{0.1em}\captionsetup{justification=justified}\caption*{\footnotesize #1}}

\newif\ifsubmissionanonymous
\ifdefined\isanonymous\submissionanonymoustrue\else\submissionanonymousfalse\fi

\renewcommand{\headrulewidth}{0.2pt}

\makeatletter
\let\MSmaketitle\maketitle
\let\MSinternalmaketitle\@maketitle
\let\MStitle\title
\let\MSauthor\author
\let\MSdate\date
\let\MSthanks\thanks
\let\MSand\and
\makeatother

\begin{document}
\begin{refsection}

\ifsubmissionanonymous
  \title{Randomization Tests in Switchback Experiments}
  \author{}
  \date{}
\else
  \title{Randomization Tests in Switchback Experiments\thanks{We thank Jinglong Zhao, Yuehao Bai, Max Tabord-Meehan, Ruoxuan Xiong, participants at the 2026 Symposium on Econometric Theory and Applications, and seminar participants at the University of Tokyo for helpful comments. All errors are our own.}}
  \author{Jizhou Liu\thanks{HSBC Business School, Peking University; \texttt{jizhou.liu@phbs.pku.edu.cn}} \and Liang Zhong\thanks{The University of Hong Kong; \texttt{samzl@hku.hk}}}
  \date{September 20, 2026}
\fi
\maketitle

\begin{abstract}
Switchback experiments assign an experimental unit, such as a market or a platform, to treatment or control over successive blocks of time. Inference can be challenging because these experiments often contain only a small number of randomized blocks, while outcomes may exhibit serial dependence, seasonality, and treatment effects that persist across periods. We develop a conditional randomization test for the null of no total treatment effect that is finite-sample valid under standard assumptions on temporal interference without requiring a parametric model for outcomes. We also develop randomization tests for these assumptions: a carryover test that assesses whether past assignments continue to affect outcomes beyond a prespecified horizon and a non-anticipation test that assesses whether future assignments affect current outcomes. For hypotheses about average treatment effects, we establish asymptotic validity of studentized randomization tests under additional regularity conditions. Finally, we derive power approximations and characterize the tradeoffs between experimental design choices, the number of informative randomized comparisons, and statistical power. Numerical experiments with stylized potential outcomes and a dynamic rideshare model illustrate finite-sample performance and implications for experimental design.
\end{abstract}


\section{Introduction}\label{sec:intro}

Randomized experiments inform decisions in marketplaces, transportation platforms, and digital advertising \citep{whymarketplace,Blake2015,Gordon2019,kohavi2020trustworthy,Wager2021,Bojinov2022Online,johari2021,li2022,Christensen2023,Tang2024,Masoero2026}. In these settings, treatment assigned to one user may affect the outcomes of other users. Randomizing users separately can therefore make it difficult to evaluate the effect of applying a treatment to the entire market or platform. Also, in some applications, individual-level assignments are infeasible or outcomes are available only at an aggregate level. These problems motivate \emph{switchback experiments}, which assign the entire system to treatment or control over successive blocks of time \citep{Bojinov2023,jiang2025principledanalysiscrossoverdesigns,lin2025unifyingregressionbaseddesignbasedcausal}.

Switchback experiments present two challenges for inference. First, decisions often must be made within days or weeks, leaving only a modest number of independently randomized blocks \citep{kohavi2020trustworthy,gupta2019top}. Furthermore, serial dependence, seasonality, and heavy-tailed disturbances can make large-sample approximations inaccurate in such experiments \citep{Bojinov2023}. These challenges motivate randomization tests \citep{Fisher1935Design}, which can provide
finite-sample size control by using the known treatment-assignment
mechanism rather than relying on large-sample approximations. Second, outcomes may depend on assignments in other periods. \emph{Carryover} occurs when past assignments affect current outcomes, and \emph{anticipation} occurs when future assignments affect current outcomes. In a rideshare pricing experiment, for example, driver locations and trips in progress may reflect an earlier pricing rule. An announced future price change may also affect rider or driver behavior before the change takes effect. The literature typically proposes assumptions on carryover and anticipation without offering finite-sample-valid methods to diagnose whether these assumptions are reasonable.

In this paper, we develop a randomization test for the \emph{total treatment effect}: the difference between the potential outcome under treatment in every period and the potential outcome under control in every period. Our main null hypothesis states that this effect is zero at every period. The total-effect null equates outcomes when all of these assignments are treatment with outcomes when all are control, but it leaves outcomes under histories containing both treatment and control unrestricted. The null is therefore \emph{partially sharp}: it determines some, but not all, of the counterfactual outcomes needed for randomization inference.  A
conventional Fisher randomization test therefore cannot generally be
implemented by simply reassigning all blocks, because some reassignments require counterfactual outcomes that the null does not determine \citep{Athey2018,Basse2019,zhong2024unconditional}.

Our first contribution is a finite-sample-valid conditional randomization test (CRT) for this null, under the assumptions that treatment effects carryover for at most $m$ periods and no anticipation. Under these assumptions, the current outcome depends only on the current assignment and the preceding $m$ assignments. Before observing the assignments, we combine adjacent randomized blocks into disjoint intervals, which we call \emph{sections}. After observing the assignments, we retain sections that received treatment throughout or control throughout and exclude the first $m$ periods of each retained section. The randomization distribution reassigns each retained section either entirely to treatment or entirely to control, while holding all other assignments fixed. Every outcome used by the test then has a recent history consisting entirely of treatment or entirely of control. The null therefore determines the outcome under every reassignment considered by the test. 


Our second contribution is a pair of finite sample valid randomization tests for the assumptions about carryover and anticipation. For a proposed carryover horizon $m$, we test whether assignments more than $m$ periods in the past can affect current outcomes, under the maintained assumption of non-anticipation. The carryover nulls are nested: if there is no carryover beyond a given horizon, there is also no carryover beyond any longer horizon. Testing horizons from shortest to longest and stopping at the first nonrejection therefore controls the family-wise error rate (FWER), the probability of rejecting at least one true null. In other words, the rejected horizons provide a lower bound on the duration of carryover. To test non-anticipation itself, we adapt the pairwise imputation-based randomization test (PIRT) of \citet{zhong2024unconditional}. We compare assignment sequences that agree through a certain number of blocks but may differ later, and use outcomes observed before the two assignment sequences first differ to test whether future treatment assignments affect current outcomes.

Our third contribution concerns \emph{weak null hypotheses}, which
restrict average treatment effects while allowing positive and negative effects to offset one another. These hypotheses do not determine missing potential outcomes, so finite-sample validity is no longer guaranteed. We instead use studentized statistics, recomputing both the effect estimator and its conservative covariance estimate under each reassignment. Under suitable design and regularity conditions, we establish that the resulting randomization tests control size asymptotically despite treatment-effect heterogeneity \citep{LiDing2017,Wu2021,ZhaoDing2021}.

We then examine power and experimental design. Under a distributed-lag model with first-order autoregressive disturbances, we derive approximations that show how block length, the number of initial periods excluded, the grouping of blocks, serial correlation, and treatment probabilities affect power. We complement the theoretical analysis with simulations evaluating the finite-sample size and power of the total-effect test under Gaussian and Cauchy disturbances. We also assess the ability of the diagnostic tests to detect violations of the carryover and non-anticipation assumptions. Finally, a dynamic rideshare simulation illustrates how block length
and the assumed duration of carryover affect the power to detect
the total effect of a pricing intervention.

The paper draws on two related literatures. First, it contributes to research on switchback and time-series experiments that develops causal estimands, experimental designs, and large-sample inference in the presence of temporal dependence and carryover \citep{Bojinov2019,Bojinov2023,jiang2025principledanalysiscrossoverdesigns,lin2025unifyingregressionbaseddesignbasedcausal,missault2025robust,missault2025carryover}. Relative to this literature, we focus on finite-sample randomization inference for a partially sharp total-effect null, together with diagnostics for the temporal-interference assumptions that justify the analysis. Second, it relates to research on randomization tests under interference that develops conditional and pairwise methods for hypotheses that do not determine all counterfactual outcomes \citep{Athey2018,Basse2019,puelz2021,basse2024,liu2025randomizationinferencetwosidedmarket,zhong2024unconditional}. Our paper differs from previous literature by examining a time-series setting where interference occurs not across different units, but across time within the same unit.

The remainder of the paper is organized as follows. Section~\ref{sec:setup} introduces the experimental design, assumptions, and null hypotheses. Section~\ref{sec:main} develops the conditional randomization test for the total effect. Section~\ref{sec:diagnostics} presents tests for carryover and non-anticipation. Section~\ref{sec:weak-null-test} studies studentized tests for average treatment effects, and Section~\ref{sec:power} examines attainable $p$-values and power. Section~\ref{sec:simulations} reports the simulation results. Section~\ref{sec:conclusion} concludes. The Online Appendix contains the proofs, implementation details, additional numerical results, and a guide to the replication files.

\section{Setup and Notation}\label{sec:setup}

\subsection{The switchback design}\label{subsec:design}

We observe one outcome time series over periods $[T]=\{1,\ldots,T\}$. Let
\[
\W=(W_1,\ldots,W_T)\in\{0,1\}^T
\]
denote the treatment assignment sequence, with $W_t=1$ for treatment and $W_t=0$ for control. For any assignment sequence $\w\in\{0,1\}^T$, let $Y_t(\w)$ be the potential outcome at time $t$. The observed outcome is
\[
Y_t^{\mathrm{obs}}=Y_t(\Wobs),\qquad t=1,\ldots,T.
\]
For integers $a\le b$, write $\w_{a:b}=(w_a,\ldots,w_b)$. For any index set $D\subseteq[T]$, write $\w_D=(w_t:t\in D)$ with the indices ordered increasingly, and let $D^c=[T]\setminus D$. Let $\one_d$ and $\zero_d$ denote length-$d$ all-one and all-zero vectors.

A regular switchback design \citep{Bojinov2023} fixes the times at which treatment is reassigned before the experiment. Let
\[
\mathbb T=\{t_0=1<t_1<\cdots<t_K\}\subseteq[T],\qquad t_{K+1}=T+1,
\]
and define the assignment blocks $B_k=\{t_k,\ldots,t_{k+1}-1\}$ for $k=0,\ldots,K$. Every period within a block receives the same treatment. Adjacent blocks may receive the same treatment, so a reassignment need not result in a switch.

\begin{definition}[Regular switchback experiment]\label{def:regular-switchback}
Fix reassignment times $\mathbb T$ and treatment probabilities $\mathbb Q=(q_0,\ldots,q_K)\in(0,1)^{K+1}$. A regular switchback experiment draws independent treatment indicators for the blocks,
\[
W^{(k)}\sim\mathrm{Bernoulli}(q_k),\qquad k=0,\ldots,K,
\]
and sets $W_t=W^{(k)}$ for every $t\in B_k$.
\end{definition}

The set of assignment sequences with positive probability under this design is
\begin{equation}\label{eq:W0}
\calW_0=\{\w\in\{0,1\}^T: w_t=w_{t'}\text{ whenever }t,t'\in B_k\text{ for some }k\}.
\end{equation}
Unless stated otherwise, we treat the potential outcomes as fixed and take probabilities over the random assignment in Definition~\ref{def:regular-switchback}.

\rideshareexample
Suppose a rideshare platform conducts a city-wide experiment comparing a baseline fare of $\$2$ per mile ($W_t=0$) with a higher fare of $\$2.20$ per mile ($W_t=1$). The experiment covers $T=20{,}000$
incoming requests, divided in advance into $20$ blocks of $1,000$ requests. For each block, the platform independently chooses either fare with probability $1/2$. For example, if the first four block assignments are $(0,1,1,0)$, requests $1-1,000$ receive the baseline fare, requests $1,001-3,000$ receive the higher fare, and requests $3,001-4,000$ receive the baseline fare. Thus, periods
index requests rather than fixed intervals of calendar time, and
$20,000$ requests provide only $20$ independent treatment assignments.

\subsection{Temporal interference and potential outcomes}\label{subsec:temporal-interference}

We impose two restrictions on how assignments in other periods can affect the current outcome. The first rules out anticipation.

\begin{assumption}[Non-anticipation]\label{ass:no-anticipation}
For any $t$ and any two assignment sequences $\w,\w'\in\{0,1\}^T$,
\[
Y_t(\w)=Y_t(\w')\quad\text{whenever }\w_{1:t}=\w'_{1:t}.
\]
\end{assumption}

Assumption~\ref{ass:no-anticipation} allows the outcome to depend arbitrarily on current and past assignments, but not on future assignments. The second restriction limits the duration of carryover.

\begin{assumption}[No carryover beyond $m$ periods]\label{ass:m-carryover}
Fix $m\in\{0,\ldots,T-1\}$. For every $t\ge m+1$ and any two assignment sequences $\w,\w'\in\{0,1\}^T$,
\[
Y_t(\w)=Y_t(\w')\quad\text{whenever }\w_{t-m:T}=\w'_{t-m:T}.
\]
\end{assumption}

Assumption~\ref{ass:m-carryover} states that changing assignments more than $m$ periods in the past has no effect on the current outcome when the remaining assignments are held fixed. It restricts carryover but does not, by itself, rule out anticipation. Together with Assumption~\ref{ass:no-anticipation}, it implies that, for $t\ge m+1$, the outcome depends on the assignment sequence only through the current assignment and the preceding $m$ assignments, $(W_{t-m},\ldots,W_t)$. We refer to these $m+1$ assignments as the recent treatment history.

\rideshareexample
Take trip completion as the outcome and consider request $1,201$,
which receives the $\$2.20$ fare. Non-anticipation requires that
changing the next block's fare, while holding current and past
fares fixed, cannot affect whether this request is completed.
This assumption would fail if the rider declined the current
offer because a lower fare had been announced for the next block.
For carryover, suppose $m=200$. Changing the first block's fare
from $\$2$ to $\$2.20$ must not affect whether request $1,201$ is
completed, holding all subsequent fares fixed. The first-block
fare may affect driver availability and hence outcomes for
requests $1,001-1,200$, but its effect on outcomes must have
disappeared by request $1,201$. 

Under these assumptions, we can write potential outcomes as functions of recent treatment histories rather than complete assignment sequences.

\begin{definition}[Local potential outcome]\label{def:local-po}
Under Assumptions~\ref{ass:no-anticipation} and~\ref{ass:m-carryover}, for $t\ge m+1$ and $h\in\{0,1\}^{m+1}$, define
\begin{equation}\label{eq:local-po}
Y_t^{[m]}(h)=Y_t(\w)\quad\text{for any assignment sequence }\w\text{ such that }\w_{t-m:t}=h.
\end{equation}
The two assumptions ensure that all assignment sequences with this recent history give the same outcome, so the definition does not depend on which sequence is used.
\end{definition}

Thus $Y_t^{[m]}(\one_{m+1})$ is the outcome at time $t$ when the current period and the preceding $m$ periods all receive treatment. The superscript $[m]$ distinguishes this function of $m+1$ assignments from $Y_t(\cdot)$, whose argument is the complete assignment sequence.

\subsection{The null of no total treatment effect}\label{subsec:nulls}

The total treatment effect at time $t$ compares the potential outcome under treatment in every period with the potential outcome under control in every period. Under Assumptions~\ref{ass:no-anticipation} and~\ref{ass:m-carryover}, the null of no total treatment effect can be written as
\begin{equation}\label{eq:sharp-null-total}
H_0^{\mathrm{tot}}(m):\quad
Y_t^{[m]}(\one_{m+1})=Y_t^{[m]}(\zero_{m+1})
\quad\text{for all }t=m+1,\ldots,T.
\end{equation}
Equivalently, $Y_t(\one_T)=Y_t(\zero_T)$ for every $t\ge m+1$. This is a period-specific null, not merely a restriction on the average effect. It is partially sharp because it allows us to infer the counterfactual outcome under a recent history consisting entirely of treatment from the outcome under a recent history consisting entirely of control, and conversely. It does not determine outcomes under recent histories containing both treatment and control.

\section{Conditional randomization test for the total effect}\label{sec:main}

\subsection{From Fisher randomization to conditioning}\label{subsec:sketch}

A Fisher randomization test uses a null hypothesis to determine the counterfactual outcomes needed to evaluate a test statistic under alternative assignments \citep{Fisher1935Design}. These outcomes are called \emph{imputable} under the null. The test compares the observed statistic with its distribution under assignments drawn from the experimental design. When the null determines every outcome needed for these comparisons, the resulting $p$-value is valid in finite samples without a parametric model for outcomes.

The total-effect null in~\eqref{eq:sharp-null-total} does not generally permit this procedure over all possible switchback assignments. For example, suppose $m=2$ and each block contains one period. Redrawing all assignments may produce recent histories such as $(0,1,1)$ or $(1,0,1)$. The null equates outcomes under $(1,1,1)$ and $(0,0,0)$, but does not determine outcomes under these other histories. We therefore cannot generally evaluate the statistic after an arbitrary reassignment.

Conditional randomization tests address this problem by restricting attention to outcomes and assignments for which the required counterfactuals are imputable \citep{Athey2018,Basse2019,puelz2021,basse2024}. In our setting, we define each section as one or more adjacent assignment blocks, with its boundaries chosen before assignments or outcomes are observed. After assignment, we retain only sections that receive treatment throughout or control throughout. We exclude the first $m$ periods of each retained section so that assignments before the section cannot affect the outcomes used by the test. We call the remaining periods \emph{focal periods}.

To construct the randomization distribution, we reassign each retained section either entirely to treatment or entirely to control and hold every assignment outside these sections fixed. Under every such reassignment, the recent history of each focal period consists entirely of treatment or entirely of control. The total-effect null therefore determines its outcome. Figure~\ref{fig:conditioning-schematic} illustrates the construction.

Two features of the construction are essential. First, section boundaries must be chosen without using the realized assignments or outcomes. They may depend on pretreatment covariates and known features of the design, such as block lengths and treatment probabilities. Second, every reassignment must retain the same sections and leave all other assignments unchanged. We next define the focal periods, derive the conditional assignment distribution, and state the testing procedure.

\begin{figure}[t]
\centering
\begin{tikzpicture}[x=0.58cm,y=0.70cm,>=Latex,every node/.style={font=\scriptsize}]
  \path[use as bounding box] (-3.0,-2.15) rectangle (12.35,2.05);

  \node[anchor=east] at (-0.25,1.40){observed assignment};
  \foreach \i/\lab/\col in {0/T/gray!28,1/T/gray!28,2/T/gray!28,3/T/gray!28,
      4/C/white,5/C/white,6/C/white,7/C/white,
      8/T/gray!28,9/C/white,10/T/gray!28,11/T/gray!28}{
    \draw[fill=\col] (\i,1.10) rectangle +(1,0.55);
    \node at (\i+0.5,1.38){\lab};
  }

  \draw[decorate,decoration={brace,amplitude=4pt}] (0,1.00)--(4,1.00)
      node[midway,below=5pt]{section 1};
  \draw[decorate,decoration={brace,amplitude=4pt}] (4,1.00)--(8,1.00)
      node[midway,below=5pt]{section 2};
  \draw[decorate,decoration={brace,amplitude=4pt}] (8,1.00)--(12,1.00)
      node[midway,below=5pt]{section 3};

  \node[anchor=east] at (-0.25,-0.05){analysis when $m=2$};
  \foreach \i in {0,1,4,5}{
    \draw[pattern=north east lines,draw=gray!70] (\i,-0.32) rectangle +(1,0.55);
  }
  \foreach \i in {2,3,6,7}{
    \draw[fill=white,draw=black,very thick] (\i,-0.32) rectangle +(1,0.55);
  }
  \draw[fill=gray!12,draw=gray!65] (8,-0.32) rectangle +(4,0.55);
  \node at (1,-0.67){excluded};
  \node at (3,-0.67){focal};
  \node at (5,-0.67){excluded};
  \node at (7,-0.67){focal};
  \node at (10,-0.05){mixed: held fixed};

  \draw[<->,thick] (0,-1.25)--(8,-1.25);
  \node[fill=white,inner sep=2pt] at (4,-1.25)
      {reassign each retained section to treatment or control};
\end{tikzpicture}
\caption{Conditional randomization for the total-effect test. Each small box represents one period and one assignment block; T and C denote treatment and control. Section boundaries are chosen before assignment. Sections 1 and 2 are retained because every block within each receives the same treatment. Section 3 contains both treatment and control, so its assignments remain fixed. The first $m=2$ periods of each retained section are excluded, and the remaining periods are used to calculate the test statistic. Each reassignment assigns a retained section entirely to treatment or entirely to control.}
\label{fig:conditioning-schematic}
\end{figure}

\subsection{Predetermined sections and focal periods}\label{subsec:crt-overview}

Fix a proposed carryover horizon $m$. A \emph{predetermined section} consists of one or more consecutive assignment blocks, with boundaries chosen before observing $\Wobs$ or $\Yobs$. Let
\begin{equation}\label{eq:presections}
\calS^{\mathrm{pre}}=\{[s_1,e_1],\ldots,[s_J,e_J]\}
\end{equation}
be disjoint sections satisfying $e_j-s_j+1\ge m+1$, so that each contains at least one period after its first $m$ periods are excluded. One way to construct $\calS^{\mathrm{pre}}$ is to combine adjacent blocks until this length requirement is met. Any blocks left at the end may be added to the preceding section or excluded. This choice, like all other choices of section boundaries, must be made without using the realized assignments or outcomes.

For an assignment sequence $\w$, the sections retained for analysis are
\begin{equation}\label{eq:selected-sections}
\calS(\w)=\{[s,e]\in\calS^{\mathrm{pre}}: w_t=w_{t'}\text{ for all }t,t'\in[s,e]\}.
\end{equation}
Within a retained section $[s,e]$, the focal periods are $\{s+m,\ldots,e\}$. The complete set of periods used to calculate the test statistic is therefore
\begin{equation}\label{eqn:focal-units}
\calU(\w)=\bigcup_{[s,e]\in\calS(\w)}\{s+m,\ldots,e\}.
\end{equation}
If $\calS(\w)$ is empty, no outcomes are available for the test. In this case, the procedure returns $p=1$ and reports ``no sections retained.'' 

Let $\calS^{\mathrm{obs}}=\calS(\Wobs)$. The periods contained in these retained sections form the set
\begin{equation}\label{eq:selected-union}
\mathcal A^{\mathrm{obs}}
=
\bigcup_{[s,e]\in\calS^{\mathrm{obs}}}\{s,s+1,\ldots,e\}.
\end{equation}
We condition on the event
\begin{equation}\label{eq:conditioning-event-total}
\mathcal C^{\mathrm{obs}}
=
\left\{
\calS(\W)=\calS^{\mathrm{obs}},
\quad
\W_{(\mathcal A^{\mathrm{obs}})^c}
=
\Wobs_{(\mathcal A^{\mathrm{obs}})^c}
\right\}.
\end{equation}
This event requires the same sections to be retained and all assignments outside them to remain at their observed values. The corresponding set of possible assignment sequences is
\begin{equation}\label{eqn:cond-space}
\begin{aligned}
\calW(\mathcal C^{\mathrm{obs}})=\{\w\in\calW_0:
&\ w_t=w_{t'}\text{ for all }t,t'\in[s,e],\ [s,e]\in\calS^{\mathrm{obs}},\\
&\ w_t=W_t^{\mathrm{obs}}\text{ for every }t\in(\mathcal A^{\mathrm{obs}})^c\}.
\end{aligned}
\end{equation}
Only the treatment assigned to each retained section varies. Every section that was not retained keeps its observed combination of treatment and control. Thus $\calS(\w)=\calS^{\mathrm{obs}}$ and $\calU(\w)=\calU(\Wobs)$ for every $\w\in\calW(\mathcal C^{\mathrm{obs}})$.

\subsection{Conditional assignment distribution and testing procedure}\label{subsec:section-law}

For a retained section $j=[s_j,e_j]$, let $I_j$ index the assignment blocks it contains. Under Definition~\ref{def:regular-switchback}, the probability that all these blocks receive treatment is $\prod_{k\in I_j}q_k$, and the probability that all receive control is $\prod_{k\in I_j}(1-q_k)$. Conditional on all blocks receiving the same treatment, the probability of treatment is
\begin{equation}\label{eq:section-prob}
p_j^\star=
\frac{\prod_{k\in I_j}q_k}
{\prod_{k\in I_j}q_k+\prod_{k\in I_j}(1-q_k)}.
\end{equation}
The treatment indicators of the retained sections are conditionally independent because the original block assignments are independent and the sections are disjoint.

Let $\calJ_{\mathrm{obs}}$ index the retained sections, and let $Z_j(\w)\in\{0,1\}$ indicate whether all blocks in section $j$ receive treatment or control. Write $\calU_j=\{s_j+m,\ldots,e_j\}$ for its focal periods, $n_j=|\calU_j|$ for their number, and $\bar Y_j^{\mathrm{obs}}$ for the mean observed outcome over these periods. Let $N_{\mathrm{obs}}=\sum_{j\in\calJ_{\mathrm{obs}}}n_j$. One choice of test statistic is the Horvitz--Thompson statistic
\begin{equation}\label{eq:section-ht}
T_{\mathrm{sec}}(\w;\Yobs)=
\frac{1}{N_{\mathrm{obs}}}
\sum_{j\in\calJ_{\mathrm{obs}}}
n_j
\left\{
\frac{Z_j(\w)\bar Y_j^{\mathrm{obs}}}{p_j^\star}
-
\frac{[1-Z_j(\w)]\bar Y_j^{\mathrm{obs}}}{1-p_j^\star}
\right\}.
\end{equation}
The factor $n_j/N_{\mathrm{obs}}$ weights each section by its number of focal periods, so the target is an average that gives equal weight to each focal period. Sections receive equal weights when they contain the same number of focal periods. Other statistics, including regression-adjusted and rank-based statistics, are also valid provided they use outcomes only from the fixed focal set and assignments only from the retained sections.

\begin{procedure}[Conditional randomization test for the total-effect null]\label{proc:section-crt}
Before observing assignments or outcomes, choose $m$, the sections $\calS^{\mathrm{pre}}$, and a test statistic $T$.
\begin{enumerate}
\item Compute $\calS^{\mathrm{obs}}=\calS(\Wobs)$ and $\calU^{\mathrm{obs}}=\calU(\Wobs)$. If $\calU^{\mathrm{obs}}=\varnothing$, return $\widehat p=1$, report that no sections were retained, and stop.
\item Compute the observed statistic $T^{\mathrm{obs}}=T(\Wobs;\{Y_t^{\mathrm{obs}}:t\in\calU^{\mathrm{obs}}\})$.
\item For $b=1,\ldots,B$, draw $\W^{(b)}$ from the design distribution conditional on the event $\mathcal C^{\mathrm{obs}}$ in~\eqref{eq:conditioning-event-total}; equivalently, independently assign each retained section to treatment with probability $p_j^\star$ and to control otherwise, while holding all other assignments fixed. Compute $T^{(b)}$ using the same focal outcomes.
\item Return the Monte Carlo $p$-value
\[
\widehat p=\frac{1+\sum_{b=1}^B\ind{T^{(b)}\ge T^{\mathrm{obs}}}}{B+1}
\]
for an upper-tail test. For a two-sided test, compare the absolute values of the observed and resampled statistics.
\end{enumerate}
\end{procedure}

\begin{lemma}[Conditional assignment distribution]\label{lem:conditional-section-law}
Under Definition~\ref{def:regular-switchback}, conditional on $\mathcal C^{\mathrm{obs}}$ in~\eqref{eq:conditioning-event-total}, the treatment indicators of the retained sections are mutually independent Bernoulli variables with probabilities $p_j^\star$ in~\eqref{eq:section-prob}.
\end{lemma}

\begin{theorem}[Finite-sample validity of the section CRT]\label{thm:valid-crt}
Suppose the assignment follows Definition~\ref{def:regular-switchback}. Fix a carryover horizon $m$ and a predetermined collection of sections. Suppose Assumptions~\ref{ass:no-anticipation} and~\ref{ass:m-carryover} hold, and suppose the total-effect null $H_0^{\mathrm{tot}}(m)$ in~\eqref{eq:sharp-null-total} is true. For every realized conditioning event $\mathcal C^{\mathrm{obs}}$ in~\eqref{eq:conditioning-event-total} having positive design probability,
\[
\pr\!\left(
\widehat p\le\alpha
\,\middle|\,
\mathcal C^{\mathrm{obs}}
\right)
\le \alpha,
\qquad
\alpha\in(0,1).
\]
Consequently, $\pr(\widehat p\le\alpha)\le\alpha$. The result applies to the Monte Carlo $p$-value in Procedure~\ref{proc:section-crt}, with probability taken over both the original assignment and the Monte Carlo reference draws. If no focal periods are available, the procedure returns $\widehat p=1$.
\end{theorem}

\paragraph{Choice of test statistic.}
Theorem~\ref{thm:valid-crt} does not require a particular statistic. Validity follows because the null determines the focal outcomes under every allowed reassignment and because the test uses the correct conditional assignment probabilities. The Online Appendix provides the proof. The statistic in~\eqref{eq:section-ht} is a natural choice: conditional on the retained sections and the assignments outside them, it is unbiased for the average total effect across focal periods. Other choices may improve power. For example, predetermined calendar variables can be used to remove seasonal variation, regression adjustment can account for prespecified covariates, and ranks or trimmed section means can reduce sensitivity to extreme outcomes. The statistic must use outcomes only from the fixed focal set and must be recomputed in the same way under every reassignment.


\rideshareexample
Before assignment, pair successive blocks into sections of $2,000$
requests. For this illustration, suppose the first six block
assignments are $(0,0,1,1,0,1)$. The first section uses the $\$2$ fare
throughout, and the second uses the $\$2.20$ fare throughout; both
are retained. The third section mixes fares and is excluded.
With $m=200$, the focal periods are requests $201-2,000$ and
$2,201-4,000$. The test reassigns each retained section to either
fare, moving its two blocks together, while keeping the focal
outcomes and the mixed section's assignments fixed.

\section{Diagnostics for temporal interference}\label{sec:diagnostics}

\subsection{Testing the carryover horizon}\label{subsec:carryover}

The total-effect CRT requires a proposed carryover horizon $m$. For each candidate value, consider the following null.
\begin{definition}[Null of no carryover beyond $m$ periods]\label{def:m-null}
\begin{equation}\label{eq:null-mcarry}
H_0^m:\quad
Y_t(\w)=Y_t(\w')
\quad\text{whenever }\w_{t-m:T}=\w'_{t-m:T},
\qquad t=m+1,\ldots,T.
\end{equation}
\end{definition}
Under non-anticipation and $H_0^m$, assignments before time $t-m$ cannot affect the outcome at time $t$. We test this implication by relating outcomes in each even-numbered section to the treatment assigned to the final block of the preceding odd-numbered section. After the first $m$ periods of the even section are excluded, its remaining outcomes should be unaffected by that earlier assignment. The test uses the original assignment probabilities and does not impose a functional form on carryover effects.

Let $\calS^{\mathrm{pre}}=\{[s_j,e_j]\}_{j=1}^J$ be same predetermined sections, ordered by time. Unlike the total-effect test, this test does not select sections according to whether their assignments are constant. For even-numbered section $2i$, define
\begin{equation}\label{eqn:focal-units-m}
\calU_{2i}^{(m)}=\{s_{2i}+m,\ldots,e_{2i}\},\qquad
\bar Y_{2i}^{\mathrm{obs}}=|\calU_{2i}^{(m)}|^{-1}\sum_{t\in\calU_{2i}^{(m)}}Y_t^{\mathrm{obs}}.
\end{equation}
Let $Z_{2i-1}$ be the treatment indicator for the final block of the preceding odd-numbered section. One choice of statistic is
\begin{equation}\label{eq:carry-stat}
T_m=\frac{1}{\lfloor J/2\rfloor}
\sum_{i=1}^{\lfloor J/2\rfloor}
\left\{
\frac{Z_{2i-1}\bar Y_{2i}^{\mathrm{obs}}}{\pr(Z_{2i-1}=1)}-
\frac{(1-Z_{2i-1})\bar Y_{2i}^{\mathrm{obs}}}{\pr(Z_{2i-1}=0)}
\right\}.
\end{equation}
The CRT holds fixed all assignments in the even-numbered sections and all assignments in the odd-numbered sections except those of the final blocks used in~\eqref{eq:carry-stat}. It independently redraws these remaining treatment indicators according to
\begin{equation}\label{eq:carry-randomization}
Z_{2i-1}\sim \operatorname{Bernoulli}\!\left(\pr(Z_{2i-1}=1)\right),
\qquad i=1,\ldots,\lfloor J/2\rfloor,
\end{equation}
and recomputes the statistic.

\begin{theorem}[Validity of the fixed-$m$ carryover CRT]\label{thm:valid-carryover}
Suppose the assignment follows the regular switchback design and non-anticipation holds, as stated in Assumption~\ref{ass:no-anticipation}. For each fixed $m$, the conditional randomization test based on~\eqref{eqn:focal-units-m}--\eqref{eq:carry-randomization} is finite-sample valid for $H_0^m$.
\end{theorem}

Before the primary total-effect analysis, the researcher should specify an upper bound $m_U$ on the duration of carryover, using prior knowledge of the application. The CRT can test whether the data contradict this bound. 

\begin{proposition}[Nested carryover nulls]\label{prop:nested-carryover}
If $H_0^{\bar m}$ is true for some $\bar m$, then $H_0^m$ is true for every $m\ge \bar m$.
\end{proposition}

\begin{procedure}[Sequential tests of the carryover horizon]
\label{proc:sequential-carryover}
Before assignment, choose a grid $\mathcal M=\{m_1<\cdots<m_K\}$ and a collection of sections $\calS^{\mathrm{pre}}(m)$ for each $m\in\mathcal M$. Starting with the shortest horizon, test each null at level $\alpha$ using the CRT in~\eqref{eqn:focal-units-m}--\eqref{eq:carry-randomization}. Stop at the first nonrejection. If the procedure rejects $H_0^{m_1},\ldots,H_0^{m_r}$ before stopping, report
\[
\mathcal R=\{m_1,\ldots,m_r\}
\]
as the set of rejected horizons. If $r=0$, no candidate horizon is ruled out. If every null is rejected, all horizons in the grid are ruled out. The first nonrejected value determines where the procedure stops; it should not be interpreted as an estimate of the carryover horizon.
\end{procedure}

\begin{theorem}[Family-wise error control for sequential carryover tests]
\label{thm:sequential-carryover}
If each fixed-$m$ carryover CRT is level $\alpha$, then Procedure~\ref{proc:sequential-carryover} satisfies
\[
\pr\{\text{the procedure rejects at least one true }H_0^m\}\le \alpha.
\]
Consequently, with probability at least $1-\alpha$, every horizon in the reported set $\mathcal R$ is too short. On this event, if a smallest valid horizon $m_\star$ exists and $\mathcal R\ne\varnothing$, then
\[
m_\star>\max\mathcal R.
\]
\end{theorem}

The rejected horizons therefore provide a lower confidence bound on the duration of carryover whenever at least one candidate is rejected. Nonrejection can reflect low power and does not establish that a horizon is long enough for the total-effect test. In particular, our results do not justify using the same outcomes both to select a horizon by this procedure and to test the total effect at that selected horizon. The primary total-effect test should use a prespecified horizon, and its validity requires that carryover does not extend beyond that horizon.

\rideshareexample
Recall that the block
assignments $(0,0,1,1,0,1)$ are divided into $3$ sections $(0,0), (1,1), (0,1)$. With $m=200$, the second section contributes requests
$2,201$--$4,000$ as focal periods. The test holds that section's
$\$2.20$ fare fixed and redraws the fare for requests
$1,001-2,000$ between $\$2$ and $\$2.20$.

\subsection{Testing non-anticipation}\label{subsec:nonanticipation}

The total-effect and carryover tests both assume non-anticipation.
We test this assumption by holding assignments fixed through a
prespecified block and redrawing assignments in later blocks. For
each pair of assignment sequences, we use outcomes observed before
their first difference in assignments to test whether future
assignments affect current outcomes. Because the periods available
for comparison can differ across pairs, we use a conditional version
of the pairwise imputation-based randomization test (PIRT) of
\citet{zhong2024unconditional}. The null is
\begin{equation}\label{eq:null-na}
H_0^{\mathrm{NA}}:\quad
Y_t(\w)=Y_t(\w')
\quad\text{whenever }\w_{1:t}=\w'_{1:t},
\qquad t=1,\ldots,T.
\end{equation}

Choose a time $h\in[T]$ before observing assignments or outcomes.
Let $h^+$ be the last period of the assignment block containing
$h$, and suppose $h^+<T$. We condition on
\[
\mathcal C_h^{\mathrm{obs}}
=
\{\W_{1:h^+}=\Wobs_{1:h^+}\}.
\]
Every resampled sequence therefore agrees with the observed
sequence through the block containing $h$. Under
Definition~\ref{def:regular-switchback}, assignments in later
blocks remain independent and are drawn using their original
treatment probabilities. Holding the entire block fixed, rather
than only its periods through $h$, ensures that each resampled
sequence is allowed by the switchback design.

\begin{definition}[Periods with jointly imputable outcomes]
\label{def:imput_units}
For two assignment sequences $\w,\w'\in\calW_0$, define
\begin{equation}\label{eq:imputable-times}
\mathbb I(\w,\w')
=
\{t\in[T]:\w_{1:t}=\w'_{1:t}\}.
\end{equation}
\end{definition}

The null requires the outcomes in $\mathbb I(\w,\w')$ to coincide
under the two sequences. Let
$I(\w,\w')=\mathbb I(\w,\w')\cap[T-1]$, so that the next-period
assignment is defined for every included period. For sequences
sharing the fixed initial assignments, this set is nonempty.
For an outcome vector $\Y=(Y_1,\ldots,Y_T)^\top$, consider the
statistic
\begin{equation}\label{eq:Tna-centered}
T_{\mathrm{NA}}(\Y,\w,\w')
=
\frac{1}{|I(\w,\w')|}
\sum_{t\in I(\w,\w')}
\{2w_{t+1}-1\}
\left(Y_t-\bar Y_{I(\w,\w')}\right),
\end{equation}
where $\bar Y_I=|I|^{-1}\sum_{t\in I}Y_t$.

This statistic measures the association between centered outcomes
and next-period treatment assignments over the periods with
imputable outcomes. To see which observations distinguish the
pairwise comparisons, consider two distinct sequences with the
same initial assignments and let
$d=\min\{t:w_t\ne w_t'\}$ be their first difference. Then
$I(\w,\w')=\{1,\ldots,d-1\}$ and
\begin{equation}\label{eq:Tna-boundary-identity}
T_{\mathrm{NA}}(\Y,\w',\w)
-
T_{\mathrm{NA}}(\Y,\w,\w')
=
\frac{2(w_d'-w_d)}{d-1}
\left(Y_{d-1}-\bar Y_{1:(d-1)}\right).
\end{equation}
Terms for periods before $d-1$ cancel in this difference.
Earlier outcomes enter only through their mean. Thus the comparison
depends on the outcome immediately before the assignments first
differ, relative to the mean outcome through that period, and on
the treatment assigned in the next period. Identical assignment
sequences produce a tie.

For $\W^\ast\sim\pr(\cdot\mid\mathcal C_h^{\mathrm{obs}})$, define
\begin{equation}\label{eq:pirt-na-pvalue}
p_{\mathrm{pair},h}(\Wobs)
=
\pr\!\left[
T_{\mathrm{NA}}(\Yobs,\W^\ast,\Wobs)
\ge
T_{\mathrm{NA}}(\Yobs,\Wobs,\W^\ast)
\,\middle|\,
\mathcal C_h^{\mathrm{obs}}
\right].
\end{equation}
A two-sided test compares the absolute values of the two statistics.
The pairwise probability need not be a valid $p$-value on its own.
We therefore multiply it by two and truncate at one; this
calibration applies to both one-sided and two-sided tests.

\begin{theorem}[Conditional test of non-anticipation]
\label{thm:valid-pirt-prefix}
Suppose the assignment follows
Definition~\ref{def:regular-switchback}. Choose $h$ before observing
assignments or outcomes, with $h^+<T$. Under $H_0^{\mathrm{NA}}$,
\[
\pr\!\left(
p_{\mathrm{pair},h}(\Wobs)\le\alpha/2
\,\middle|\,
\mathcal C_h^{\mathrm{obs}}
\right)
\le \alpha,
\qquad
\alpha\in(0,1).
\]
\end{theorem}

The \OA{} gives the general conditional PIRT construction, its
validity proof and Monte Carlo implementation, and a finite-sum
expression for the pairwise probability used in the simulations.

\rideshareexample
Fix the first two block assignments at $(0,0)$. Comparing
$(0,0,1,1,\ldots)$ with $(0,0,1,0,\ldots)$ uses the outcome for
request $3,000$ relative to the average for requests $1-3,000$.
This comparison targets a response to the fourth block's fare
before that block begins: for example, a rider may decline a
$\$2.20$-per-mile offer after learning that the fare will fall
to $\$2$.

\section{Inference for average treatment effects}\label{sec:weak-null-test}

The CRT in Section~\ref{sec:main} tests whether the total effect is zero at every period specified by the null. On the conditional assignment set, this null determines all focal outcomes needed to recompute the statistic. We now consider \emph{weak null hypotheses}, which restrict average treatment effects but allow individual effects to differ. Positive and negative effects may offset one another, so these hypotheses do not determine the missing potential outcomes. We use studentized statistics and establish asymptotic validity under regularity conditions \citep{chung2013,DiCiccioRomano2017,LiDing2017,Wu2021,ZhaoDing2021}.

For simplicity, in this section we consider a specialized design with $J$ independently randomized blocks of a common length $L$ and $m<L$. Then, each block is a section, and every block contributes focal periods after its first $m$ periods are excluded. There is no selection of blocks based on their realized assignments. The randomization distribution therefore redraws the original block assignments. In this equal-length-block setting, write
\[
s_j=(j-1)L+1,
\qquad
e_j=jL,
\qquad
Z_j\sim\mathrm{Bernoulli}(p_j),
\qquad
\underline p\le p_j\le 1-\underline p.
\]
Here $\underline p>0$ bounds treatment probabilities away from zero and one. Each block has $n=L-m$ focal periods. Index their positions by $\ell=1,\ldots,n$, so that position $\ell$ in block $j$ is period $s_j+m+\ell-1$. For $z\in\{0,1\}$, define
\[
Y_{j\ell}(z)
:=
Y_{s_j+m+\ell-1}^{[m]}(z\one_{m+1}),
\qquad
\boldsymbol Y_j(z)
:=
\bigl(Y_{j1}(z),\ldots,Y_{jn}(z)\bigr)^\top.
\]
Under Assumptions~\ref{ass:no-anticipation} and~\ref{ass:m-carryover}, the vector of observed focal outcomes in block $j$ is
\[
\boldsymbol Y_j^{\mathrm{obs}}
=
Z_j\boldsymbol Y_j(1)+(1-Z_j)\boldsymbol Y_j(0).
\]
The vector of finite-population average treatment effects, with one component for each focal position, is
\begin{equation}\label{eq:weak-effect-profile}
\boldsymbol\tau_J
=
\frac{1}{J}\sum_{j=1}^J
\bigl\{\boldsymbol Y_j(1)-\boldsymbol Y_j(0)\bigr\}
\in\mathbb R^n.
\end{equation}

Let $A\in\mathbb R^{d\times n}$ be a prespecified matrix of fixed dimension and full row rank. Its rows specify the linear restrictions to be tested. We consider the weak null
\begin{equation}\label{eq:weak-null-general}
H_0(A):
\qquad
A\boldsymbol\tau_J=\zero_d.
\end{equation}
This formulation includes several hypotheses of practical interest. Taking $A=n^{-1}\one_n^\top$ gives the null that the average effect across all focal periods is zero:
\begin{equation}\label{eq:weak-null-focal}
H_0^{\calU}:
\qquad
\frac{1}{Jn}\sum_{j=1}^J\sum_{\ell=1}^n
\bigl\{Y_{j\ell}(1)-Y_{j\ell}(0)\bigr\}=0.
\end{equation}
Taking $A=\mathbf I_n$, where $\mathbf I_n$ is the $n\times n$ identity matrix, gives the joint null that the effect averaged across blocks is zero at every focal position:
\begin{equation}\label{eq:weak-null-session}
H_0^{\mathrm{sw}}:
\qquad
\frac{1}{J}\sum_{j=1}^J
\bigl\{Y_{j\ell}(1)-Y_{j\ell}(0)\bigr\}=0
\quad\text{for every }\ell=1,\ldots,n.
\end{equation}
Finally, $A=\boldsymbol e_\ell^\top$, where $\boldsymbol e_\ell$ is the $\ell$th coordinate vector, gives the null that the average effect is zero at position $\ell$ alone. All these hypotheses restrict averages across blocks rather than the effect in each block, so none is a sharp null.

For a given $A$, define the transformed potential outcomes
\[
\boldsymbol X_j(z)=A\boldsymbol Y_j(z),
\qquad
\boldsymbol X_j^{\mathrm{obs}}
=A\boldsymbol Y_j^{\mathrm{obs}}
=Z_j\boldsymbol X_j(1)+(1-Z_j)\boldsymbol X_j(0).
\]
The Horvitz--Thompson estimator of $A\boldsymbol\tau_J$ is
\begin{equation}\label{eq:ht-general-vector}
\widehat{\boldsymbol\tau}_A
=
\frac{1}{J}\sum_{j=1}^J
\left\{
\frac{Z_j\boldsymbol X_j^{\mathrm{obs}}}{p_j}
-
\frac{(1-Z_j)\boldsymbol X_j^{\mathrm{obs}}}{1-p_j}
\right\}.
\end{equation}
We use the conservative covariance estimator
\begin{equation}\label{eq:Vup-vector}
\widehat{\boldsymbol V}_{A,\mathrm{up}}
=
\frac{1}{J^2}\sum_{j=1}^J
\boldsymbol X_j^{\mathrm{obs}}\boldsymbol X_j^{\mathrm{obs}\top}
\left\{
\frac{Z_j}{p_j^2}
+
\frac{1-Z_j}{(1-p_j)^2}
\right\}.
\end{equation}
For $d\ge1$, define the quadratic-form statistic
\begin{equation}\label{eq:weak-quadratic}
Q_A
=
\widehat{\boldsymbol\tau}_A^\top
\widehat{\boldsymbol V}_{A,\mathrm{up}}^{-1}
\widehat{\boldsymbol\tau}_A.
\end{equation}
The inverse is used whenever the matrix is nonsingular; under Assumption~\ref{ass:weak-null-asymp} below, nonsingularity holds with probability tending to one. In finite samples, a prespecified Moore--Penrose inverse may be used, provided the same rule is applied to the observed statistic and every reference draw. When $d=1$, $Q_A=T_A^2$, where
\begin{equation}\label{eq:weak-signed}
T_A
=
\frac{\widehat\tau_A}{\sqrt{\widehat V_{A,\mathrm{up}}}}
\end{equation}
can be used for a one-sided alternative.

The studentized randomization test holds $\{\boldsymbol X_j^{\mathrm{obs}}\}_{j=1}^J$ fixed, draws independent treatment indicators $Z_j^\ast\sim\mathrm{Bernoulli}(p_j)$, and recomputes both the estimator and its covariance estimate in~\eqref{eq:ht-general-vector}--\eqref{eq:weak-quadratic}. Let $Q_A^\ast$ be the resulting statistic. The randomization $p$-value is
\begin{equation}\label{eq:weak-pvalue}
\widehat p_A
=
\mathbb P^\ast\!\left(
Q_A^\ast\ge Q_A
\,\middle|\,
\{\boldsymbol X_j^{\mathrm{obs}},p_j\}_{j=1}^J
\right).
\end{equation}
Here $\mathbb P^\ast$ denotes probability over the resampled assignments, conditional on the observed outcomes. The probability in~\eqref{eq:weak-pvalue} can be approximated by independent draws, using the Monte Carlo correction in Procedure~\ref{proc:section-crt}. For a one-sided scalar test, compare $T_A^\ast$ with $T_A$ instead.

\begin{remark}[Centering and power]\label{rem:centering-power}
The covariance estimator in~\eqref{eq:Vup-vector} is conservative, and a test based on it can have low power when mean outcome levels or seasonal variation are large relative to treatment effects. Removing variation explained by predetermined calendar variables can improve power, as can the weighted regression alternative studied in the \OA{}. More general covariate adjustment requires additional regularity conditions.
\end{remark}

\begin{assumption}[Regularity conditions for average-effect inference]\label{ass:weak-null-asymp}
The block length $L$ and carryover horizon $m<L$ are fixed as $J\to\infty$, so $n=L-m$ is fixed. The number $d$ of rows of the prespecified matrix $A$ is also fixed. There is a constant $\underline p\in(0,1/2)$ such that $p_j\in[\underline p,1-\underline p]$ for every $j$. Let
\[
M_j
=
\max\bigl\{\|\boldsymbol Y_j(1)\|_2,\|\boldsymbol Y_j(0)\|_2\bigr\}.
\]
Assume:
\begin{enumerate}[label=(\roman*)]
\item $\sup_J J^{-1}\sum_{j=1}^J M_j^4<\infty$;
\item the empirical measures
\[
\frac{1}{J}\sum_{j=1}^J
\delta_{(p_j,\boldsymbol Y_j(1),\boldsymbol Y_j(0))}
\]
converge weakly to a probability distribution $\nu$ for $(P,\boldsymbol Y_1,\boldsymbol Y_0)$, where $\delta_{(p_j,\boldsymbol Y_j(1),\boldsymbol Y_j(0))}$ denotes a point mass;
\item defining
\begin{align*}
\boldsymbol\Sigma_{\mathrm R}
&=
\mathbb E_\nu\!\left[
\frac{\boldsymbol Y_1\boldsymbol Y_1^\top}{1-P}
+
\frac{\boldsymbol Y_0\boldsymbol Y_0^\top}{P}
\right],\\
\boldsymbol\Sigma_{\mathrm{up}}
&=
\mathbb E_\nu\!\left[
\frac{\boldsymbol Y_1\boldsymbol Y_1^\top}{P}
+
\frac{\boldsymbol Y_0\boldsymbol Y_0^\top}{1-P}
\right],
\end{align*}
the transformed matrices $A\boldsymbol\Sigma_{\mathrm R}A^\top$ and $A\boldsymbol\Sigma_{\mathrm{up}}A^\top$ are positive definite.
\end{enumerate}
\end{assumption}

\begin{theorem}[Asymptotic validity for linear restrictions on average effects]
\label{thm:weak-null-general}
Suppose Assumptions~\ref{ass:no-anticipation},~\ref{ass:m-carryover}, and~\ref{ass:weak-null-asymp} hold. Let $A\in\mathbb R^{d\times n}$ be prespecified and of full row rank. Under $H_0(A)$ in~\eqref{eq:weak-null-general}, the studentized randomization test based on $Q_A$ satisfies
\[
\limsup_{J\to\infty}
\mathbb P\{\widehat p_A\le\alpha\}
\le \alpha,
\qquad
\alpha\in(0,1).
\]
The conclusion holds when $\widehat p_A$ is computed from the randomization distribution in~\eqref{eq:weak-pvalue} or from a Monte Carlo approximation with the correction in Procedure~\ref{proc:section-crt} and $B_J\to\infty$ draws.

When $d=1$, the analogous one-sided upper-tail test based on $T_A$
is asymptotically valid for $\alpha\in(0,1/2)$.
The two-sided test based on $T_A^2=Q_A$ is asymptotically valid
for every $\alpha\in(0,1)$.
\end{theorem}

The choice of $A$ should reflect the question of interest. The overall average, obtained with $A=n^{-1}\one_n^\top$, summarizes the total treatment effect across all focal periods. The joint test with $A=\mathbf I_n$ asks whether any focal position has a nonzero effect averaged across blocks. It can detect positive and negative effects at different positions even when they cancel in the overall average. Choosing $A=\boldsymbol e_\ell^\top$ tests the average effect at one prespecified position. Separate tests at multiple positions require an appropriate adjustment for multiple testing.

\section{Power analysis}\label{sec:power}

Sections~\ref{sec:main} and~\ref{sec:diagnostics} establish validity with potential outcomes treated as fixed. We now study power under a stochastic model for the potential-outcome time series. The approximations describe how block length and the number of blocks combined into each section affect the precision of the studentized statistics.

Throughout this section, assignment blocks have a common length $L>m$ and treatment probability $q\in(0,1)$. Before the experiment, we form each section from $r$ consecutive blocks. The total-effect CRT retains a section only when all of its blocks receive the same treatment. Thus
\begin{equation}\label{eq:design-constq}
W^{(k)}\stackrel{\mathrm{ind}}{\sim}\mathrm{Bernoulli}(q),
\qquad
\pi_r(q)=\pr\{\text{section is retained}\}=q^r+(1-q)^r.
\end{equation}
Conditional on a section being retained, its treatment probability is
\begin{equation}\label{eq:p-rq}
p(r;q)=\frac{q^r}{q^r+(1-q)^r}.
\end{equation}
The number of focal periods in a retained section is
\begin{equation}\label{eq:n-def}
n=rL-m.
\end{equation}
Combining more blocks increases the number $n$ of focal periods in each retained section and can improve the precision of its mean outcome. For a fixed experimental duration, however, it also reduces the number of sections, and each section is less likely to have the same treatment throughout. If there are $J$ predetermined sections and $J_{\mathrm{tot}}$ is the number retained, then
\[
J_{\mathrm{tot}}\sim\operatorname{Binomial}\{J,\pi_r(q)\},
\qquad
\mathbb E(J_{\mathrm{tot}})=J\pi_r(q).
\]
Under balanced assignment, $\pi_r(1/2)=2^{1-r}$. A section formed from two blocks is retained with probability one half; the corresponding probabilities for three and four blocks are one quarter and one eighth.

\paragraph{Attainable $p$-values.}
The discreteness of the assignment distribution limits the possible $p$-values. Write $p=p(r;q)$, and let $J_1=\sum_{j=1}^{J_{\mathrm{tot}}}Z_j$ and $J_0=J_{\mathrm{tot}}-J_1$ be the numbers of treated and control sections. The exact upper-tail probability includes the observed assignment itself. The conditional randomization $p$-value must therefore be at least the probability of that assignment,
\begin{equation}\label{eq:observed-assignment-atom}
\omega_{\mathrm{obs}}=p^{J_1}(1-p)^{J_0}.
\end{equation}
When the retained sections have treatment probability $p=1/2$, this lower bound is $\omega_{\mathrm{obs}}=2^{-J_{\mathrm{tot}}}$. A necessary condition for rejection at level $\alpha$ is therefore
\begin{equation}\label{eq:min-retained-alpha}
J_{\mathrm{tot}}\ge j_\alpha,
\qquad
j_\alpha=\left\lceil\log_2(1/\alpha)\right\rceil.
\end{equation}
At the $5\%$ level, at least five retained sections are necessary; ties and two-sided statistics can require more. In the balanced design,
\begin{align}\label{eq:finite-support-probabilities}
\pr(J_{\mathrm{tot}}=0)
&=\{1-\pi_r(1/2)\}^J,\\
\pr(J_{\mathrm{tot}}<j_\alpha)
&=\sum_{j=0}^{j_\alpha-1}
\binom{J}{j}\pi_r(1/2)^j\{1-\pi_r(1/2)\}^{J-j}.
\end{align}
These probabilities identify designs that frequently yield no focal outcomes or too few retained sections for the exact test to reject at the nominal level. They provide information that is not captured by the normal approximations to power below.

To study power, consider the distributed-lag model
\begin{equation}\label{eq:dgp-main}
Y_t(\w)=\mu+\sum_{\ell=0}^{m_0}\beta_\ell w_{t-\ell}+\varepsilon_t,
\end{equation}
where the disturbance follows a stationary first-order autoregressive, or AR(1), process:
\begin{equation}\label{eq:ar1-main}
\varepsilon_t=\rho\varepsilon_{t-1}+\eta_t,
\qquad |\rho|<1,
\qquad
\E(\eta_t)=0,
\qquad
\Var(\eta_t)=\sigma_\eta^2>0,
\qquad
\E(\eta_t^4)<\infty.
\end{equation}
The innovations are independent and identically distributed and are independent of treatment assignment. Here $m_0$ is the true carryover horizon and $\beta_\ell$ is the effect of assignment $\ell$ periods earlier. The total effect is $\tau_{\mathrm{tot}}=\sum_{\ell=0}^{m_0}\beta_\ell$. When the analysis uses $m\ge m_0$, all assignments relevant to a focal outcome lie within its retained section. The variance of the average of $n$ consecutive AR(1) disturbances is
\begin{equation}\label{eq:var-bar-eps}
\sigma_{\bar\varepsilon}^2(n,\rho)
=\frac{\sigma_\varepsilon^2}{n^2}
\left\{n+2\sum_{h=1}^{n-1}(n-h)\rho^h\right\},
\qquad
\sigma_\varepsilon^2=\frac{\sigma_\eta^2}{1-\rho^2}.
\end{equation}

Index the retained sections by $j=1,\ldots,J_{\mathrm{tot}}$. With $p=p(r;q)$, define the estimator, variance estimate, and studentized statistic
\begin{align}
\widehat\tau_{\mathrm{HT}}
&=
\frac{1}{J_{\mathrm{tot}}}
\sum_{j=1}^{J_{\mathrm{tot}}}
\left\{
\frac{Z_j\bar Y_j^{\mathrm{obs}}}{p}
-
\frac{(1-Z_j)\bar Y_j^{\mathrm{obs}}}{1-p}
\right\},\notag\\
\widehat V_{\mathrm{up}}
&=
\frac{1}{J_{\mathrm{tot}}^2}
\sum_{j=1}^{J_{\mathrm{tot}}}
(\bar Y_j^{\mathrm{obs}})^2
\left\{
\frac{Z_j}{p^2}
+
\frac{1-Z_j}{(1-p)^2}
\right\},
\qquad
T_{\mathrm{tot}}
=
\frac{\widehat\tau_{\mathrm{HT}}}{\sqrt{\widehat V_{\mathrm{up}}}}.
\label{eq:power-stat-total}
\end{align}

\begin{proposition}[Power for the total-effect test]\label{prop:power-total}
Under~\eqref{eq:dgp-main}--\eqref{eq:ar1-main}, suppose $m_0\le m<rL$. Consider the one-sided upper-tail test based on $T_{\mathrm{tot}}$ in~\eqref{eq:power-stat-total}, and let $\varphi_{\mathrm{tot}}$ indicate rejection. Consider local alternatives indexed by the number $J$ of predetermined sections, with $\sqrt{J}\tau_{\mathrm{tot}}$ bounded. Along realizations satisfying $J_{\mathrm{tot}}/J\to\pi_r(q)$, the rejection probability conditional on $J_{\mathrm{tot}}$ is approximated by
\begin{equation}\label{eq:power-total}
\pr\{\varphi_{\mathrm{tot}}=1\mid J_{\mathrm{tot}}\}
\approx
1-\Phi\left(
\frac{z_{1-\alpha}-\mu_{\mathrm{tot}}(J_{\mathrm{tot}})}{\sigma_{\mathrm{tot}}}
\right),
\end{equation}
where $\Phi$ is the standard normal distribution function, $z_{1-\alpha}$ is its $(1-\alpha)$ quantile, and, with $p=p(r;q)$,
\begin{align}\label{eq:mu-sigma-total}
\mu_{\mathrm{tot}}(J_{\mathrm{tot}})
&=
\frac{\tau_{\mathrm{tot}}\sqrt{J_{\mathrm{tot}}}}
{\left[
\sigma_{\bar\varepsilon}^2(n,\rho)\{p(1-p)\}^{-1}
+(\mu+\tau_{\mathrm{tot}})^2p^{-1}
+\mu^2(1-p)^{-1}
\right]^{1/2}},\\
\sigma_{\mathrm{tot}}^2
&=
\frac{\sigma_{\bar\varepsilon}^2(n,\rho)+\{\mu+(1-p)\tau_{\mathrm{tot}}\}^2}
{\sigma_{\bar\varepsilon}^2(n,\rho)+p\mu^2+(1-p)(\mu+\tau_{\mathrm{tot}})^2}.
\end{align}
\end{proposition}

Proposition~\ref{prop:power-total} shows three ways in which the design affects power. First, the standardized mean $\mu_{\mathrm{tot}}(J_{\mathrm{tot}})$ increases with the square root of the number of retained sections, holding the other terms fixed. Second, serial correlation affects the variance of each section mean through~\eqref{eq:var-bar-eps}. Additional observations improve precision less when $\rho$ is close to one. Third, treatment probabilities far from one half increase the inverse-probability weights. For centered outcomes under local alternatives, balanced assignment is generally favorable for power.

For the carryover test, power depends on a weighted sum of effects beyond the proposed horizon, rather than simply on $\sum_{\ell>m}\beta_\ell$. The statistic in~\eqref{eq:carry-stat} uses the treatment assigned to the final block of the preceding odd-numbered section. A coefficient $\beta_\ell$ contributes when, for a focal period in the even-numbered section, the assignment $\ell$ periods earlier lies in that block. The resulting weighted sum is
\begin{equation}\label{eq:delta-tail-main}
\delta_m(n)=
\frac{1}{n}
\sum_{\ell=m+1}^{m_0}\min\{L,\ell-m\}\beta_\ell,
\qquad n=rL-m.
\end{equation}

Let $J_e=\lfloor J/2\rfloor$ and, under the common assignment probability $q$, define
\begin{equation}\label{eq:carry-studentized-power}
\widehat V_{\mathrm{up}}^{(m)}
=
\frac{1}{J_e^2}
\sum_{i=1}^{J_e}
(\bar Y_{2i}^{\mathrm{obs}})^2
\left\{
\frac{Z_{2i-1}}{q^2}
+
\frac{1-Z_{2i-1}}{(1-q)^2}
\right\},
\qquad
T_m^{\mathrm{stud}}
=
\frac{T_m}{\sqrt{\widehat V_{\mathrm{up}}^{(m)}}}.
\end{equation}
Let $\bar Y_2(z)$ denote the mean outcome over the focal periods of an even-numbered section when the final block of the preceding odd-numbered section is assigned $z$. Average over the stationary distribution of all other assignments and the AR(1) disturbances, and define
\begin{align}
 v_{\mathrm{up}}^{(m)}
 &=
 \mathbb E\!\left[
 \frac{\bar Y_2(1)^2}{q}
 +
 \frac{\bar Y_2(0)^2}{1-q}
 \right],\notag\\
 \sigma_\delta^2
 &=
 \mathbb E\!\left[
 \left\{
 \sqrt{\frac{1-q}{q}}\,\bar Y_2(1)
 +
 \sqrt{\frac{q}{1-q}}\,\bar Y_2(0)
 \right\}^2
 \right].
 \label{eq:carry-variance-limits}
\end{align}

\begin{proposition}[Power for the carryover test]\label{prop:power-carry}
Suppose $m<rL$, $m_0>m$, and $m_0\le rL$. Consider a triangular array of coefficient vectors
$\boldsymbol\beta_J=(\beta_{0,J},\ldots,\beta_{m_0,J})$ such that, for every $\ell=0,\ldots,m_0$,
$\beta_{\ell,J}\to\beta_{\ell,\infty}$. Write
$\boldsymbol\beta_\infty=(\beta_{0,\infty},\ldots,\beta_{m_0,\infty})$ and interpret
$\delta_m(n)$ in~\eqref{eq:delta-tail-main} using the coefficients
$\{\beta_{\ell,J}\}$. Suppose
$\sqrt{J_e}\,\delta_m(n)$ remains bounded and that the two quantities
in~\eqref{eq:carry-variance-limits}, evaluated at
$\boldsymbol\beta_\infty$, are finite and strictly positive. Then the
one-sided upper-tail rejection probability based on
$T_m^{\mathrm{stud}}$ in~\eqref{eq:carry-studentized-power} satisfies
the following approximation, where $\varphi_m$ indicates rejection:
\begin{equation}\label{eq:power-carry}
\pr\{\varphi_m=1\}
\approx
1-\Phi\left(
\frac{z_{1-\alpha}-\mu_m(J)}{\sigma_m}
\right),
\end{equation}
with
\begin{equation}\label{eq:mu-sigma-carry}
\mu_m(J)=\frac{\delta_m(n)\sqrt{\lfloor J/2\rfloor}}{\sqrt{v_{\mathrm{up}}^{(m)}}},
\qquad
\sigma_m^2=\frac{\sigma_\delta^2}{v_{\mathrm{up}}^{(m)}}.
\end{equation}
The quantities $v_{\mathrm{up}}^{(m)}$ and $\sigma_\delta^2$ are defined in~\eqref{eq:carry-variance-limits}.
\end{proposition}

The carryover test is designed to detect effects of the treatment assigned to the final block of an odd-numbered section on focal outcomes in the following section. It may have low power even when $H_0^m$ is false. Effects at different lags may offset one another, or the assignments responsible for the carryover may lie outside the block being reassigned. The choice of block length, section boundaries, and excluded periods should account for these possibilities.

For a fixed experimental duration, short blocks provide more independent assignments but may leave few observations once the first $m$ periods are excluded. Longer blocks provide more focal observations within each block but fewer assignments overall. Combining blocks can increase the precision of section means, but fewer sections are available and each is less likely to receive the same treatment throughout. Pilot data can help specify plausible values of $\rho$, $\sigma_\varepsilon^2$, and the lag coefficients $\beta_\ell$ for evaluating the power approximations.

These calculations can be used before the experiment to compare feasible designs. First, specify a plausible range of carryover horizons using prior experiments, simulations, or knowledge of the application. Next, consider feasible block lengths and ways of combining blocks into sections. For each design, calculate the expected number of retained sections, the probability that none is retained, and the probability that fewer than $j_\alpha$ are retained. With unequal treatment probabilities, also examine the probability of the observed assignment in~\eqref{eq:observed-assignment-atom}. Evaluate the normal approximations to power under plausible total effects and patterns of carryover, and prespecify the horizon and statistic for the primary analysis. The reported results should include the numbers of retained sections and focal periods, the treatment--control balance among retained sections, and their conditional treatment probabilities.

\section{Simulation study}\label{sec:simulations}

We examine finite-sample performance in two settings. The first is a stylized potential-outcome model in which the magnitudes and timing of treatment effects are specified directly. It allows us to evaluate rejection rates under the null hypotheses and under alternatives involving total effects, carryover, and anticipation. The second is a dynamic rideshare simulation in which pricing affects rider acceptance, driver availability, and subsequent rewards. This simulation illustrates the tests in an environment where the duration of carryover is not directly observed.

\subsection{Stylized potential-outcome simulations}\label{subsec:stylized-sim}

We first consider a potential-outcome model with additive effects of past, current, and future treatment assignments:
\begin{equation}
\label{eq:dgp_twosided}
Y_t(w_{1:T})
= \mu+\alpha_t+
\sum_{\ell=-a}^{m_0}\delta_\ell w_{t-\ell}
+\sigma_t(w_{1:T})\varepsilon_t,
\qquad t=1,\ldots,T.
\end{equation}
The coefficients at positive lags, $\ell>0$, represent carryover; $\ell=0$ represents the contemporaneous effect; and negative lags, $\ell<0$, represent anticipation. Thus $m_0$ and $a$ specify the maximum numbers of lagged and future assignments in the additive component. Throughout these simulations, $\mu=0$ and $\alpha_t=\log(t)$. We consider two disturbance distributions,
\[
\varepsilon_t\stackrel{i.i.d.}{\sim}N(0,1)
\qquad\text{and}\qquad
\varepsilon_t\stackrel{i.i.d.}{\sim}t_1,
\]
where $t_1$ denotes the Cauchy distribution.  The disturbance scale is allowed
to depend on treatment history:
\[
\sigma_t(w_{1:T})
=
\begin{cases}
0, & t=1,2,\\
\alpha_t\,\mathbf 1\{w_{t-2:t}\in\{\zero_{3},\one_{3}\}\}, & t\ge3.
\end{cases}
\]
The trend, disturbance distributions, and dependence of the disturbance scale on recent assignments are the same in the total-effect, carryover, and anticipation simulations. The simulations differ in the coefficients $\{\delta_\ell\}$ that specify each null and alternative.

We use the regular switchback design of \citet{Bojinov2023}, with treatment probability $1/2$ and two-period blocks except for longer blocks at the beginning and end of the experiment. We vary $T\in\{60,120,180,240,300\}$. The figures report $T/2\in\{30,60,90,120,150\}$ on the horizontal axis. Because of the longer first and last blocks, the actual number of independently randomized blocks is $K_{\mathrm{blk}}=T/2-2\in\{28,58,88,118,148\}$.

All tests use nominal level $0.05$. The total-effect simulations use 1,000 Monte Carlo replications per configuration and 1,000 randomization draws for the CRT, Fisher randomization test (FRT), and PIRT comparisons. The carryover and anticipation simulations also use 1,000 Monte Carlo replications in the baseline specification. The carryover CRT uses 999 randomization draws. For the non-anticipation test, we evaluate the conditional pairwise probability analytically, so there is no additional simulation error from approximating that probability.

\paragraph{Total-effect null.}
For the total-effect experiment, the true carryover horizon is $m=2$ and the
analysis uses the correctly specified horizon.  The null imposes
$\delta_0=\delta_1=\delta_2=0$.  Alternatives set
\[
\delta_0=\delta_1=\delta_2=\delta,
\qquad
\delta\in\{1,2,3\}.
\]
We compare four procedures: the proposed CRT, a Fisher randomization test based on the sharp null of no effect under any assignment (FRT-sharp), the PIRT comparison, and the asymptotic Horvitz--Thompson test with the conservative variance bound of \citet{Bojinov2023}. The PIRT comparison uses the unadjusted pairwise probability with cutoff $0.05$. It is included as a descriptive comparison, not as a procedure with guaranteed size control when the conditions needed for that cutoff are not imposed.

\begin{figure}[t]
\centering
\includegraphics[width=\textwidth]{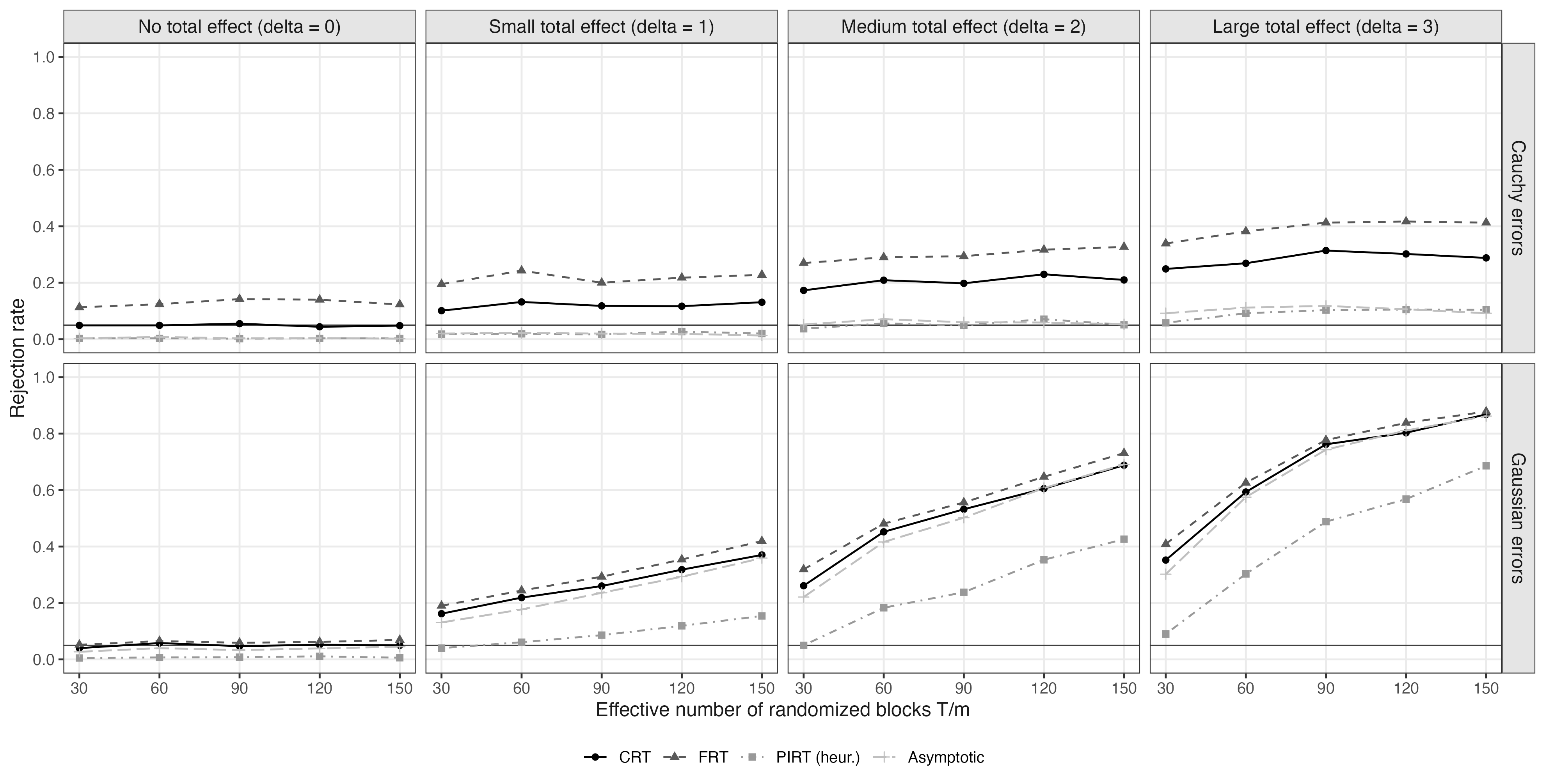}
\caption{Rejection frequencies for the total-effect null and alternatives.}
\label{fig:power-global}
\notes{\textbf{Notes:} The horizontal line marks the nominal level $0.05$. Columns correspond to the coefficient $\delta$ in $\delta_0=\delta_1=\delta_2=\delta$; the first column has $\delta=0$. CRT is the proposed conditional randomization test, FRT is the test based on the sharp no-effect null, and ``PIRT (heur.)'' uses the unadjusted pairwise probability with cutoff $0.05$. The horizontal-axis values are $T/2$; the number of independently randomized blocks is $T/2-2$.}
\end{figure}

Figure~\ref{fig:power-global} shows that the proposed CRT has rejection rates near the nominal level under both Gaussian and Cauchy disturbances. The FRT based on the sharp no-effect null can overreject, particularly under Cauchy disturbances, because the total-effect null does not imply that outcomes are unchanged under every assignment. The asymptotic test can be conservative when the number of randomized blocks is modest. Among the procedures that control size in these simulations, the CRT has competitive power.

\paragraph{Carryover diagnostic.}
We next test the null of no carryover beyond $m=2$,
\[
H_{0}^{\mathrm{carry}}:
\delta_\ell=0
\quad\text{for all }\ell>2.
\]
To study power, we specify nonzero effects at lags beyond the tested horizon. For $\delta\in\{0,1.5,3\}$, we consider four patterns of these coefficients, referred to as the carryover tail:
\[
\begin{array}{ll}
\text{near tail:} &
\delta_3=\delta_4=\delta_5=\delta,\\[0.06cm]
\text{delayed tail:} &
\delta_6=\delta_7=\delta_8=\delta,\\[0.06cm]
\text{decaying tail:} &
(\delta_3,\ldots,\delta_8)
=\delta(1,0.6,0.36,0.216,0.13,0.078),\\[0.06cm]
\text{concentrated tail:} &
\delta_3=3\delta.
\end{array}
\]
All other coefficients at lags greater than $m$ are zero. The case $\delta=0$ satisfies the null, while larger values imply larger carryover effects beyond the tested horizon. These specifications vary both the lags at which carryover occurs and its distribution across those lags.

\begin{figure}[t]
\centering
\includegraphics[width=\textwidth]{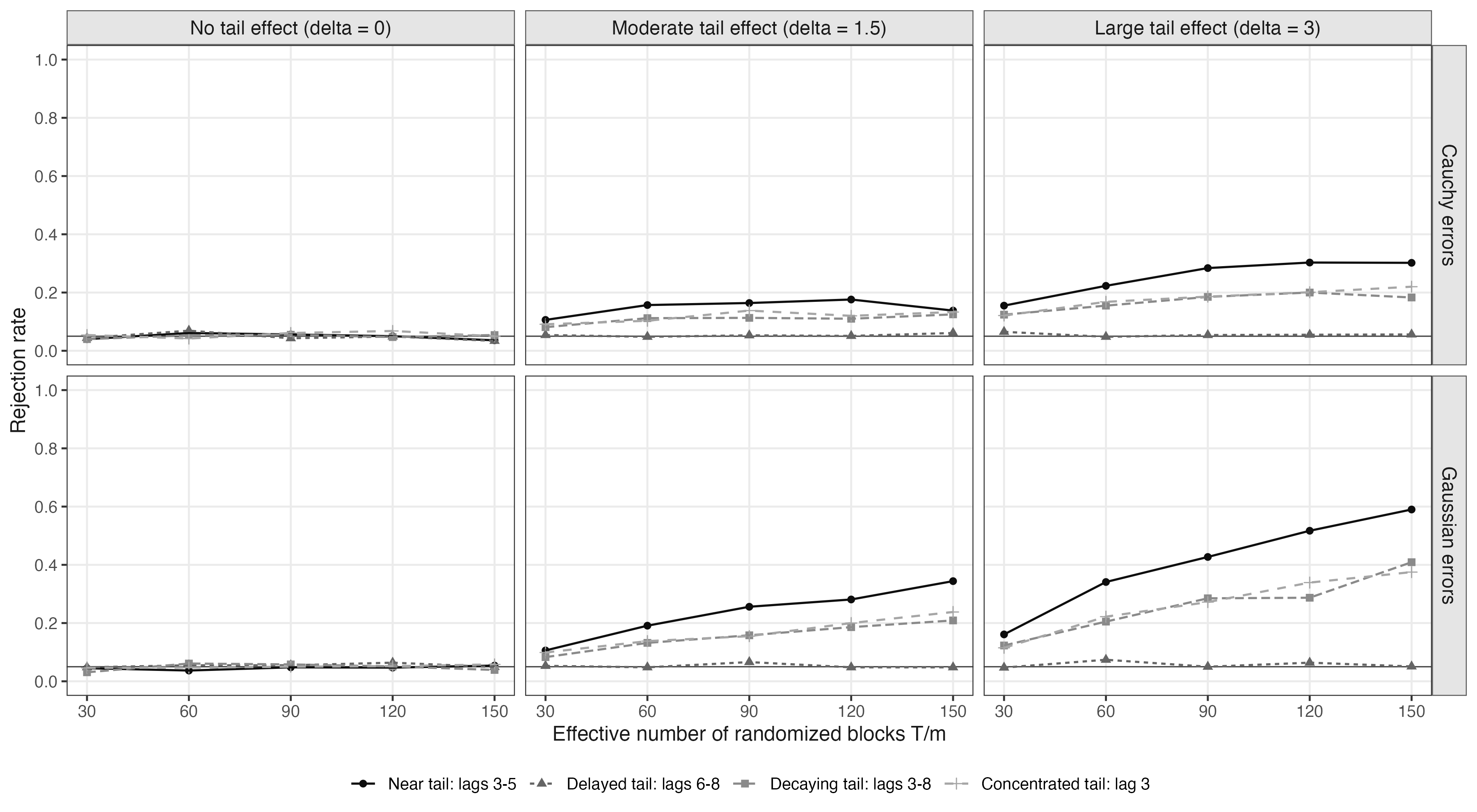}
\caption{Rejection frequencies for carryover diagnostics.}
\label{fig:power-carryover}
\notes{\textbf{Notes:} The null is no carryover beyond $m=2$. The horizontal line marks $0.05$. The curve labels refer to the four patterns of carryover coefficients defined in the text: near, delayed, decaying, and concentrated. The horizontal-axis values are $T/2$; the number of independently randomized blocks is $T/2-2$.}
\end{figure}

Figure~\ref{fig:power-carryover} shows that the carryover CRT controls size when $\delta=0$. Power is highest when nonzero carryover begins just beyond the tested horizon and is lower when it occurs only at more distant lags. The treatment assigned to the final block of the preceding odd-numbered section is most informative about carryover into the early focal periods of the next section. Detecting effects at more distant lags may require a different prespecified statistic.

\paragraph{Non-anticipation diagnostic.}
Finally, we evaluate the conditional PIRT for non-anticipation. The null sets $\delta_\ell=0$ for every $\ell<0$. The alternatives set
\[
\delta_{-2}=\delta_{-1}=\delta_0=\delta,
\qquad
\delta\in\{1,2,3\}.
\]
The one- and two-period leads allow outcomes to depend on assignments shortly in the future, as the comparison in~\eqref{eq:Tna-boundary-identity} is designed to detect. The contemporaneous effect keeps outcome levels comparable across specifications. For a prespecified $h$, we hold all assignments through the end of the block containing $h$ fixed and redraw later assignments. Each curve uses a different value of $h$, chosen before observing outcomes.

\begin{figure}[t]
\centering
\includegraphics[width=\textwidth]{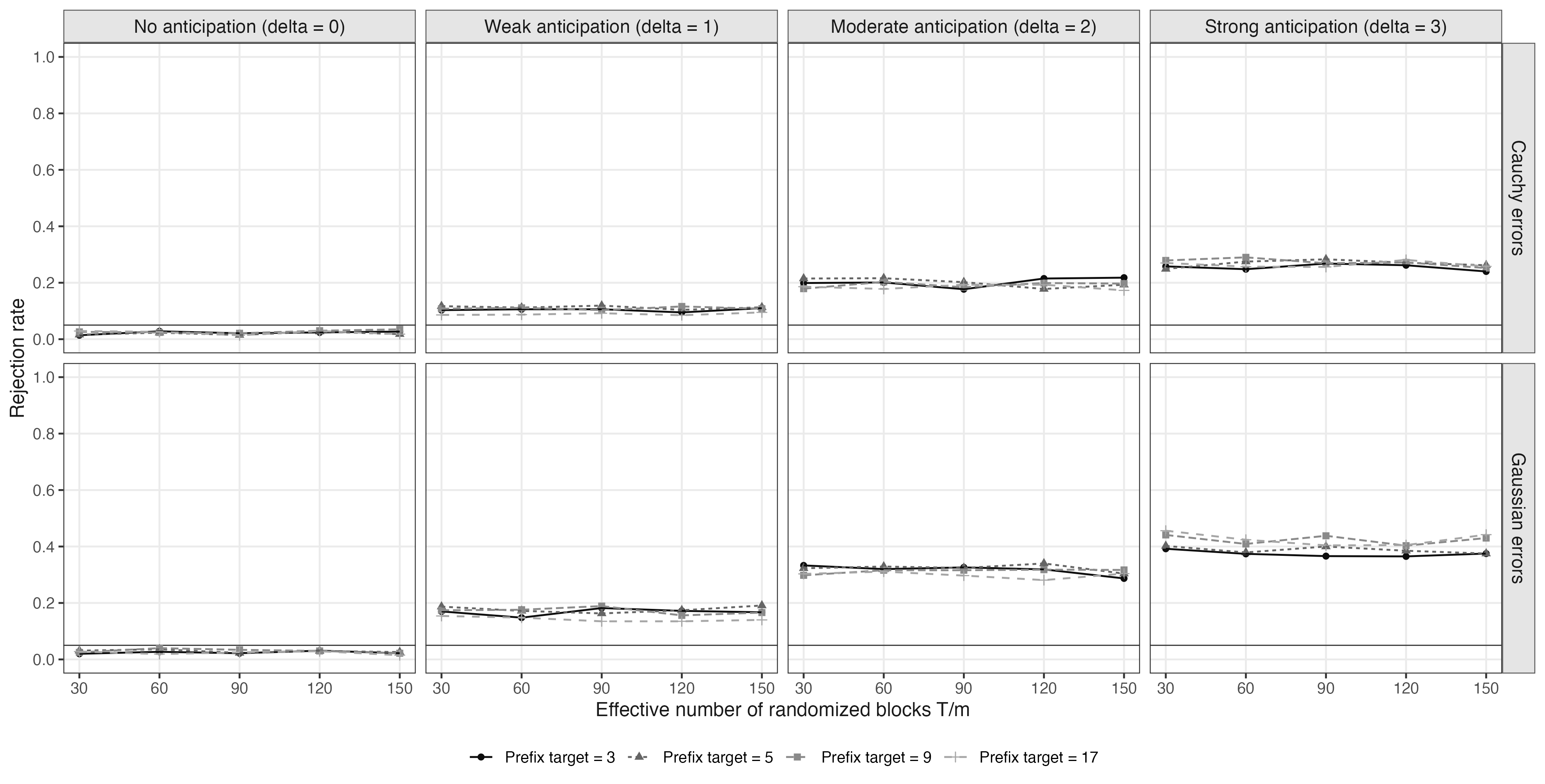}
\caption{Rejection frequencies for non-anticipation diagnostics.}
\label{fig:power-ant}
\notes{\textbf{Notes:} The null rules out effects of future assignments. Each curve uses a separate prespecified value of $h$, called the ``prefix target'' in the legend, and holds assignments fixed through the end of the block containing $h$. The rejection rule is $p_{\mathrm{cPIRT}}=\min\{1,2p_{\mathrm{pair}}\}\le0.05$, with the pairwise probability evaluated analytically. The horizontal-axis values are $T/2$; the number of independently randomized blocks is $T/2-2$.}
\end{figure}

Figure~\ref{fig:power-ant} reports rejection rates for the
conditional non-anticipation test in
Theorem~\ref{thm:valid-pirt-prefix}. Rejection rates are controlled under the null and increase as the anticipation effects grow. Power need not increase or decrease monotonically with $h$: this choice affects both the first period at which the observed and resampled assignments differ and the mean used to center the preceding outcome. The \OA{} also reports the descriptive comparison that rejects when the unadjusted pairwise probability is at most $0.05$.

\subsection{Dynamic rideshare simulation}\label{subsec:rideshare-simulator}

We also study a dynamic rideshare simulation motivated by the marketplace model of \citet{FariasEtAl2022}. Requests arrive sequentially, and completed trips change driver availability and location. A pricing intervention can therefore affect both the reward from the current request and outcomes for later requests. We test $H_0^{\mathrm{tot}}(m)$, the null that each request after the first $m$ has the same potential reward when the comparison pricing rule is used throughout as when the baseline rule is used throughout. To describe the magnitudes of the alternatives, we also compare average rewards in paired simulations that apply each pricing rule throughout. These average differences are effect-size summaries, not the null hypothesis tested by the CRT.

In the placebo specification, the pricing rule is the same under treatment and control. The other two specifications increase the price per unit of distance by $0.005$ and $0.010$, respectively. The paired simulations produce average reward differences of $0$, $0.730$, and $1.292$ per request. The \OA{} reports the parameter values, software environment, and replication details.

Each simulated experiment contains $300{,}000$ requests and a fleet of 300 vehicles. The pricing rule is independently randomized across blocks of length
\[
L\in\{1000,2500,5000,10000\},
\]
giving $300$, $120$, $60$, or $30$ effective randomized blocks. We analyze the candidate horizons
\[
m\in\{250,500,1000,2500,5000\}.
\]
Each configuration uses 50 independently simulated experiments, and each randomization $p$-value uses 999 reference draws.

The analysis with a prespecified horizon corresponds to the test in Section~\ref{sec:main}. Table~\ref{tab:rideshare-placebo-fixed} reports placebo rejection rates at the fixed horizon $m=250$ in a larger experiment with repeated randomization. The rates are close to $0.05$ across the four block counts. For comparison, we also report an exploratory procedure that applies the sequential carryover tests and then tests the total effect at the first nonrejected horizon, denoted by $\widehat m$. Theorems~\ref{thm:valid-crt} and~\ref{thm:sequential-carryover} do not establish size control for this combined procedure. We report it to illustrate what happens when horizon selection and effect testing use the same outcomes, not as a theoretically justified replacement for the test at a prespecified horizon.

\begin{table}[t]
\centering
\caption{Placebo rejection rates using a fixed or data-selected carryover horizon}
\label{tab:rideshare-placebo-fixed}
\begin{threeparttable}
\begin{tabular}{rrrrr}
\toprule
Randomized blocks & Fixed $m=250$ & Selected $m$ & Mean selected $m$ & $\pr(\widehat m>250)$\\
\midrule
300 & 0.040 & 0.041 & 282.8 & 0.039\\
120 & 0.060 & 0.059 & 378.3 & 0.057\\
 60 & 0.052 & 0.051 & 419.0 & 0.046\\
 30 & 0.043 & 0.041 & 367.0 & 0.036\\
\bottomrule
\end{tabular}
\begin{tablenotes}[flushleft]
\footnotesize
\item The table uses 1,000 randomized replications for each block count and 499 reference draws per $p$-value. The second and third columns report rejection rates for the total-effect test. ``Selected $m$'' uses the first nonrejected horizon from the sequential carryover tests applied to the same outcomes. This data-selected procedure is exploratory and is not covered by the finite-sample validity results.
\end{tablenotes}
\end{threeparttable}
\end{table}

Figure~\ref{fig:rideshare-total-tests} reports the results of this exploratory procedure in the experiment with 50 independent simulations per configuration. For each simulation, we test the carryover nulls over the prespecified grid and then test the total effect at the first nonrejected horizon. If all candidate horizons are rejected, the implementation uses the largest value in the grid. Placebo rejection rates remain close to $0.05$. For the two nonzero price changes, the proposed CRT generally rejects more often than the FRT based on the sharp no-effect null. 

\begin{figure}[t]
\centering
\includegraphics[width=\textwidth]{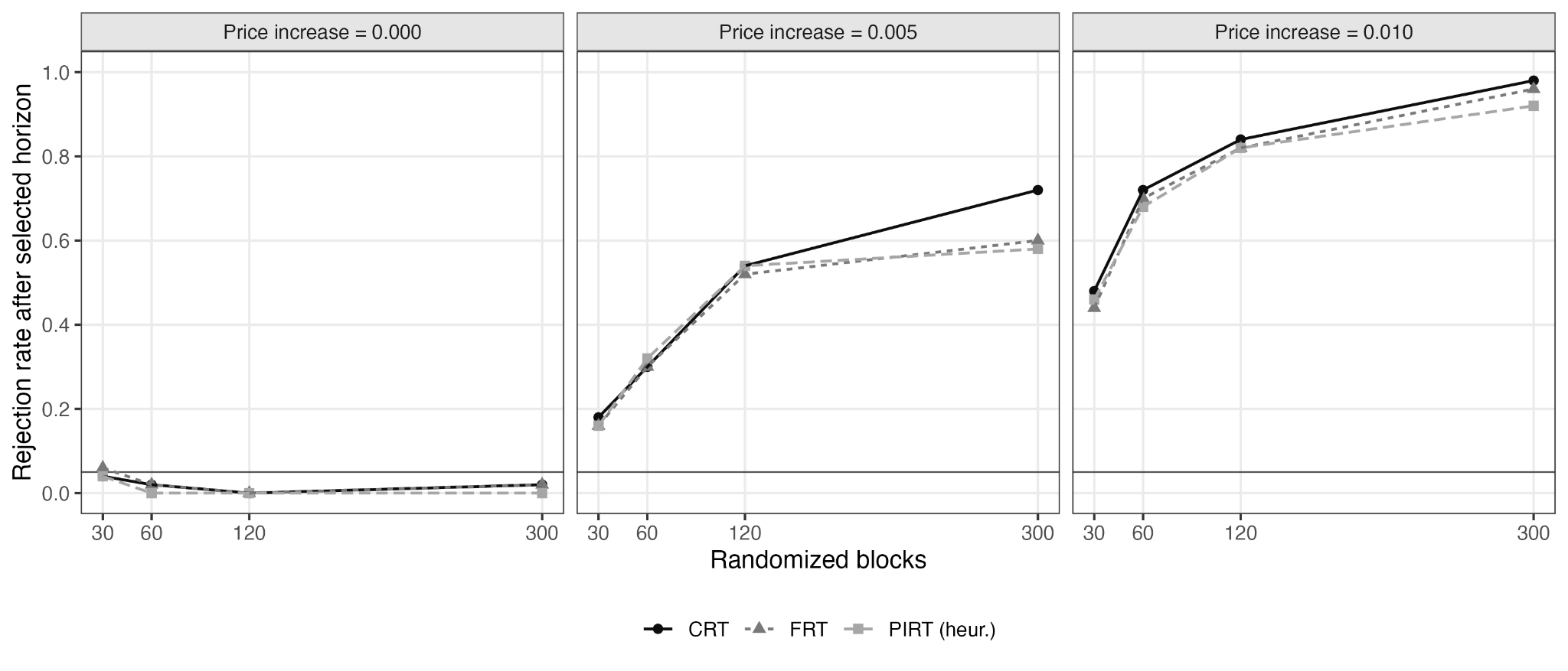}
\caption{Rejection frequencies when the carryover horizon is selected from the data.}
\label{fig:rideshare-total-tests}
\notes{\textbf{Notes:} Panels correspond to price increases relative to the baseline, and the horizontal axis reports the number of randomized blocks. The horizontal line marks $0.05$. The PIRT curve compares the unadjusted pairwise probability with $0.05$ and is included only as a descriptive comparison. Each point uses 50 independently simulated experiments. The carryover horizon is selected using the same outcomes as the total-effect test; the combined procedure is exploratory.}
\end{figure}

The number of possible reassignments is particularly restrictive when many short blocks must be combined. With 300 blocks of length 1,000 and $m=5000$, each section must contain six blocks. This gives 50 predetermined sections, each retained with probability $1/32$, and an expected number of retained sections of $1.56$. The probability that no section is retained is $0.204$. The probability of retaining fewer than the five sections needed for an exact one-sided test to reject at the $5\%$ level is $0.980$. In the experiment with 1,000 randomized replications, 227 have no focal outcomes available for the CRT at this horizon. The exploratory selection rule chooses this horizon infrequently in this design, but the calculation illustrates the tradeoff: combining blocks provides longer intervals with the same treatment while reducing the number of randomized comparisons.

\section{Conclusion}\label{sec:conclusion}

Switchback experiments can produce many observations but relatively few independent treatment assignments. When treatment has carryover effects, the null of no total effect does not determine outcomes under all combinations of past treatment and control. We construct a conditional randomization test that uses only outcomes for which the null determines the required counterfactuals. The test reassigns predetermined sections that received the same treatment throughout, while holding all other assignments fixed. Under non-anticipation and a correctly specified carryover horizon, it controls size in finite samples without distributional assumptions on outcomes. We also develop a carryover test under non-anticipation and a separate pairwise randomization test for anticipation.

The design calculations show why the number of observations alone is not enough to assess power. Block length, section boundaries, and the number of excluded initial periods determine which outcomes and treatment assignments remain available for analysis. Power also depends on treatment--control balance and the serial dependence of outcomes within sections. Calculations of attainable $p$-values complement the power approximations when the number of retained sections is small. For hypotheses about average treatment effects, studentized tests are asymptotically valid in the equal-length-block setting with $m<L$ and no assignment-dependent selection. Together, the results provide methods for testing total and average effects, assessing assumptions about temporal interference, and choosing an experimental design with enough randomized comparisons for useful inference.

\begingroup
\setstretch{1.1}
\interlinepenalty=10000
\printbibliography
\endgroup

\clearpage
\end{refsection}

\makeatletter
\def\@ptsize{0}
\input{size10.clo}
\makeatother
\newgeometry{margin=1in}
\setlength{\headheight}{12pt}
\setlength{\headsep}{25pt}
\setlength{\footskip}{30pt}
\setlength{\topskip}{10pt}
\setlength{\maxdepth}{5pt}
\makeatletter
\global\@maxdepth=5pt
\global\maxdepth=5pt
\makeatother
\setlength{\marginparsep}{11pt}
\setstretch{1.10}
\normalsize

\setlength{\leftmargini}{2.5em}
\setlength{\leftmarginii}{2.2em}
\setlength{\leftmarginiii}{1.87em}
\setlength{\leftmarginiv}{1.7em}
\setlength{\leftmarginv}{1em}
\setlength{\leftmarginvi}{1em}
\setlength{\labelsep}{.5em}
\setlength{\tabbingsep}{\labelsep}
\setlength{\heavyrulewidth}{.08em}
\setlength{\lightrulewidth}{.05em}
\setlength{\cmidrulewidth}{.03em}
\setlength{\belowrulesep}{.65ex}
\setlength{\aboverulesep}{.4ex}
\setlength{\cmidrulekern}{.5em}
\setlength{\defaultaddspace}{.5em}
\setlength{\bibhang}{15pt}
\setlength{\biblabelsep}{10pt}
\setlength{\bibitemsep}{0pt}
\setlist[itemize]{leftmargin=2em,itemsep=0.15em,topsep=0.2em}
\setlist[enumerate]{leftmargin=2em,itemsep=0.15em,topsep=0.2em}
\setlength{\emergencystretch}{2em}
\renewcommand{\Var}{\operatorname*{Var}}
\renewcommand*{\proofname}{Proof}
\renewcommand{\qedsymbol}{\rule{2mm}{2mm}}
\pagestyle{fancy}
\fancyhf{}
\fancyhead[L]{\small Online Appendix: Randomization Tests in Switchback Experiments}
\fancyhead[R]{\small \thepage}
\renewcommand{\headrulewidth}{0.2pt}

\setcounter{page}{1}
\setcounter{section}{0}
\setcounter{subsection}{0}
\setcounter{subsubsection}{0}
\setcounter{equation}{0}
\setcounter{figure}{0}
\setcounter{table}{0}
\setcounter{footnote}{0}
\makeatletter
\let\maketitle\MSmaketitle
\let\@maketitle\MSinternalmaketitle
\let\title\MStitle
\let\author\MSauthor
\let\date\MSdate
\let\thanks\MSthanks
\let\and\MSand
\gdef\@thanks{}
\makeatother
\begin{refsection}

\makeatletter
\@for\MSlabel:=ass:m-carryover,ass:no-anticipation,ass:weak-null-asymp,def:imput_units,def:local-po,def:m-null,def:regular-switchback,eq:Tna-boundary-identity,eq:Tna-centered,eq:Vup-vector,eq:W0,eq:ar1-main,eq:carry-randomization,eq:carry-stat,eq:carry-studentized-power,eq:carry-variance-limits,eq:conditioning-event-total,eq:delta-tail-main,eq:design-constq,eq:dgp-main,eq:dgp_twosided,eq:finite-support-probabilities,eq:ht-general-vector,eq:imputable-times,eq:local-po,eq:min-retained-alpha,eq:mu-sigma-carry,eq:mu-sigma-total,eq:n-def,eq:null-mcarry,eq:null-na,eq:observed-assignment-atom,eq:p-rq,eq:pirt-na-pvalue,eq:power-carry,eq:power-stat-total,eq:power-total,eq:presections,eq:section-ht,eq:section-prob,eq:selected-sections,eq:selected-union,eq:sharp-null-total,eq:var-bar-eps,eq:weak-effect-profile,eq:weak-null-focal,eq:weak-null-general,eq:weak-null-session,eq:weak-pvalue,eq:weak-quadratic,eq:weak-signed,eqn:cond-space,eqn:focal-units,eqn:focal-units-m,fig:conditioning-schematic,fig:power-ant,fig:power-carryover,fig:power-global,fig:rideshare-total-tests,lem:conditional-section-law,proc:section-crt,proc:sequential-carryover,prop:nested-carryover,prop:power-carry,prop:power-total,rem:centering-power,sec:conclusion,sec:diagnostics,sec:intro,sec:main,sec:power,sec:setup,sec:simulations,sec:weak-null-test,subsec:carryover,subsec:crt-overview,subsec:design,subsec:nonanticipation,subsec:nulls,subsec:rideshare-simulator,subsec:section-law,subsec:sketch,subsec:stylized-sim,subsec:temporal-interference,tab:rideshare-placebo-fixed,thm:sequential-carryover,thm:valid-carryover,thm:valid-crt,thm:valid-pirt-prefix,thm:weak-null-general\do{%
  \expandafter\def\csname MSmainlabel@\MSlabel\endcsname{}%
}
\NewCommandCopy\MSoriginalref\ref
\RenewDocumentCommand{\ref}{s m}{%
  \IfBooleanTF{#1}{\MSoriginalref*{#2}}{%
    \ifcsname MSmainlabel@#2\endcsname
      \begingroup\let\@linkcolor\@filecolor\MSoriginalref{#2}\endgroup
    \else
      \MSoriginalref{#2}%
    \fi
  }%
}
\makeatother

\sloppy
\title{Online Appendix to ``Randomization Tests in Switchback Experiments''}
\ifsubmissionanonymous
\author{}\date{}
\else
\author{Jizhou Liu \and Liang Zhong}\date{September 20, 2026}
\fi
\maketitle

\section{Proof of Theorem~\ref{thm:valid-crt}}\label{app:proof-valid-crt}

If $\mathcal U^{\mathrm{obs}}=\varnothing$, Procedure~\ref{proc:section-crt} reports $p=1$ and the result is immediate. Assume henceforth that the focal set is nonempty.\par

Let $\mathcal S^{\mathrm{obs}}=\mathcal S(\mathbf W^{\mathrm{obs}})$ and
\[
\mathcal A^{\mathrm{obs}}
=
\bigcup_{[s,e]\in\mathcal S^{\mathrm{obs}}}\{s,s+1,\ldots,e\}.
\]
Let $\mathcal C^{\mathrm{obs}}$ be the conditioning event in~\eqref{eq:conditioning-event-total}. Its support is exactly the assignment space in~\eqref{eqn:cond-space}. Because the predetermined sections are disjoint, assignments on every realized nonselected section are fixed, while each selected section remains constant. Consequently, for every $\mathbf w$ in the conditional support,
\[
\mathcal S(\mathbf w)=\mathcal S^{\mathrm{obs}},
\qquad
\mathcal U(\mathbf w)=\mathcal U(\mathbf W^{\mathrm{obs}})=:\mathcal U^{\mathrm{obs}}.
\]

Fix $t\in\mathcal U^{\mathrm{obs}}$. There is a selected section $[s,e]$ with $s+m\le t\le e$. Both $\mathbf W^{\mathrm{obs}}$ and every $\mathbf w$ in the conditional support therefore generate a constant local history at $t$:
\[
\mathbf W^{\mathrm{obs}}_{t-m:t}\in\{\mathbf 0_{m+1},\mathbf 1_{m+1}\},
\qquad
\mathbf w_{t-m:t}\in\{\mathbf 0_{m+1},\mathbf 1_{m+1}\}.
\]
The two histories need not have the same label. Assumptions~\ref{ass:no-anticipation} and~\ref{ass:m-carryover} together reduce each full-path outcome to its local potential outcome, and the partially sharp null~\eqref{eq:sharp-null-total} equates the two constant local histories. Hence
\[
Y_t(\mathbf w)=Y_t(\mathbf W^{\mathrm{obs}})=Y_t^{\mathrm{obs}}
\quad\text{for every }t\in\mathcal U^{\mathrm{obs}}
\text{ and every }\mathbf w\in\operatorname{supp}(\mathbf W\mid\mathcal C^{\mathrm{obs}}).
\]
Thus all focal outcomes used by the statistic are imputable throughout the conditional support. The conditioning event is also invariant on that support, so the same focal set is used for the observed and resampled assignments.

By Lemma~\ref{lem:conditional-section-law}, Procedure~\ref{proc:section-crt} draws from the exact law of $\mathbf W$ given $\mathcal C^{\mathrm{obs}}$. Conditional on this event, the observed assignment and the Monte Carlo assignments are exchangeable under the null. The randomization rank is therefore super-uniform; the plus-one Monte Carlo construction gives
\[
\mathbb P\!\left(\widehat p\le\alpha\mid\mathcal C^{\mathrm{obs}}\right)\le\alpha.
\]
Averaging over conditioning events proves marginal validity. This is the standard imputability argument for conditional randomization tests \citep{Basse2019}. \qed

\subsection{Conditional assignment law on realized constant merged sections}
\label{app:cond-law-merged}

This subsection verifies that Procedure~\ref{proc:section-crt} samples from the correct
conditional assignment distribution under the regular switchback design, conditional on
the realized section structure used for the test.

\begin{proof}[Proof of Lemma~\ref{lem:conditional-section-law}]
Fix a merged section that covers block indices $a,\ldots,b$.
By assumption, the block assignments are independent with
$W^{(k)}\sim\mathrm{Bernoulli}(q_k)$.

Define the event that the section is constant at the block level:
\[
E_{a:b}:=\{W^{(a)}=\cdots=W^{(b)}\}.
\]
On $E_{a:b}$ there are exactly two possibilities:
either all blocks equal $1$ or all blocks equal $0$. Let
\[
A_{a:b}:=\{W^{(a)}=\cdots=W^{(b)}=1\},
\qquad
B_{a:b}:=\{W^{(a)}=\cdots=W^{(b)}=0\}.
\]
Then $E_{a:b}=A_{a:b}\,\dot\cup\,B_{a:b}$ is a disjoint union. By independence,
\[
\mathbb{P}(A_{a:b})=\prod_{k=a}^{b} q_k,
\qquad
\mathbb{P}(B_{a:b})=\prod_{k=a}^{b} (1-q_k),
\]
and therefore
\[
\mathbb{P}(E_{a:b})
=
\prod_{k=a}^{b} q_k
+
\prod_{k=a}^{b} (1-q_k).
\]
Let $Z$ denote the common value on the section under $E_{a:b}$, i.e.\ $Z:=W^{(a)}$ on $E_{a:b}$.
Then
\[
\mathbb{P}(Z=1\mid E_{a:b})
=
\frac{\mathbb{P}(A_{a:b})}{\mathbb{P}(E_{a:b})}
=
\frac{\prod_{k=a}^{b} q_k}{\prod_{k=a}^{b} q_k+\prod_{k=a}^{b}(1-q_k)}.
\]
Hence, conditional on $E_{a:b}$, the section’s common assignment is Bernoulli with probability
\[
p=\frac{\prod_{k=a}^{b} q_k}{\prod_{k=a}^{b} q_k+\prod_{k=a}^{b}(1-q_k)},
\]
and the entire block vector $(W^{(a)},\ldots,W^{(b)})$ equals $(Z,\ldots,Z)$.

We now prove the claimed independence across disjoint sections.
Let sections $j=1,\ldots,J$ be disjoint at the block level, where section $j$ covers blocks
$a_j,\ldots,b_j$, and define
\[
E_j:=\{W^{(a_j)}=\cdots=W^{(b_j)}\},
\qquad
Z_j:=W^{(a_j)} \ \text{on } E_j.
\]
Because the underlying blocks are independent and the collections of indices
$\{a_j,\ldots,b_j\}$ are disjoint across $j$, the random vectors
$\{W^{(a_j)},\ldots,W^{(b_j)}\}$ are mutually independent, and the events $\{E_j\}$ are also mutually independent.

For any $(z_1,\ldots,z_J)\in\{0,1\}^J$, define the event
\[
F(z_1,\ldots,z_J)
:=
\bigcap_{j=1}^J
\left\{
W^{(a_j)}=\cdots=W^{(b_j)}=z_j
\right\}.
\]
The events $\{F(z_1,\ldots,z_J)\}$ are disjoint and partition $\bigcap_{j=1}^J E_j$.
Using independence across disjoint sections,
\[
\mathbb{P}\!\left(F(z_1,\ldots,z_J)\right)
=
\prod_{j=1}^J
\left[
\left(\prod_{k=a_j}^{b_j} q_k\right)^{z_j}
\left(\prod_{k=a_j}^{b_j} (1-q_k)\right)^{1-z_j}
\right].
\]
Moreover,
\[
\mathbb{P}\!\left(\bigcap_{j=1}^J E_j\right)
=
\prod_{j=1}^J
\left(
\prod_{k=a_j}^{b_j} q_k
+
\prod_{k=a_j}^{b_j} (1-q_k)
\right).
\]
Therefore,
\begin{align*}
\mathbb{P}(Z_1=z_1,\ldots,Z_J=z_J \mid \cap_{j=1}^J E_j)
&=
\frac{\mathbb{P}(F(z_1,\ldots,z_J))}
{\mathbb{P}(\cap_{j=1}^J E_j)} \\
&=
\prod_{j=1}^J
\frac{
\left(\prod_{k=a_j}^{b_j} q_k\right)^{z_j}
\left(\prod_{k=a_j}^{b_j} (1-q_k)\right)^{1-z_j}
}{
\left(\prod_{k=a_j}^{b_j} q_k\right)
+
\left(\prod_{k=a_j}^{b_j} (1-q_k)\right)
}.
\end{align*}
This factorization shows that, conditional on $\cap_{j=1}^J E_j$, the common values
$Z_1,\ldots,Z_J$ are independent, with $Z_j\sim\mathrm{Bernoulli}(p_j)$ where
\[
p_j=
\frac{\prod_{k=a_j}^{b_j} q_k}{\prod_{k=a_j}^{b_j} q_k+\prod_{k=a_j}^{b_j}(1-q_k)}.
\]
Finally, the block labels outside the union of these sections are
independent of the block vectors entering $\{E_j,Z_j\}_{j=1}^J$,
because the corresponding sets of design blocks are disjoint.
Consequently, conditioning additionally on the realized values of all
outside block labels leaves the conditional factorization above
unchanged. This is precisely the additional conditioning imposed by
$\mathcal C^{\mathrm{obs}}$.
This completes the proof.
\end{proof}

\paragraph{Connection to Procedure~\ref{proc:section-crt}.}
In Procedure~\ref{proc:section-crt}, each realized section $[s_j,e_j]$ corresponds to merging blocks
$a_j,\ldots,b_j$ and imposing the constraint $W^{(a_j)}=\cdots=W^{(b_j)}$.
The lemma above shows that, conditional on this constraint, the section’s common value is
$\mathrm{Bernoulli}(p_j)$ with $p_j$ as computed in the algorithm, and that these section-level
values are independent across disjoint sections. Since periods outside the union of realized
sections are held fixed to $W^{\mathrm{obs}}$ in the conditional randomization space
\eqref{eqn:cond-space}, Procedure~\ref{proc:section-crt} samples exactly from the conditional
assignment distribution used by the conditional randomization test.

\section{Proofs for the carryover diagnostic}\label{app:proof-carryover}

\subsection{Fixed-horizon validity}

\begin{proof}[Local-dependence argument for the fixed-horizon diagnostic]
Fix an even section $[s_{2i},e_{2i}]$ and $t\in\mathcal U_{2i}^{(m)}$. Then $t-m\ge s_{2i}$. Consider two full assignment paths that agree from $s_{2i}$ through $T$ but may differ before $s_{2i}$. They agree from $t-m$ through $T$, so $H_0^m$ in~\eqref{eq:null-mcarry} equates their outcomes at $t$. This proves the claimed invariance.
\end{proof}

\begin{proof}[Proof of Theorem~\ref{thm:valid-carryover}]
If no odd--even pair has a nonempty focal window, the diagnostic returns $p=1$ and validity is immediate. Otherwise proceed with the nonempty pairs.
Condition on the complete assignment paths in all even sections and on every odd-section block label except the final block label used by~\eqref{eq:carry-stat}. Because predetermined sections are unions of complete design blocks, the remaining labels are distinct block assignments. Under Definition~\ref{def:regular-switchback}, they are conditionally independent Bernoulli variables with their original probabilities.

Fix a focal time $t$ in even section $2i$ and a reference assignment in the conditional support. Replace the reference assignment after time $t$ by the observed future path. Non-anticipation leaves the reference outcome at $t$ unchanged. The resulting hybrid path and the observed path agree from $t-m$ through $T$: assignments within the even section are fixed, $t-m\ge s_{2i}$, and the hybrid uses the observed future. The null $H_0^m$ therefore equates the hybrid and observed outcomes. Thus every focal mean is imputable even though a joint reference draw may change odd-section labels that occur in the future of earlier focal outcomes.

The randomized labels have exactly the conditional law in~\eqref{eq:carry-randomization}. The usual conditional randomization-rank argument, with the plus-one correction under Monte Carlo sampling, gives a conditional rejection probability no larger than $\alpha$. Averaging over the conditioned assignments proves finite-sample validity.
\end{proof}

\subsection{Nestedness and sequential testing}

\begin{proof}[Proof of Proposition~\ref{prop:nested-carryover}]
Fix $m\ge\bar m$, $t\ge m+1$, and two paths $\mathbf w,\mathbf w'$ that agree from $t-m$ through $T$. Since $m\ge\bar m$, they also agree from $t-\bar m$ through $T$. The null $H_0^{\bar m}$ therefore gives $Y_t(\mathbf w)=Y_t(\mathbf w')$, which is exactly the defining implication for $H_0^m$.
\end{proof}

\begin{proof}[Proof of Theorem~\ref{thm:sequential-carryover}]
If no null in the prespecified grid $\mathcal M$ is true, the FWER is zero by definition. Otherwise, let $\bar m$ be the smallest true grid value. By Proposition~\ref{prop:nested-carryover}, every tested null with index below $\bar m$ is false and every tested null with index at least $\bar m$ is true.

A family-wise error occurs only if the ordered procedure reaches and rejects $H_0^{\bar m}$. Hence
\[
\mathrm{FWER}
=
\Pr(\text{reach and reject }H_0^{\bar m})
\le
\Pr(\varphi_{\bar m}=1)
\le \alpha,
\]
where $\varphi_{\bar m}$ is the rejection indicator from the fixed-$\bar m$ test whether or not the sequential procedure reaches that test. The final inequality follows from its level-$\alpha$ validity.
\end{proof}

\section{Conditional pairwise imputation-based randomization tests}
\label{app:proof-pirt}

The non-anticipation test in Section~\ref{subsec:nonanticipation}
of the main paper applies a conditional version of the pairwise
imputation-based randomization test of
\citet{zhong2024unconditional}. The construction restricts
comparisons to assignments belonging to the same set in a
prespecified partition of the assignment support.

\subsection{Conditioning sets and pairwise imputation}

Let $\Gamma:\mathcal W_0\to\mathcal G$ be a function specified
before assignments or outcomes are observed. It partitions the
assignment support into sets
\[
S_g
=
\{\w\in\mathcal W_0:\Gamma(\w)=g\},
\qquad g\in\Gamma(\mathcal W_0).
\]
Alternatively, let $S(\w)$ denote the reference set associated
with assignment $\w$, and suppose
\begin{equation}\label{eq:pirt-cell-consistency}
\w\in S(\w),
\qquad
\w'\in S(\w)
\ \Longrightarrow\
S(\w')=S(\w).
\end{equation}
Then the distinct values of $S(\w)$ form a partition of
$\mathcal W_0$. Every assignment belongs to its own selected
set. If two selected sets intersect, take
$\mathbf v\in S(\w)\cap S(\w')$.
Condition~\eqref{eq:pirt-cell-consistency} gives
$S(\mathbf v)=S(\w)=S(\w')$, so the two sets are identical.
Conversely, any fixed partition can be represented by taking
$S(\w)$ to be the set containing $\w$. Thus one may define
$\Gamma(\w)$ as the label of $S(\w)$.

After observing assignment $\w$, the analysis conditions on
\begin{equation}\label{eq:pirt-observed-stratum}
\mathcal C_\Gamma(\w)
=
\{\Gamma(\W)=\Gamma(\w)\}.
\end{equation}
The rule $S(\cdot)$, or equivalently $\Gamma$, is fixed before
randomization; the realized assignment selects the set on which
the analysis conditions. The realized set itself need not be
known before assignment. For every $g$ with
$\pr\{\Gamma(\W)=g\}>0$, let $P_g$ denote the original assignment
distribution conditional on $\Gamma(\W)=g$.

Write $\Y(\w)=(Y_1(\w),\ldots,Y_T(\w))^\top$ for the vector
of potential outcomes under $\w$. We require the pairwise
statistic $T(\Y,\w,\w')$ to be imputable under the tested null
for assignments in the same conditioning set:
\begin{equation}\label{eq:pairwise-null-imputability}
T\{\Y(\w),\w,\w'\}
=
T\{\Y(\w'),\w,\w'\}
\quad\text{whenever }\Gamma(\w)=\Gamma(\w').
\end{equation}
For an observed assignment $\w$ and a reference draw
$\W^\ast\sim P_{\Gamma(\w)}$, define
\begin{equation}\label{eq:conditional-pirt-general}
p_{\mathrm{pair}}^\Gamma(\w)
=
P_{\Gamma(\w)}\!\left[
T\{\Y(\w),\W^\ast,\w\}
\ge
T\{\Y(\w),\w,\W^\ast\}
\right].
\end{equation}
The calibrated $p$-value is
\begin{equation}\label{eq:conditional-pirt-calibrated}
p_{\mathrm{cPIRT}}^\Gamma(\w)
=
\min\{1,2p_{\mathrm{pair}}^\Gamma(\w)\}.
\end{equation}
The factor of two calibrates the pairwise probability; it is
not an adjustment specific to a two-sided alternative. A two-sided
test replaces $T$ with $|T|$ throughout.

\subsection{Finite-sample validity}

\begin{theorem}[Finite-sample validity of the conditional PIRT]
\label{thm:valid-pirt}
Suppose $\Gamma$ is fixed before assignments or outcomes are
observed, and $T$ satisfies
\eqref{eq:pairwise-null-imputability}. Under the tested null,
for every $g$ with $\pr\{\Gamma(\W)=g\}>0$ and every
$\alpha\in(0,1)$,
\[
\pr\!\left\{
p_{\mathrm{cPIRT}}^\Gamma(\W)\le\alpha
\,\middle|\,
\Gamma(\W)=g
\right\}
\le\alpha.
\]
Consequently,
$\pr\{p_{\mathrm{cPIRT}}^\Gamma(\W)\le\alpha\}\le\alpha$.

For the Monte Carlo implementation, set $\W^{(0)}=\Wobs$.
Conditional on $\Gamma(\Wobs)=g$, draw
$\W^{(1)},\ldots,\W^{(B)}$ independently from $P_g$,
independently of $\W^{(0)}$, and define
\begin{equation}\label{eq:conditional-pirt-mc}
\widehat p_{\mathrm{pair}}^\Gamma
=
\frac{
1+\sum_{b=1}^B
\ind{
T(\Yobs,\W^{(b)},\Wobs)
\ge
T(\Yobs,\Wobs,\W^{(b)})
}
}{B+1}.
\end{equation}
With probability taken jointly over the original assignment
and the Monte Carlo draws,
\[
\pr\!\left\{
\min\{1,2\widehat p_{\mathrm{pair}}^\Gamma\}\le\alpha
\,\middle|\,
\Gamma(\Wobs)=g
\right\}
\le\alpha.
\]
The same size bound holds unconditionally.
\end{theorem}

\begin{proof}
Fix $g\in\Gamma(\mathcal W_0)$ with
$\pr\{\Gamma(\W)=g\}>0$. Conditional on
$\Gamma(\Wobs)=g$, the observed assignment has law $P_g$.
For $\w,\w'\in S_g$, define
\[
H_g(\w,\w')
=
\ind{
T\{\Y(\w),\w',\w\}
\ge
T\{\Y(\w),\w,\w'\}
}.
\]
Pairwise null-imputability in
\eqref{eq:pairwise-null-imputability}, applied to the ordered
pairs $(\w',\w)$ and $(\w,\w')$, gives
\begin{align*}
T\{\Y(\w),\w',\w\}
&=T\{\Y(\w'),\w',\w\},\\
T\{\Y(\w),\w,\w'\}
&=T\{\Y(\w'),\w,\w'\}.
\end{align*}
The two indicators obtained by reversing the assignments
therefore compare the same two quantities in opposite directions.
Hence
\begin{equation}\label{eq:pirt-pair-symmetry}
H_g(\w,\w')+H_g(\w',\w)\ge1.
\end{equation}
Within $S_g$, the exact pairwise probability is
\[
p_{\mathrm{pair}}^\Gamma(\w)
=
\sum_{\w'\in S_g}H_g(\w,\w')P_g(\w').
\]
Define
\[
D_g
=
\{\w\in S_g:p_{\mathrm{pair}}^\Gamma(\w)\le\alpha/2\},
\qquad
r_g=P_g(D_g).
\]
If $r_g=0$, the desired inequality is immediate. Otherwise,
\begin{align*}
\frac{\alpha r_g}{2}
&\ge
\sum_{\w\in D_g}
p_{\mathrm{pair}}^\Gamma(\w)P_g(\w)\\
&\ge
\sum_{\w\in D_g}\sum_{\w'\in D_g}
H_g(\w,\w')P_g(\w)P_g(\w')\\
&=
\frac12\sum_{\w\in D_g}\sum_{\w'\in D_g}
\{H_g(\w,\w')+H_g(\w',\w)\}
P_g(\w)P_g(\w')\\
&\ge\frac{r_g^2}{2},
\end{align*}
where the final inequality uses
\eqref{eq:pirt-pair-symmetry}. Therefore $r_g\le\alpha$, which
is equivalent to
\[
\pr\!\left\{
p_{\mathrm{cPIRT}}^\Gamma(\W)\le\alpha
\,\middle|\,
\Gamma(\W)=g
\right\}
\le\alpha.
\]
Averaging over the conditioning sets proves marginal validity.
This is the pairwise-symmetry argument of
\citet{zhong2024unconditional} applied within the set selected
by the observed assignment.

For the Monte Carlo result, set $\W^{(0)}=\Wobs$ and condition
on $\Gamma(\Wobs)=g$. Draw
$\W^{(1)},\ldots,\W^{(B)}$ independently from $P_g$.
Conditional on $g$, the $B+1$ assignments are i.i.d. and therefore
exchangeable. For $i=0,\ldots,B$, define
\[
\widehat p_i
=
\frac{1}{B+1}\sum_{b=0}^B
H_g(\W^{(i)},\W^{(b)}).
\]
Because $H_g(\w,\w)=1$,
\[
\widehat p_0
=
\frac{1+\sum_{b=1}^B H_g(\Wobs,\W^{(b)})}{B+1},
\]
which is the plus-one estimate in
\eqref{eq:conditional-pirt-mc}. Apply the preceding argument
to the empirical distribution placing mass $(B+1)^{-1}$ on each
sampled index. Deterministically, no more than an $\alpha$
fraction of the indices can satisfy $\widehat p_i\le\alpha/2$.
Exchangeability gives
\[
\pr\!\left\{
\min\{1,2\widehat p_0\}\le\alpha
\,\middle|\,
\Gamma(\Wobs)=g
\right\}
\le\alpha.
\]
Averaging over $g$ completes the Monte Carlo proof.
\end{proof}

\subsection{Application to non-anticipation}

For the test in Section~\ref{subsec:nonanticipation} of the main
paper, choose $h$ before observing assignments or outcomes.
Let $h^+<T$ be the last period of the assignment block containing
$h$, and define
\begin{equation}\label{eq:prefix-map}
\Gamma_h(\w)=\w_{1:h^+},
\qquad
S_h(\w)
=
\{\w'\in\mathcal W_0:\Gamma_h(\w')=\Gamma_h(\w)\}.
\end{equation}
The function $\Gamma_h$ is fixed by the prespecified $h$ and
the design-block boundaries. Its realized conditioning event is
\[
\mathcal C_{\Gamma_h}(\Wobs)
=
\{\W_{1:h^+}=\Wobs_{1:h^+}\}
=
\mathcal C_h^{\mathrm{obs}}.
\]

\begin{proof}[Proof of Theorem~\ref{thm:valid-pirt-prefix}]
Under the independent block assignment in
Definition~\ref{def:regular-switchback}, conditioning on
$S_h(\Wobs)$ fixes all complete blocks through $h^+$ and
leaves every later block label independent with its original
treatment probability. This is the conditional design law used
by the test.

Fix $\w,\w'$ in the same conditioning set. Since the sequences
agree through $h^+\ge1$ and $h^+<T$, the set
$I(\w,\w')=\mathbb I(\w,\w')\cap[T-1]$ is nonempty and
symmetric in its arguments. Under $H_0^{\mathrm{NA}}$,
\[
Y_t(\w)=Y_t(\w')
\qquad\text{for every }t\in I(\w,\w').
\]
Consequently, the outcomes and their centering mean in
\eqref{eq:Tna-centered} are identical whether the statistic
is evaluated with $\Y(\w)$ or $\Y(\w')$. For either ordered
pair $(\mathbf a,\mathbf b)\in\{(\w,\w'),(\w',\w)\}$,
\[
T_{\mathrm{NA}}\{\Y(\w),\mathbf a,\mathbf b\}
=
T_{\mathrm{NA}}\{\Y(\w'),\mathbf a,\mathbf b\}.
\]
Thus $T_{\mathrm{NA}}$ satisfies
\eqref{eq:pairwise-null-imputability} under $\Gamma_h$.
Theorem~\ref{thm:valid-pirt} gives conditional validity in every
realized conditioning set and marginal validity under the
original switchback design. The equality also holds after
taking absolute values, which establishes the two-sided result.
\end{proof}

\subsection{Boundary identity for the centered prefix score}

\begin{proof}[Proof of~\eqref{eq:Tna-boundary-identity}]
Let $\w\ne\w'$ belong to the same conditioning set, and let
\[
d=\min\{t:w_t\ne w_t'\}.
\]
Then the paths agree through $d-1$, so
$I(\w,\w')=\{1,\ldots,d-1\}$. In the difference between the
two directional scores, the weights $2w_{t+1}-1$ and
$2w'_{t+1}-1$ coincide for $t\le d-2$ and differ only at
$t=d-1$. Therefore
\begin{align*}
&T_{\mathrm{NA}}(\Y,\w',\w)
-T_{\mathrm{NA}}(\Y,\w,\w')\\
&\qquad=
\frac{2(w'_d-w_d)}{d-1}
\left(Y_{d-1}-\bar Y_{1:(d-1)}\right),
\end{align*}
which proves~\eqref{eq:Tna-boundary-identity}. Earlier outcomes
enter the directional comparison only through the centering mean.
\end{proof}

\section{Weak-null theory for fixed-length sessions: proofs}\label{app:weak-null}

This section proves Theorem~\ref{thm:weak-null-general}. The focal-average, coordinatewise, and joint position-wise procedures are obtained by taking $A=n^{-1}\one_n^\top$, $A=\boldsymbol e_\ell^\top$, and $A=\mathbf I_n$, respectively.

Throughout, $L>m$ is fixed, $T=JL$, and session $j$ is $[s_j,e_j]=\{(j-1)L+1,\ldots,jL\}$. Each session is one assignment block with label $Z_j\sim\mathrm{Bernoulli}(p_j)$, and its $n=L-m$ post--burn-in potential outcomes are collected in $\boldsymbol Y_j(z)\in\mathbb R^n$. Under Assumptions~\ref{ass:no-anticipation} and~\ref{ass:m-carryover},
\[
\boldsymbol Y_j^{\mathrm{obs}}
=Z_j\boldsymbol Y_j(1)+(1-Z_j)\boldsymbol Y_j(0).
\]

\subsection{Session-level randomization}

\begin{lemma}[Independent Bernoulli assignment on sessions]\label{lem:cond-bernoulli-p}
Under Definition~\ref{def:regular-switchback}, if the switch times are aligned with the fixed-length session boundaries, then $Z_1,\ldots,Z_J$ are independent and $Z_j\sim\mathrm{Bernoulli}(p_j)$, where $p_j$ is the corresponding design-block treatment probability.
\end{lemma}

\begin{proof}
Each session is exactly one design block. The claim follows directly from the independent block-level Bernoulli assignment in Definition~\ref{def:regular-switchback}.
\end{proof}

Fix a prespecified full-row-rank matrix $A\in\mathbb R^{d\times n}$ and write
\[
\boldsymbol X_j(z)=A\boldsymbol Y_j(z),
\qquad
\boldsymbol X_j^{\mathrm{obs}}=\boldsymbol X_j(Z_j),
\qquad
\boldsymbol\Delta_{j,A}=\boldsymbol X_j(1)-\boldsymbol X_j(0).
\]
For later use, define the finite-$J$ matrices
\begin{align}
\boldsymbol\Sigma_{\mathrm R,A,J}
&=
\frac{1}{J}\sum_{j=1}^J
\left\{
\frac{\boldsymbol X_j(1)\boldsymbol X_j(1)^\top}{1-p_j}
+
\frac{\boldsymbol X_j(0)\boldsymbol X_j(0)^\top}{p_j}
\right\},
\label{eq:Sigma-RAJ}\\
\boldsymbol\Sigma_{\mathrm{up},A,J}
&=
\frac{1}{J}\sum_{j=1}^J
\left\{
\frac{\boldsymbol X_j(1)\boldsymbol X_j(1)^\top}{p_j}
+
\frac{\boldsymbol X_j(0)\boldsymbol X_j(0)^\top}{1-p_j}
\right\},
\label{eq:Sigma-upAJ}\\
\boldsymbol\Sigma_{\mathrm S,A,J}
&=
\frac{1}{J}\sum_{j=1}^J
\boldsymbol h_{j,A}\boldsymbol h_{j,A}^\top,
\qquad
\boldsymbol h_{j,A}
=
\sqrt{\frac{1-p_j}{p_j}}\,\boldsymbol X_j(1)
+
\sqrt{\frac{p_j}{1-p_j}}\,\boldsymbol X_j(0).
\label{eq:Sigma-SAJ}
\end{align}
Assumption~\ref{ass:weak-null-asymp} and uniform integrability implied by the fourth-moment bound give
\begin{equation}\label{eq:Sigma-A-limits}
\boldsymbol\Sigma_{\mathrm R,A,J}\to\boldsymbol\Sigma_{\mathrm R,A}:=A\boldsymbol\Sigma_{\mathrm R}A^\top,
\qquad
\boldsymbol\Sigma_{\mathrm{up},A,J}\to\boldsymbol\Sigma_{\mathrm{up},A}:=A\boldsymbol\Sigma_{\mathrm{up}}A^\top,
\end{equation}
and $\boldsymbol\Sigma_{\mathrm S,A,J}\to\boldsymbol\Sigma_{\mathrm S,A}$ for the corresponding limit of~\eqref{eq:Sigma-SAJ}. Assumption~\ref{ass:weak-null-asymp}(iii) states directly that $\boldsymbol\Sigma_{\mathrm R,A}$ and $\boldsymbol\Sigma_{\mathrm{up},A}$ are positive definite.

\subsection{Conditional randomization limit}

Given the observed transformed vectors, draw $Z_1^\ast,\ldots,Z_J^\ast$ independently with $Z_j^\ast\sim\mathrm{Bernoulli}(p_j)$ and define
\begin{align*}
\widehat{\boldsymbol\tau}_A^\ast
&=
\frac{1}{J}\sum_{j=1}^J
\boldsymbol X_j^{\mathrm{obs}}
\left\{
\frac{Z_j^\ast}{p_j}
-
\frac{1-Z_j^\ast}{1-p_j}
\right\},\\
\widehat{\boldsymbol V}_{A,\mathrm{up}}^\ast
&=
\frac{1}{J^2}\sum_{j=1}^J
\boldsymbol X_j^{\mathrm{obs}}\boldsymbol X_j^{\mathrm{obs}\top}
\left\{
\frac{Z_j^\ast}{p_j^2}
+
\frac{1-Z_j^\ast}{(1-p_j)^2}
\right\},\\
Q_A^\ast
&=
\widehat{\boldsymbol\tau}_A^{\ast\top}
\bigl(\widehat{\boldsymbol V}_{A,\mathrm{up}}^\ast\bigr)^{-1}
\widehat{\boldsymbol\tau}_A^\ast.
\end{align*}

\begin{lemma}[Conditional randomization limit]\label{lem:weak-rand-vector}
Under Assumption~\ref{ass:weak-null-asymp}, conditional on $\{\boldsymbol X_j^{\mathrm{obs}}\}_{j=1}^J$,
\[
Q_A^\ast\Rightarrow\chi_d^2
\]
in $\mathbb P$-probability. When $d=1$,
\[
T_A^\ast
=
\frac{\widehat\tau_A^\ast}{\sqrt{\widehat V_{A,\mathrm{up}}^\ast}}
\Rightarrow\mathcal N(0,1)
\]
in $\mathbb P$-probability.
\end{lemma}

\begin{proof}
Condition on $\{\boldsymbol X_j^{\mathrm{obs}}\}_{j=1}^J$ and set
\[
\boldsymbol\xi_{j,A}^\ast
=
\boldsymbol X_j^{\mathrm{obs}}
\left\{
\frac{Z_j^\ast}{p_j}
-
\frac{1-Z_j^\ast}{1-p_j}
\right\}.
\]
The vectors $\boldsymbol\xi_{j,A}^\ast$ are conditionally independent and mean zero, with
\[
\operatorname{Var}^\ast(\boldsymbol\xi_{j,A}^\ast)
=
\frac{\boldsymbol X_j^{\mathrm{obs}}\boldsymbol X_j^{\mathrm{obs}\top}}
{p_j(1-p_j)}.
\]
Let
\begin{equation}\label{eq:Robs-A}
\boldsymbol R_{A,J}^{\mathrm{obs}}
=
\frac{1}{J}\sum_{j=1}^J
\frac{\boldsymbol X_j^{\mathrm{obs}}\boldsymbol X_j^{\mathrm{obs}\top}}
{p_j(1-p_j)}.
\end{equation}
Under the original assignment,
\[
\mathbb E\!\left[
\frac{\boldsymbol X_j^{\mathrm{obs}}\boldsymbol X_j^{\mathrm{obs}\top}}
{p_j(1-p_j)}
\right]
=
\frac{\boldsymbol X_j(1)\boldsymbol X_j(1)^\top}{1-p_j}
+
\frac{\boldsymbol X_j(0)\boldsymbol X_j(0)^\top}{p_j}.
\]
The fourth-moment bound and $p_j\in[\underline p,1-\underline p]$ imply an entrywise triangular-array law of large numbers, so, by~\eqref{eq:Sigma-A-limits},
\begin{equation}\label{eq:Robs-limit}
\boldsymbol R_{A,J}^{\mathrm{obs}}
\xrightarrow{\mathbb P}
\boldsymbol\Sigma_{\mathrm R,A}.
\end{equation}
In particular, the smallest eigenvalue of $\boldsymbol R_{A,J}^{\mathrm{obs}}$ is bounded away from zero with probability tending to one.

For any fixed nonzero $\boldsymbol a\in\mathbb R^d$, the conditional variance of $\sqrt J\,\boldsymbol a^\top\widehat{\boldsymbol\tau}_A^\ast$ is $\boldsymbol a^\top\boldsymbol R_{A,J}^{\mathrm{obs}}\boldsymbol a$, which converges in probability to the strictly positive quantity $\boldsymbol a^\top\boldsymbol\Sigma_{\mathrm R,A}\boldsymbol a$. Bounded assignment probabilities and the fourth-moment condition give
\[
\sum_{j=1}^J
\mathbb E^\ast\!\left[
|\boldsymbol a^\top\boldsymbol\xi_{j,A}^\ast|^4
\right]
=O_{\mathbb P}(J).
\]
Consequently, the conditional Lyapunov ratio for this projection is $O_{\mathbb P}(J^{-1})$. The scalar conditional Lyapunov theorem and the Cram\'er--Wold device therefore yield
\[
\sqrt J\,\widehat{\boldsymbol\tau}_A^\ast
\Rightarrow
\mathcal N(\zero_d,\boldsymbol\Sigma_{\mathrm R,A})
\]
conditionally in $\mathbb P$-probability. Combining this result with~\eqref{eq:Robs-limit} and Slutsky's theorem gives
\begin{equation}\label{eq:rand-vector-clt}
\bigl(\boldsymbol R_{A,J}^{\mathrm{obs}}\bigr)^{-1/2}
\sqrt J\,\widehat{\boldsymbol\tau}_A^\ast
\Rightarrow
\mathcal N(\zero_d,\mathbf I_d)
\end{equation}
conditionally in $\mathbb P$-probability.

Next,
\[
J\widehat{\boldsymbol V}_{A,\mathrm{up}}^\ast
=
\frac{1}{J}\sum_{j=1}^J
\boldsymbol X_j^{\mathrm{obs}}\boldsymbol X_j^{\mathrm{obs}\top}
\left\{
\frac{Z_j^\ast}{p_j^2}
+
\frac{1-Z_j^\ast}{(1-p_j)^2}
\right\}.
\]
Its conditional expectation is $\boldsymbol R_{A,J}^{\mathrm{obs}}$. The conditional variance of each matrix entry is bounded by a constant times
$J^{-2}\sum_j\|\boldsymbol X_j^{\mathrm{obs}}\|_2^4=O_{\mathbb P}(J^{-1})$. Therefore,
\begin{equation}\label{eq:rand-studentizer}
J\widehat{\boldsymbol V}_{A,\mathrm{up}}^\ast
-
\boldsymbol R_{A,J}^{\mathrm{obs}}
\xrightarrow{\mathbb P^\ast}0
\end{equation}
in $\mathbb P$-probability. Combining~\eqref{eq:rand-vector-clt} and~\eqref{eq:rand-studentizer} with Slutsky's theorem gives $Q_A^\ast\Rightarrow\chi_d^2$. The scalar signed result follows from the same argument with $d=1$.
\end{proof}

\subsection{Sampling limit and covariance dominance}

Define the centered HT summand
\begin{equation}\label{eq:psi-A}
\boldsymbol\psi_{j,A}
=
\left(\frac{Z_j}{p_j}-1\right)\boldsymbol X_j(1)
-
\left(\frac{1-Z_j}{1-p_j}-1\right)\boldsymbol X_j(0).
\end{equation}
Then
\begin{equation}\label{eq:ht-centered-A}
\widehat{\boldsymbol\tau}_A-A\boldsymbol\tau_J
=
\frac{1}{J}\sum_{j=1}^J\boldsymbol\psi_{j,A},
\end{equation}
and the summands are independent and mean zero.

\begin{lemma}[Sampling limit and conservative covariance]\label{lem:weak-sampling-vector}
Under Assumption~\ref{ass:weak-null-asymp},
\begin{align}
\sqrt J\bigl(\widehat{\boldsymbol\tau}_A-A\boldsymbol\tau_J\bigr)
&\Rightarrow
\mathcal N(\zero_d,\boldsymbol\Sigma_{\mathrm S,A}),
\label{eq:sampling-clt-A}\\
J\widehat{\boldsymbol V}_{A,\mathrm{up}}
&\xrightarrow{\mathbb P}
\boldsymbol\Sigma_{\mathrm{up},A}.
\label{eq:Vup-limit-A}
\end{align}
Moreover,
\begin{equation}\label{eq:matrix-upper-identity}
\boldsymbol\Sigma_{\mathrm{up},A,J}
-
\boldsymbol\Sigma_{\mathrm S,A,J}
=
\frac{1}{J}\sum_{j=1}^J
\boldsymbol\Delta_{j,A}\boldsymbol\Delta_{j,A}^\top
\succeq\zero_{d\times d},
\end{equation}
so $\boldsymbol\Sigma_{\mathrm S,A}\preceq\boldsymbol\Sigma_{\mathrm{up},A}$.
\end{lemma}

\begin{proof}
A direct two-point calculation from~\eqref{eq:psi-A} gives
\begin{align*}
\operatorname{Var}(\boldsymbol\psi_{j,A})
&=
\frac{1-p_j}{p_j}\boldsymbol X_j(1)\boldsymbol X_j(1)^\top
+
\frac{p_j}{1-p_j}\boldsymbol X_j(0)\boldsymbol X_j(0)^\top\\
&\quad+
\boldsymbol X_j(1)\boldsymbol X_j(0)^\top
+
\boldsymbol X_j(0)\boldsymbol X_j(1)^\top
=
\boldsymbol h_{j,A}\boldsymbol h_{j,A}^\top.
\end{align*}
For each fixed $\boldsymbol a\in\mathbb R^d$, let
$v_{\boldsymbol a,J}=\boldsymbol a^\top\boldsymbol\Sigma_{\mathrm S,A,J}\boldsymbol a$ and
$v_{\boldsymbol a}=\lim_{J\to\infty}v_{\boldsymbol a,J}$. If $v_{\boldsymbol a}>0$, bounded assignment probabilities and the fourth-moment bound verify the scalar Lyapunov condition for that projection. If $v_{\boldsymbol a}=0$, then
\[
\operatorname{Var}\!\left[
\sqrt J\,\boldsymbol a^\top
\{\widehat{\boldsymbol\tau}_A-A\boldsymbol\tau_J\}
\right]
=v_{\boldsymbol a,J}\to0,
\]
so the projection converges to zero in probability by Chebyshev's inequality. The Cram\'er--Wold device therefore proves~\eqref{eq:sampling-clt-A}, including the possibility that $\boldsymbol\Sigma_{\mathrm S,A}$ is singular.

For the covariance estimator,
\[
J\widehat{\boldsymbol V}_{A,\mathrm{up}}
=
\frac{1}{J}\sum_{j=1}^J
\boldsymbol X_j^{\mathrm{obs}}\boldsymbol X_j^{\mathrm{obs}\top}
\left\{
\frac{Z_j}{p_j^2}
+
\frac{1-Z_j}{(1-p_j)^2}
\right\}.
\]
The expectation of the $j$th summand is
\[
\frac{\boldsymbol X_j(1)\boldsymbol X_j(1)^\top}{p_j}
+
\frac{\boldsymbol X_j(0)\boldsymbol X_j(0)^\top}{1-p_j}.
\]
An entrywise variance bound using the fourth moments gives~\eqref{eq:Vup-limit-A}.

Finally, subtract the covariance of $\boldsymbol\psi_{j,A}$ from the $j$th summand of~\eqref{eq:Sigma-upAJ}:
\begin{align*}
&\frac{\boldsymbol X_j(1)\boldsymbol X_j(1)^\top}{p_j}
+
\frac{\boldsymbol X_j(0)\boldsymbol X_j(0)^\top}{1-p_j}
-
\operatorname{Var}(\boldsymbol\psi_{j,A})\\
&\qquad=
\bigl\{\boldsymbol X_j(1)-\boldsymbol X_j(0)\bigr\}
\bigl\{\boldsymbol X_j(1)-\boldsymbol X_j(0)\bigr\}^\top
=
\boldsymbol\Delta_{j,A}\boldsymbol\Delta_{j,A}^\top.
\end{align*}
Averaging gives~\eqref{eq:matrix-upper-identity}; taking limits gives
$\boldsymbol\Sigma_{\mathrm S,A}\preceq\boldsymbol\Sigma_{\mathrm{up},A}$.
\end{proof}

\subsection{Proof of Theorem~\ref{thm:weak-null-general}}

\begin{proof}[Proof of Theorem~\ref{thm:weak-null-general}]
Let $c_{A,1-\alpha}$ be the $(1-\alpha)$ quantile of the full conditional randomization distribution of $Q_A^\ast$. By Lemma~\ref{lem:weak-rand-vector} and continuity of the $\chi_d^2$ distribution,
\begin{equation}\label{eq:critical-A}
c_{A,1-\alpha}
\xrightarrow{\mathbb P}
\chi^2_{d,1-\alpha}.
\end{equation}
Under $H_0(A)$, $A\boldsymbol\tau_J=\zero_d$. Lemma~\ref{lem:weak-sampling-vector} and Slutsky's theorem imply
\[
Q_A
\Rightarrow
Q_{A,\infty}
:=
\boldsymbol Z^\top
\boldsymbol\Lambda_A
\boldsymbol Z,
\qquad
\boldsymbol Z\sim\mathcal N(\zero_d,\mathbf I_d),
\]
where
\[
\boldsymbol\Lambda_A
=
\boldsymbol\Sigma_{\mathrm{up},A}^{-1/2}
\boldsymbol\Sigma_{\mathrm S,A}
\boldsymbol\Sigma_{\mathrm{up},A}^{-1/2}
\preceq\mathbf I_d.
\]
Let $\lambda_1,\ldots,\lambda_d\in[0,1]$ be the eigenvalues of $\boldsymbol\Lambda_A$. By orthogonal diagonalization,
\[
Q_{A,\infty}
\stackrel{d}{=}
\sum_{r=1}^d\lambda_r Z_r^2
\le
\sum_{r=1}^d Z_r^2
\quad\text{under a common coupling}.
\]
Consequently,
\[
\mathbb P\!
\left(Q_{A,\infty}\ge\chi^2_{d,1-\alpha}\right)
\le\alpha.
\]
Combining this stochastic dominance with~\eqref{eq:critical-A} yields
\[
\limsup_{J\to\infty}
\mathbb P\!
\left(Q_A\ge c_{A,1-\alpha}\right)
\le\alpha.
\]
Under the usual lower-quantile convention, the upper-tail randomization
$p$-value satisfies
\[
\{\widehat p_A\le\alpha\}
\subseteq
\{Q_A\ge c_{A,1-\alpha}\}.
\]
The inclusion may be strict because the finite randomization distribution
can be discrete and contain ties. Therefore,
\[
\limsup_{J\to\infty}
\mathbb P(\widehat p_A\le\alpha)
\le\alpha.
\]

When $d=1$, Lemmas~\ref{lem:weak-rand-vector} and~\ref{lem:weak-sampling-vector} give
\[
T_A^\ast\Rightarrow\mathcal N(0,1)
\quad\text{conditionally},
\qquad
T_A\Rightarrow\mathcal N(0,\rho_A^2),
\qquad
0\le\rho_A^2\le1.
\]
For an upper-tail test and $\alpha\in(0,1/2)$, we have $z_{1-\alpha}>0$. If $\rho_A=0$, the limiting rejection probability at this positive threshold is zero. If $\rho_A>0$, then
\[
\mathbb P\{\mathcal N(0,\rho_A^2)\ge z_{1-\alpha}\}
=1-\Phi\!\left(z_{1-\alpha}/\rho_A\right)
\le 1-\Phi(z_{1-\alpha})
=\alpha.
\]
Conditional critical-value convergence proves the one-sided result. The two-sided statistic is $T_A^2=Q_A$, so its validity for every $\alpha\in(0,1)$ is already covered by the quadratic-form argument.

Finally, let $p_{A,J}^{\mathrm{full}}$ denote the full-reference
$p$-value and let $\widehat p_{A,J}^{\mathrm{MC}}$ denote its plus-one
Monte Carlo approximation based on $B_J$ independent draws. If
$B_J\to\infty$, the conditional law of large numbers gives
\[
\left|
\widehat p_{A,J}^{\mathrm{MC}}
-
p_{A,J}^{\mathrm{full}}
\right|
=o_{\mathbb P}(1).
\]
Hence, for every $\varepsilon>0$,
\[
\mathbb P\!\left(
\widehat p_{A,J}^{\mathrm{MC}}\le\alpha
\right)
\le
\mathbb P\!\left(
p_{A,J}^{\mathrm{full}}\le\alpha+\varepsilon
\right)
+
\mathbb P\!\left(
\left|
\widehat p_{A,J}^{\mathrm{MC}}
-
p_{A,J}^{\mathrm{full}}
\right|>\varepsilon
\right).
\]
Taking $\limsup$ and then letting $\varepsilon\downarrow0$ transfers
the preceding size bound to the Monte Carlo implementation. For the
one-sided scalar result, take $\varepsilon$ sufficiently small that
$\alpha+\varepsilon<1/2$.
\end{proof}

\paragraph{Special cases.}
For $A=n^{-1}\one_n^\top$, Theorem~\ref{thm:weak-null-general} gives the focal-average test. For $A=\boldsymbol e_\ell^\top$, it gives the position-$\ell$ signed or two-sided test. For $A=\mathbf I_n$, it gives the joint position-wise quadratic-form test. More generally, the same proof applies to any prespecified fixed-dimensional full-row-rank linear contrast of the post--burn-in effect profile.

\begin{lemma}[Joint position-wise null implies the focal-average null]\label{lem:sw-implies-tauc}
If $H_0^{\mathrm{sw}}$ in~\eqref{eq:weak-null-session} holds, then $H_0^{\mathcal U}$ in~\eqref{eq:weak-null-focal} holds.
\end{lemma}

\begin{proof}
By definition,
\[
\frac{1}{Jn}\sum_{j=1}^J\sum_{\ell=1}^n
\{Y_{j\ell}(1)-Y_{j\ell}(0)\}
=
\frac{1}{n}\sum_{\ell=1}^n
\left[
\frac{1}{J}\sum_{j=1}^J
\{Y_{j\ell}(1)-Y_{j\ell}(0)\}
\right].
\]
Every bracketed term is zero under $H_0^{\mathrm{sw}}$.
\end{proof}

\paragraph{Connection to the regression-based focal-average test.}
Let $\boldsymbol a_n=n^{-1}\one_n$ and define $\bar Y_j(z)=\boldsymbol a_n^\top\boldsymbol Y_j(z)$. The scalar limit law used below is the pushforward $\nu_{\mathcal U}$ of $\nu$ under
\[
(p,\boldsymbol y_1,\boldsymbol y_0)
\longmapsto
(p,\boldsymbol a_n^\top\boldsymbol y_1,\boldsymbol a_n^\top\boldsymbol y_0).
\]
Assumption~\ref{ass:weak-null-asymp} implies the scalar fourth-moment and stabilization conditions for this projection. The regression result below imposes its additional residual-variance nondegeneracy conditions with expectations taken under $\nu_{\mathcal U}$. For notational simplicity, that proof denotes the pushforward law again by $\nu$.

\begin{corollary}[Asymptotic validity of the regression-based weak-null test]\label{cor:weak-null-valid-reg}
Under the fixed-length session setup, Assumptions~\ref{ass:no-anticipation} and~\ref{ass:m-carryover}, and Assumption~\ref{ass:weak-null-asymp}, the conditional p-value based on the regression studentized statistic defined in Section~\ref{app:proof-cor-weak-null-valid-reg} satisfies, under $H_0^{\cal U}$,
\[
\limsup_{J\to\infty}\mathbb P\{\widehat p_{\mathrm{reg}}\le \alpha\}\le \alpha,
\qquad \alpha\in(0,1/2],
\]
provided the limiting conditional-reference, sampling, and HC0 residual variances
are nondegenerate. The result uses the exact conditional reference law or a
Monte Carlo approximation with $B_J\to\infty$.
\end{corollary}

\subsection{Proof of Corollary~\ref{cor:weak-null-valid-reg}}\label{app:proof-cor-weak-null-valid-reg}

Recall that $Z_1,\dots,Z_J$ are independent with
$Z_j\sim\mathrm{Bernoulli}(p_j)$, where $\underline p\le p_j\le 1-\underline p$ for all $j$ and all $J$.
Let $\bar Y_j(z)$ be the session-level focal mean potential outcome under the constant path $\mathbf z_T$ and
$\bar Y_j^{\mathrm{obs}}=\bar Y_j(Z_j)$.

\paragraph{Regression-based statistic.}
The weighted-regression coefficient $\hat\tau_{\mathrm{reg}}$ equals the difference of the following IPW-normalized means:
\[
\hat\tau_{\mathrm{reg}}=\hat\mu_1-\hat\mu_0,
\qquad
\hat\mu_1=\frac{\sum_{j=1}^J \frac{Z_j\,\bar Y_j(1)}{p_j}}{\sum_{j=1}^J \frac{Z_j}{p_j}},
\qquad
\hat\mu_0=\frac{\sum_{j=1}^J \frac{(1-Z_j)\,\bar Y_j(0)}{1-p_j}}{\sum_{j=1}^J \frac{1-Z_j}{1-p_j}}.
\]
Let $\widehat V_{\mathrm{reg}}$ be the corresponding HC0 variance estimator and set
$T_{\mathrm{reg}}=\hat\tau_{\mathrm{reg}}/\sqrt{\widehat V_{\mathrm{reg}}}$.
For the randomization distribution, conditional on $\{\bar Y_j^{\mathrm{obs}}\}_{j=1}^J$ draw
$Z_1^\ast,\dots,Z_J^\ast$ independently with $Z_j^\ast\sim\mathrm{Bernoulli}(p_j)$ and define
$T_{\mathrm{reg}}^\ast$ by recomputing the normalized means and HC0-studentized statistic with $Z^\ast$ in place of $Z$.
To make the statistic well defined on the entire assignment support, set $T_{\mathrm{reg}}=0$ whenever either normalized-arm denominator is zero, and apply the same convention to every reference draw. Because $p_j\in[\underline p,1-\underline p]$, the probability of an empty treatment or control arm is at most $2(1-\underline p)^J$ and therefore vanishes exponentially. This finite-sample convention does not affect the asymptotic result.

\paragraph{Additional nondegeneracy.}
Because $\widehat V_{\mathrm{reg}}$ is based on within-arm residual variation, we impose the following mild
nondegeneracy condition in addition to Assumption~\ref{ass:weak-null-asymp}.  Let $(P,Y_1,Y_0)\sim\nu$ be the limit law
from Assumption~\ref{ass:weak-null-asymp}(ii), and define $\mu_1:=\mathbb{E}_\nu[Y_1]$ and $\mu_0:=\mathbb{E}_\nu[Y_0]$.
Also define the probability-weighted observed-outcome mean
\[
\mu_{\mathrm{obs}}
:=\mathbb E_\nu\{P Y_1+(1-P)Y_0\}.
\]
Assume the three conditions
\begin{align}\label{eq:reg-nondeg-proof}
\sigma_{\mathrm R,\mathrm{reg}}^2
&:=
\mathbb E_\nu\!\left[
\frac{(Y_1-\mu_{\mathrm{obs}})^2}{1-P}
+\frac{(Y_0-\mu_{\mathrm{obs}})^2}{P}
\right]>0,
\notag\\[-0.1cm]
\sigma_{\mathrm{up,reg}}^2
&:=
\mathbb{E}_\nu\!\left[\frac{(Y_1-\mu_1)^2}{P}+\frac{(Y_0-\mu_0)^2}{1-P}\right]
>0,
\notag\\[-0.1cm]
\sigma_{\mathrm{reg}}^2
&:=
\mathbb{E}_\nu\!\left[
\left\{
\sqrt{\frac{1-P}{P}}(Y_1-\mu_1)
+\sqrt{\frac{P}{1-P}}(Y_0-\mu_0)
\right\}^2
\right]>0.
\end{align}
The first condition rules out an asymptotically constant conditional reference
distribution. The second makes the observed HC0 limit nondegenerate, and the
third excludes exact cancellation in the sampling contrast.

\subsubsection*{Step 1: a randomization CLT for $T_{\mathrm{reg}}^\ast$ and quantile convergence}

\begin{lemma}[Randomization CLT]\label{lem:rand-clt-reg-proof}
Under Assumption~\ref{ass:weak-null-asymp} and \eqref{eq:reg-nondeg-proof}, conditional on
$\{\bar Y_j^{\mathrm{obs}}\}_{j=1}^J$ we have
\[
T_{\mathrm{reg}}^\ast \ \Rightarrow\ \mathcal N(0,1),
\]
in $\mathbb P$-probability (where $\mathbb P$ is the original assignment law generating
$\{\bar Y_j^{\mathrm{obs}}\}$).
Consequently, if $c_{1-\alpha}$ denotes the $(1-\alpha)$ quantile of the conditional randomization distribution of
$T_{\mathrm{reg}}^\ast$ given $\{\bar Y_j^{\mathrm{obs}}\}$, then $c_{1-\alpha}\to z_{1-\alpha}$ in probability.
\end{lemma}

\begin{proof}
The empty-arm event has probability at most $2(1-\underline p)^J$; the convention above defines the statistic there. It therefore suffices to establish the claimed limit on the complementary event. Fix $J$ and condition on the realized focal means $\{\bar Y_j^{\mathrm{obs}}\}_{j=1}^J$; write
$y_j:=\bar Y_j^{\mathrm{obs}}$ and $\bar y:=J^{-1}\sum_{j=1}^J y_j$.
Define the resampled IPW means
\[
\hat\mu_1^\ast
=
\frac{A_{1,J}^\ast}{B_{1,J}^\ast},
\qquad
A_{1,J}^\ast:=\frac{1}{J}\sum_{j=1}^J \frac{Z_j^\ast y_j}{p_j},
\quad
B_{1,J}^\ast:=\frac{1}{J}\sum_{j=1}^J \frac{Z_j^\ast}{p_j},
\]
\[
\hat\mu_0^\ast
=
\frac{A_{0,J}^\ast}{B_{0,J}^\ast},
\qquad
A_{0,J}^\ast:=\frac{1}{J}\sum_{j=1}^J \frac{(1-Z_j^\ast) y_j}{1-p_j},
\quad
B_{0,J}^\ast:=\frac{1}{J}\sum_{j=1}^J \frac{1-Z_j^\ast}{1-p_j},
\]
and $\hat\tau_{\mathrm{reg}}^\ast:=\hat\mu_1^\ast-\hat\mu_0^\ast$.

\paragraph{Step 1a: linearization of $\hat\tau_{\mathrm{reg}}^\ast$.}
Since $\mathbb E[Z_j^\ast/p_j]=1$ and $\mathbb E[(1-Z_j^\ast)/(1-p_j)]=1$,
\[
\mathbb E\!\left[A_{1,J}^\ast \,\middle|\,\{y_j\}\right]=\bar y,
\quad
\mathbb E\!\left[B_{1,J}^\ast \,\middle|\,\{y_j\}\right]=1,
\qquad
\mathbb E\!\left[A_{0,J}^\ast \,\middle|\,\{y_j\}\right]=\bar y,
\quad
\mathbb E\!\left[B_{0,J}^\ast \,\middle|\,\{y_j\}\right]=1.
\]
Write
\[
\hat\mu_1^\ast-\bar y
=
\frac{(A_{1,J}^\ast-\bar y)-\bar y(B_{1,J}^\ast-1)}{B_{1,J}^\ast},
\qquad
\hat\mu_0^\ast-\bar y
=
\frac{(A_{0,J}^\ast-\bar y)-\bar y(B_{0,J}^\ast-1)}{B_{0,J}^\ast}.
\]
A direct simplification yields
\[
(A_{1,J}^\ast-\bar y)-\bar y(B_{1,J}^\ast-1)
=
\frac{1}{J}\sum_{j=1}^J\left(\frac{Z_j^\ast}{p_j}-1\right)(y_j-\bar y),
\]
\[
(A_{0,J}^\ast-\bar y)-\bar y(B_{0,J}^\ast-1)
=
\frac{1}{J}\sum_{j=1}^J\left(\frac{1-Z_j^\ast}{1-p_j}-1\right)(y_j-\bar y).
\]
Moreover, conditional on $\{y_j\}$,
\[
\mathrm{Var}\!\left(B_{1,J}^\ast \,\middle|\,\{y_j\}\right)
=
\frac{1}{J^2}\sum_{j=1}^J \frac{1-p_j}{p_j}
\le \frac{C}{J},
\qquad
\mathrm{Var}\!\left(B_{0,J}^\ast \,\middle|\,\{y_j\}\right)
=
\frac{1}{J^2}\sum_{j=1}^J \frac{p_j}{1-p_j}
\le \frac{C}{J},
\]
so
\[
B_{1,J}^\ast-1=O_{\mathbb P^\ast}(J^{-1/2}),
\qquad
B_{0,J}^\ast-1=O_{\mathbb P^\ast}(J^{-1/2})
\]
conditionally on $\{y_j\}$, in $\mathbb P$-probability. The centered
numerators in the two ratio identities above are also
$O_{\mathbb P^\ast}(J^{-1/2})$, since their conditional variances are
bounded by
\[
\frac{C_{\underline p}}{J^2}
\sum_{j=1}^J (y_j-\bar y)^2
=
O_{\mathbb P}(J^{-1}).
\]
Consequently, replacing $(B_{1,J}^\ast)^{-1}$ and
$(B_{0,J}^\ast)^{-1}$ by one contributes only
$O_{\mathbb P^\ast}(J^{-1})=o_{\mathbb P^\ast}(J^{-1/2})$.
Hence,
\begin{equation}\label{eq:lin-tau-star}
\hat\tau_{\mathrm{reg}}^\ast
=
\frac{1}{J}\sum_{j=1}^J \xi_{j,J}^\ast
+ r_J^\ast,
\qquad
\xi_{j,J}^\ast
:=
(y_j-\bar y)\left(\frac{Z_j^\ast}{p_j}-\frac{1-Z_j^\ast}{1-p_j}\right),
\end{equation}
where
$r_J^\ast=o_{\mathbb P^\ast}(J^{-1/2})$
conditionally on $\{y_j\}$, in $\mathbb P$-probability.

\paragraph{Step 1b: conditional CLT for $\sum \xi_{j,J}^\ast$.}
Conditional on $\{y_j\}$, the $\{\xi_{j,J}^\ast\}_{j=1}^J$ are independent, mean zero, and satisfy
\[
\mathrm{Var}\!\left(\xi_{j,J}^\ast \,\middle|\,\{y_j\}\right)
=
\frac{(y_j-\bar y)^2}{p_j(1-p_j)}.
\]
Let
\[
s_J^2
:=
\sum_{j=1}^J \frac{(y_j-\bar y)^2}{p_j(1-p_j)}.
\]
Because $p_j\in[\underline p,1-\underline p]$, we have $s_J^2\asymp \sum_{j=1}^J (y_j-\bar y)^2$.
Assumption~\ref{ass:weak-null-asymp}(i) implies $J^{-1}\sum_{j=1}^J y_j^4=O_{\mathbb P}(1)$, hence
$J^{-1}\sum_{j=1}^J (y_j-\bar y)^4=O_{\mathbb P}(1)$, and therefore
\[
\sum_{j=1}^J (y_j-\bar y)^4 = O_{\mathbb P}(J).
\]
In addition, the first condition in~\eqref{eq:reg-nondeg-proof} gives
\[
\frac{s_J^2}{J}\ \xrightarrow{\ \mathbb P\ }\
\sigma_{\mathrm R,\mathrm{reg}}^2>0.
\]
To see this, the triangular-array law of large numbers gives
$\bar y\to\mu_{\mathrm{obs}}$ in probability. Replacing $\bar y$ by
$\mu_{\mathrm{obs}}$ in $s_J^2/J$ is then $o_{\mathbb P}(1)$ by the fourth-moment
bound, while another law of large numbers and stabilization of the empirical
measures give the displayed limit. Thus $s_J^2=\Theta_{\mathbb P}(J)$.

Furthermore, conditional on $\{y_j\}$,
\[
\mathbb E\!\left[(\xi_{j,J}^\ast)^4\,\middle|\,\{y_j\}\right]
=
(y_j-\bar y)^4\left(\frac{1}{p_j^3}+\frac{1}{(1-p_j)^3}\right)
\le C_{\underline p}\,(y_j-\bar y)^4.
\]
Hence, with $\mathbb P$-probability tending to one,
\[
\frac{1}{s_J^4}\sum_{j=1}^J
\mathbb E\!\left[(\xi_{j,J}^\ast)^4\,\middle|\,\{y_j\}\right]
\le
\frac{C_{\underline p}\sum_{j=1}^J (y_j-\bar y)^4}{s_J^4}
=
O_{\mathbb P}\!\left(\frac{J}{J^2}\right)
\to 0.
\]
By the conditional Lyapunov CLT,
\begin{equation}\label{eq:clt-xi-star}
\frac{\sum_{j=1}^J \xi_{j,J}^\ast}{s_J}\ \Rightarrow\ \mathcal N(0,1)
\qquad\text{conditionally on }\{y_j\}.
\end{equation}
Combining \eqref{eq:lin-tau-star} and \eqref{eq:clt-xi-star} yields
\[
\frac{\sqrt{J}\,\hat\tau_{\mathrm{reg}}^\ast}{\sqrt{s_J^2/J}}\ \Rightarrow\ \mathcal N(0,1)
\qquad\text{conditionally on }\{y_j\}.
\]

\paragraph{Step 1c: consistency of $J\widehat V_{\mathrm{reg}}^\ast$ for $s_J^2/J$.}
Define the resampled residuals $\hat r_{j,1}^\ast:=y_j-\hat\mu_1^\ast$ for $Z_j^\ast=1$ and
$\hat r_{j,0}^\ast:=y_j-\hat\mu_0^\ast$ for $Z_j^\ast=0$, and let $\widehat V_{\mathrm{reg}}^\ast$ be the HC0
variance estimator obtained from $(Z^\ast,\hat\mu_1^\ast,\hat\mu_0^\ast)$.
Write
\[
J\,\widehat V_{\mathrm{reg}}^\ast
=
\frac{N_{1,J}^\ast}{(B_{1,J}^\ast)^2}+\frac{N_{0,J}^\ast}{(B_{0,J}^\ast)^2},
\]
where
\[
N_{1,J}^\ast:=\frac{1}{J}\sum_{j=1}^J \frac{Z_j^\ast(\hat r_{j,1}^\ast)^2}{p_j^2},
\qquad
N_{0,J}^\ast:=\frac{1}{J}\sum_{j=1}^J \frac{(1-Z_j^\ast)(\hat r_{j,0}^\ast)^2}{(1-p_j)^2}.
\]
Since $B_{1,J}^\ast\to 1$ and $B_{0,J}^\ast\to 1$ in conditional probability, it suffices to show
$N_{1,J}^\ast+N_{0,J}^\ast \to s_J^2/J$ in conditional probability.

First replace $\hat\mu_1^\ast,\hat\mu_0^\ast$ by $\bar y$.
By the two ratio identities and variance bounds in Step~1a,
\[
\hat\mu_1^\ast-\bar y
=
O_{\mathbb P^\ast}(J^{-1/2}),
\qquad
\hat\mu_0^\ast-\bar y
=
O_{\mathbb P^\ast}(J^{-1/2})
\]
conditionally on $\{y_j\}$, in $\mathbb P$-probability. Therefore, a
direct expansion gives
\[
\left|N_{1,J}^\ast-\tilde N_{1,J}^\ast\right|
+
\left|N_{0,J}^\ast-\tilde N_{0,J}^\ast\right|
=
o_{\mathbb P^\ast}(1)
\qquad\text{conditionally on }\{y_j\},
\text{ in $\mathbb P$-probability},
\]
where
\[
\tilde N_{1,J}^\ast:=\frac{1}{J}\sum_{j=1}^J \frac{Z_j^\ast(y_j-\bar y)^2}{p_j^2},
\qquad
\tilde N_{0,J}^\ast:=\frac{1}{J}\sum_{j=1}^J \frac{(1-Z_j^\ast)(y_j-\bar y)^2}{(1-p_j)^2}.
\]
Next, conditional on $\{y_j\}$,
\[
\mathbb E[\tilde N_{1,J}^\ast+\tilde N_{0,J}^\ast \mid \{y_j\}]
=
\frac{1}{J}\sum_{j=1}^J (y_j-\bar y)^2\left(\frac{1}{p_j}+\frac{1}{1-p_j}\right)
=
\frac{s_J^2}{J}.
\]
Moreover, using $p_j\in[\underline p,1-\underline p]$ and independence,
\[
\mathrm{Var}\!\left(\tilde N_{1,J}^\ast+\tilde N_{0,J}^\ast \,\middle|\, \{y_j\}\right)
\le
\frac{C_{\underline p}}{J^2}\sum_{j=1}^J (y_j-\bar y)^4
=
O_{\mathbb P}(1/J)
\to 0.
\]
Thus $\tilde N_{1,J}^\ast+\tilde N_{0,J}^\ast \to s_J^2/J$ in conditional probability, and hence
$N_{1,J}^\ast+N_{0,J}^\ast \to s_J^2/J$ as well. Therefore
\[
\frac{J\,\widehat V_{\mathrm{reg}}^\ast}{s_J^2/J}
\ \xrightarrow{\ \mathbb P^\ast\ }\ 1
\qquad\text{conditionally on }\{y_j\},
\text{ in $\mathbb P$-probability}.
\]
Combining with \eqref{eq:clt-xi-star} and Slutsky's theorem yields
$T_{\mathrm{reg}}^\ast \Rightarrow \mathcal N(0,1)$ conditionally on $\{y_j\}$, in $\mathbb P$-probability.

\paragraph{Quantile convergence.}
Let $F_J^\ast(\cdot):=\mathbb P(T_{\mathrm{reg}}^\ast\le \cdot \mid \{y_j\})$ be the conditional CDF\@.
The above shows $F_J^\ast(t)\to \Phi(t)$ in probability for each $t\in\mathbb R$, and $\Phi$ is continuous and strictly
increasing. Therefore the conditional $(1-\alpha)$ quantile
$c_{1-\alpha}:=\inf\{t:F_J^\ast(t)\ge 1-\alpha\}$ satisfies $c_{1-\alpha}\to z_{1-\alpha}$ in probability.
\end{proof}

\subsubsection*{Step 2: a sampling CLT for $T_{\mathrm{reg}}$ under $H_0^{\cal U}$}

Define the (deterministic) session-level means
\[
\mu_{1,J}:=\frac{1}{J}\sum_{j=1}^J \bar Y_j(1),
\qquad
\mu_{0,J}:=\frac{1}{J}\sum_{j=1}^J \bar Y_j(0),
\qquad
\tau_{\mathcal C}=\mu_{1,J}-\mu_{0,J}.
\]
Under $H_0^{\cal U}$, $\tau_{\mathcal C}=0$ by definition; the stronger $H_0^{\mathrm{sw}}$ implies the same equality by Lemma~\ref{lem:sw-implies-tauc}.

\begin{lemma}[Sampling CLT for $\hat\tau_{\mathrm{reg}}$]\label{lem:samp-clt-reg-proof}
Under Assumption~\ref{ass:weak-null-asymp} and \eqref{eq:reg-nondeg-proof},
\[
\sqrt{J}\bigl(\hat\tau_{\mathrm{reg}}-\tau_{\mathcal C}\bigr)\ \Rightarrow\ \mathcal N(0,\sigma_{\mathrm{reg}}^2),
\]
where
\[
\sigma_{\mathrm{reg}}^2
:=
\mathbb E_\nu\!\left[
\left(
\sqrt{\frac{1-P}{P}}\,(Y_1-\mu_1)
+
\sqrt{\frac{P}{1-P}}\,(Y_0-\mu_0)
\right)^2
\right].
\]
\end{lemma}

\begin{proof}
Write $\hat\tau_{\mathrm{reg}}=\hat\mu_1-\hat\mu_0$ with
\[
\hat\mu_1=\frac{A_{1,J}}{B_{1,J}},
\qquad
A_{1,J}:=\frac{1}{J}\sum_{j=1}^J \frac{Z_j\,\bar Y_j(1)}{p_j},
\quad
B_{1,J}:=\frac{1}{J}\sum_{j=1}^J \frac{Z_j}{p_j},
\]
\[
\hat\mu_0=\frac{A_{0,J}}{B_{0,J}},
\qquad
A_{0,J}:=\frac{1}{J}\sum_{j=1}^J \frac{(1-Z_j)\,\bar Y_j(0)}{1-p_j},
\quad
B_{0,J}:=\frac{1}{J}\sum_{j=1}^J \frac{1-Z_j}{1-p_j}.
\]
Because $\mathbb E[Z_j/p_j]=1$ and $\mathbb E[(1-Z_j)/(1-p_j)]=1$, we have
$\mathbb E[B_{1,J}]=\mathbb E[B_{0,J}]=1$ and
\[
\mathrm{Var}(B_{1,J})
=
\frac{1}{J^2}\sum_{j=1}^J \frac{1-p_j}{p_j}
\le \frac{C}{J},
\qquad
\mathrm{Var}(B_{0,J})
=
\frac{1}{J^2}\sum_{j=1}^J \frac{p_j}{1-p_j}
\le \frac{C}{J},
\]
so $B_{1,J}\to 1$ and $B_{0,J}\to 1$ in probability.

Next, note the exact identity
\[
\hat\mu_1-\mu_{1,J}
=
\frac{(A_{1,J}-\mu_{1,J})-\mu_{1,J}(B_{1,J}-1)}{B_{1,J}}
=
\frac{1}{B_{1,J}}\cdot \frac{1}{J}\sum_{j=1}^J \left(\frac{Z_j}{p_j}-1\right)\bigl(\bar Y_j(1)-\mu_{1,J}\bigr).
\]
The centered sum in the preceding identity is
$O_{\mathbb P}(J^{-1/2})$ by a direct variance bound, while
$B_{1,J}-1=O_{\mathbb P}(J^{-1/2})$. Hence replacing
$B_{1,J}^{-1}$ by one contributes
$O_{\mathbb P}(J^{-1})=o_{\mathbb P}(J^{-1/2})$, and therefore
\begin{equation}\label{eq:lin-mu1}
\hat\mu_1-\mu_{1,J}
=
\frac{1}{J}\sum_{j=1}^J \left(\frac{Z_j}{p_j}-1\right)\bigl(\bar Y_j(1)-\mu_{1,J}\bigr)
+o_{\mathbb P}(J^{-1/2}).
\end{equation}
The identical argument for the control arm gives
\begin{equation}\label{eq:lin-mu0}
\hat\mu_0-\mu_{0,J}
=
\frac{1}{J}\sum_{j=1}^J \left(\frac{1-Z_j}{1-p_j}-1\right)\bigl(\bar Y_j(0)-\mu_{0,J}\bigr)
+o_{\mathbb P}(J^{-1/2}).
\end{equation}
Subtracting \eqref{eq:lin-mu0} from \eqref{eq:lin-mu1} yields
\[
\hat\tau_{\mathrm{reg}}-\tau_{\mathcal C}
=
\frac{1}{J}\sum_{j=1}^J \psi_{j,J}
+o_{\mathbb P}(J^{-1/2}),
\]
where the summands are independent and mean zero:
\[
\psi_{j,J}
:=
\left(\frac{Z_j}{p_j}-1\right)\bigl(\bar Y_j(1)-\mu_{1,J}\bigr)
-
\left(\frac{1-Z_j}{1-p_j}-1\right)\bigl(\bar Y_j(0)-\mu_{0,J}\bigr).
\]
A direct two-point calculation in $Z_j\in\{0,1\}$ gives
\begin{equation}\label{eq:var-psi}
\mathrm{Var}(\psi_{j,J})
=
\left(
\sqrt{\frac{1-p_j}{p_j}}\,(\bar Y_j(1)-\mu_{1,J})
+
\sqrt{\frac{p_j}{1-p_j}}\,(\bar Y_j(0)-\mu_{0,J})
\right)^2
=:v_{j,J}.
\end{equation}
Let $S_J^2:=\sum_{j=1}^J v_{j,J}$ and $\sigma_{\mathrm{reg}}^2(J):=S_J^2/J$.
By Assumption~\ref{ass:weak-null-asymp}(ii), $\mu_{1,J}\to\mu_1$ and $\mu_{0,J}\to\mu_0$, and the empirical measures
$\nu_J\Rightarrow \nu$ imply
\[
\sigma_{\mathrm{reg}}^2(J)\ \to\
\mathbb E_\nu\!\left[
\left(
\sqrt{\frac{1-P}{P}}\,(Y_1-\mu_1)
+
\sqrt{\frac{P}{1-P}}\,(Y_0-\mu_0)
\right)^2
\right]
=:\sigma_{\mathrm{reg}}^2.
\]
Moreover, \eqref{eq:reg-nondeg-proof} assumes $\sigma_{\mathrm{reg}}^2>0$, hence $S_J^2=\Theta(J)$.

To verify Lyapunov's condition, first establish an aggregate fourth-moment
bound for the centered potential outcomes. For $z\in\{0,1\}$, write
\[
\mu_{z,J}:=\frac{1}{J}\sum_{j=1}^J \bar Y_j(z).
\]
Because
$\bar Y_j(z)=\boldsymbol a_n^\top\boldsymbol Y_j(z)$,
the Cauchy--Schwarz inequality gives
\[
|\bar Y_j(z)|
\le
\|\boldsymbol a_n\|_2\,
\|\boldsymbol Y_j(z)\|_2
\le
\|\boldsymbol a_n\|_2 M_j.
\]
Moreover, using $|x-y|^4\le 8(|x|^4+|y|^4)$ and Jensen's inequality,
\begin{equation}\label{eq:centered-fourth-reg}
\begin{aligned}
\sum_{j=1}^J
|\bar Y_j(z)-\mu_{z,J}|^4
&\le
8\sum_{j=1}^J|\bar Y_j(z)|^4
+
8J|\mu_{z,J}|^4\\
&\le
16\sum_{j=1}^J|\bar Y_j(z)|^4\\
&\le
16\|\boldsymbol a_n\|_2^4
\sum_{j=1}^J M_j^4.
\end{aligned}
\end{equation}
Here the second inequality follows because
\[
J|\mu_{z,J}|^4
=
J\left|
\frac{1}{J}\sum_{j=1}^J\bar Y_j(z)
\right|^4
\le
\sum_{j=1}^J|\bar Y_j(z)|^4.
\]

The bound $p_j\in[\underline p,1-\underline p]$ implies that, for some
finite constant $C_{\underline p}$ independent of $j$ and $J$,
\[
|\psi_{j,J}|^4
\le
C_{\underline p}
\left\{
|\bar Y_j(1)-\mu_{1,J}|^4
+
|\bar Y_j(0)-\mu_{0,J}|^4
\right\}.
\]
Consequently, by \eqref{eq:centered-fourth-reg},
\begin{align*}
\sum_{j=1}^J\mathbb E[\psi_{j,J}^4]
&\le
C_{\underline p}
\sum_{z\in\{0,1\}}
\sum_{j=1}^J
|\bar Y_j(z)-\mu_{z,J}|^4\\
&\le
C_{\underline p,n}
\sum_{j=1}^J M_j^4
=
O(J),
\end{align*}
where the last equality follows from
Assumption~\ref{ass:weak-null-asymp}(i). Since
$S_J^2=\Theta(J)$, it follows that
\[
\frac{1}{S_J^4}
\sum_{j=1}^J\mathbb E[\psi_{j,J}^4]
=
\frac{O(J)}{\Theta(J^2)}
=
O(J^{-1})
\longrightarrow 0.
\]
The Lyapunov CLT gives
$\sum_{j=1}^J \psi_{j,J}/S_J \Rightarrow \mathcal N(0,1)$, and since
$\sqrt{J}(\hat\tau_{\mathrm{reg}}-\tau_{\mathcal C})=(\sum_{j=1}^J\psi_{j,J})/\sqrt{J}+o_{\mathbb P}(1)$ with
$S_J^2/J\to\sigma_{\mathrm{reg}}^2$, the claim follows.
\end{proof}

\begin{lemma}[Consistency and conservativeness of $\widehat V_{\mathrm{reg}}$]\label{lem:v-reg-proof}
Under Assumption~\ref{ass:weak-null-asymp} and \eqref{eq:reg-nondeg-proof},
\[
J\,\widehat V_{\mathrm{reg}}\ \xrightarrow{\ \mathbb P\ }\ \sigma_{\mathrm{up,reg}}^2,
\qquad
\text{and}\qquad
\sigma_{\mathrm{reg}}^2\le \sigma_{\mathrm{up,reg}}^2.
\]
\end{lemma}

\begin{proof}
Write $B_1:=\sum_{j=1}^J Z_j/p_j = J B_{1,J}$ and $B_0:=\sum_{j=1}^J (1-Z_j)/(1-p_j)=J B_{0,J}$.
From above, $B_{1,J}\to 1$ and $B_{0,J}\to 1$ in probability.  Using the HC0 formula, we can write
\[
J\,\widehat V_{\mathrm{reg}}
=
\frac{N_{1,J}}{(B_{1,J})^2}+\frac{N_{0,J}}{(B_{0,J})^2},
\]
where
\[
N_{1,J}:=\frac{1}{J}\sum_{j=1}^J \frac{Z_j(\bar Y_j(1)-\hat\mu_1)^2}{p_j^2},
\qquad
N_{0,J}:=\frac{1}{J}\sum_{j=1}^J \frac{(1-Z_j)(\bar Y_j(0)-\hat\mu_0)^2}{(1-p_j)^2}.
\]
Since $(B_{1,J})^{-2}=1+o_{\mathbb P}(1)$ and $(B_{0,J})^{-2}=1+o_{\mathbb P}(1)$, it suffices to prove
$N_{1,J}+N_{0,J}\xrightarrow{\mathbb P}\sigma_{\mathrm{up,reg}}^2$.

Define the deterministic means $\mu_{1,J},\mu_{0,J}$ as above and write $\Delta_1:=\hat\mu_1-\mu_{1,J}$ and
$\Delta_0:=\hat\mu_0-\mu_{0,J}$. By the arm-specific linearizations
\eqref{eq:lin-mu1}--\eqref{eq:lin-mu0} and the corresponding variance
bounds,
\[
\Delta_1=O_{\mathbb P}(J^{-1/2}),
\qquad
\Delta_0=O_{\mathbb P}(J^{-1/2}).
\] 
Expand
\[
(\bar Y_j(1)-\hat\mu_1)^2-(\bar Y_j(1)-\mu_{1,J})^2
=
\Delta_1^2-2(\bar Y_j(1)-\mu_{1,J})\Delta_1.
\]
Therefore
\[
N_{1,J}
=
N_{1,J}^\circ
+
\Delta_1^2\cdot \frac{1}{J}\sum_{j=1}^J \frac{Z_j}{p_j^2}
-
2\Delta_1\cdot \frac{1}{J}\sum_{j=1}^J \frac{Z_j(\bar Y_j(1)-\mu_{1,J})}{p_j^2},
\]
where $N_{1,J}^\circ:=\frac{1}{J}\sum_{j=1}^J \frac{Z_j(\bar Y_j(1)-\mu_{1,J})^2}{p_j^2}$.
Since $\Delta_1^2=O_{\mathbb P}(1/J)$ and $J^{-1}\sum Z_j/p_j^2=O_{\mathbb P}(1)$, the second term is
$o_{\mathbb P}(1)$.  For the third term, $J^{-1}\sum Z_j(\bar Y_j(1)-\mu_{1,J})/p_j^2=O_{\mathbb P}(1)$ by a variance
bound (using Assumption~\ref{ass:weak-null-asymp}(i)), hence it is also $o_{\mathbb P}(1)$ because
$\Delta_1=O_{\mathbb P}(J^{-1/2})$. Thus $N_{1,J}-N_{1,J}^\circ=o_{\mathbb P}(1)$.
The same argument yields $N_{0,J}-N_{0,J}^\circ=o_{\mathbb P}(1)$, where
$N_{0,J}^\circ:=\frac{1}{J}\sum_{j=1}^J \frac{(1-Z_j)(\bar Y_j(0)-\mu_{0,J})^2}{(1-p_j)^2}$.

Next compute expectations:
\[
\mathbb E[N_{1,J}^\circ]
=
\frac{1}{J}\sum_{j=1}^J \frac{(\bar Y_j(1)-\mu_{1,J})^2}{p_j},
\qquad
\mathbb E[N_{0,J}^\circ]
=
\frac{1}{J}\sum_{j=1}^J \frac{(\bar Y_j(0)-\mu_{0,J})^2}{1-p_j}.
\]
Moreover, independence across sessions gives
\begin{align*}
\mathrm{Var}(N_{1,J}^\circ)
&=
\frac{1}{J^2}
\sum_{j=1}^J
\frac{1-p_j}{p_j^3}
\bigl\{\bar Y_j(1)-\mu_{1,J}\bigr\}^4,\\
\mathrm{Var}(N_{0,J}^\circ)
&=
\frac{1}{J^2}
\sum_{j=1}^J
\frac{p_j}{(1-p_j)^3}
\bigl\{\bar Y_j(0)-\mu_{0,J}\bigr\}^4.
\end{align*}
Therefore, using $p_j\in[\underline p,1-\underline p]$ and
\eqref{eq:centered-fourth-reg},
\begin{align*}
\mathrm{Var}(N_{1,J}^\circ)
+
\mathrm{Var}(N_{0,J}^\circ)
&\le
\frac{C_{\underline p}}{J^2}
\sum_{z\in\{0,1\}}
\sum_{j=1}^J
|\bar Y_j(z)-\mu_{z,J}|^4\\
&\le
\frac{C_{\underline p,n}}{J^2}
\sum_{j=1}^J M_j^4\\
&=
O(J^{-1})
\longrightarrow 0.
\end{align*}
Hence
\[
N_{1,J}^\circ+N_{0,J}^\circ
-
\frac{1}{J}\sum_{j=1}^J\left\{\frac{(\bar Y_j(1)-\mu_{1,J})^2}{p_j}+\frac{(\bar Y_j(0)-\mu_{0,J})^2}{1-p_j}\right\}
\ \xrightarrow{\ \mathbb P\ }\ 0.
\]
By Assumption~\ref{ass:weak-null-asymp}(ii) and $\mu_{1,J}\to\mu_1$, $\mu_{0,J}\to\mu_0$, the displayed deterministic
average converges to $\sigma_{\mathrm{up,reg}}^2$ in \eqref{eq:reg-nondeg-proof}. Combining these steps gives
$N_{1,J}+N_{0,J}\xrightarrow{\mathbb P}\sigma_{\mathrm{up,reg}}^2$, and hence
$J\widehat V_{\mathrm{reg}}\xrightarrow{\mathbb P}\sigma_{\mathrm{up,reg}}^2$.

Finally, the pointwise inequality $\sigma_{\mathrm{reg}}^2\le\sigma_{\mathrm{up,reg}}^2$ follows from
\[
\frac{(Y_1-\mu_1)^2}{P}+\frac{(Y_0-\mu_0)^2}{1-P}
-
\left(
\sqrt{\frac{1-P}{P}}\,(Y_1-\mu_1)+\sqrt{\frac{P}{1-P}}\,(Y_0-\mu_0)
\right)^2
=
\bigl\{(Y_1-\mu_1)-(Y_0-\mu_0)\bigr\}^2\ge 0,
\]
and taking expectations under $\nu$.
\end{proof}

\paragraph{Consequence for $T_{\mathrm{reg}}$ under $H_0^{\cal U}$.}
Under $H_0^{\cal U}$, $\tau_{\mathcal C}=0$ by definition, so
Lemma~\ref{lem:samp-clt-reg-proof} and Lemma~\ref{lem:v-reg-proof} imply by Slutsky that
\[
T_{\mathrm{reg}}
=
\frac{\hat\tau_{\mathrm{reg}}}{\sqrt{\widehat V_{\mathrm{reg}}}}
=
\frac{\sqrt{J}\,\hat\tau_{\mathrm{reg}}}{\sqrt{J\,\widehat V_{\mathrm{reg}}}}
\ \Rightarrow\
\mathcal N\!\left(0,\ \frac{\sigma_{\mathrm{reg}}^2}{\sigma_{\mathrm{up,reg}}^2}\right),
\qquad
\frac{\sigma_{\mathrm{reg}}^2}{\sigma_{\mathrm{up,reg}}^2}\le 1.
\]
In particular,
\begin{equation}\label{eq:tail-z}
\limsup_{J\to\infty}\ \mathbb P\!\left(T_{\mathrm{reg}}\ge z_{1-\alpha}\right)\ \le\ \alpha.
\end{equation}

\subsubsection*{Step 3: conclude Corollary~\ref{cor:weak-null-valid-reg}}

\begin{proof}[Proof of Corollary~\ref{cor:weak-null-valid-reg}]
Fix $\alpha\in(0,1/2]$, so $z_{1-\alpha}\ge0$.
Let $c_{1-\alpha}$ be the $(1-\alpha)$ quantile of the conditional randomization distribution of
$T_{\mathrm{reg}}^\ast$ given $\{\bar Y_j^{\mathrm{obs}}\}_{j=1}^J$. By Lemma~\ref{lem:rand-clt-reg-proof},
$c_{1-\alpha}\to z_{1-\alpha}$ in probability.

For the corresponding one-sided randomization $p$-value, the
lower-quantile convention gives
\[
\{\hat p_{\mathrm{reg}}\le\alpha\}
\subseteq
\{T_{\mathrm{reg}}\ge c_{1-\alpha}\}.
\]
The inclusion may be strict because of discreteness and ties in the
finite randomization distribution. It therefore suffices to bound the
probability of the latter event. Fix $\varepsilon>0$. Then
\[
\mathbb P\!\left(T_{\mathrm{reg}}\ge c_{1-\alpha}\right)
\le
\mathbb P\!\left(T_{\mathrm{reg}}\ge z_{1-\alpha}-\varepsilon\right)
+
\mathbb P\!\left(|c_{1-\alpha}-z_{1-\alpha}|>\varepsilon\right).
\]
Taking $\limsup$ and using $c_{1-\alpha}\to z_{1-\alpha}$ in probability gives
\[
\limsup_{J\to\infty}\ \mathbb P\!\left(T_{\mathrm{reg}}\ge c_{1-\alpha}\right)
\le
\limsup_{J\to\infty}\ \mathbb P\!\left(T_{\mathrm{reg}}\ge z_{1-\alpha}-\varepsilon\right).
\]
By the weak convergence of $T_{\mathrm{reg}}$ established above, the right-hand side equals
$1-\Phi\!\big((z_{1-\alpha}-\varepsilon)/\rho\big)$ where
$\rho^2=\sigma_{\mathrm{reg}}^2/\sigma_{\mathrm{up,reg}}^2\le 1$.
Letting $\varepsilon\downarrow 0$ and using continuity of $\Phi$ yields
\[
\limsup_{J\to\infty}\ \mathbb P\!\left(T_{\mathrm{reg}}\ge c_{1-\alpha}\right)
\le
1-\Phi\!\left(\frac{z_{1-\alpha}}{\rho}\right)
\le
1-\Phi(z_{1-\alpha})
=
\alpha,
\]
which proves
\[
\limsup_{J\to\infty}
\mathbb P(\hat p_{\mathrm{reg}}\le\alpha)
\le\alpha
\]
under $H_0^{\cal U}$. When $B_J\to\infty$, the plus-one Monte Carlo
$p$-value differs from the full-reference $p$-value by
$o_{\mathbb P}(1)$. The same $\varepsilon$-enlargement argument used
in the proof of Theorem~\ref{thm:weak-null-general} transfers the
result for $\alpha<1/2$. At the endpoint $\alpha=1/2$, the conclusion
follows directly from convergence of the conditional median to zero
and the centered normal limit of $T_{\mathrm{reg}}$.
\end{proof}

\section{Proofs and supporting derivations for the power approximations}\label{app:power}

This section proves Propositions~\ref{prop:power-total} and~\ref{prop:power-carry}. The assignment law, distributed-lag model, AR(1) process, variance formula, noncentrality parameters, and power approximations are stated in Section~\ref{sec:power} of the main paper.
Throughout, we work with the predetermined pooled-section construction:
$T=ML$, pooled size $r\mid M$, $J=M/r$, predetermined pooled sections
$[s_j,e_j]=\{(j-1)rL+1,\dots,jrL\}$, and the constancy indicators
$E_j=\{W^{((j-1)r+1)}=\cdots=W^{(jr)}\}$.
For the total-effect test, the selected set is $\mathcal J(\mathbf W)=\{j:\,E_j\}$ with
$J_{\mathrm{tot}}=|\mathcal J(\mathbf W)|$.
We write $p=p(r;q)=q^r/(q^r+(1-q)^r)$ and $\pi_r(q)=q^r+(1-q)^r$, and $n=rL-m$.

\subsection{Conditional assignment law under predetermined pooling}\label{app:pool-pre-law}

\begin{lemma}[Constancy probability and pooled-section count]\label{lem:J-binomial}
Under~\eqref{eq:design-constq}, the constancy indicators $\mathbf 1\{E_j\}$ are i.i.d.\ Bernoulli with
\[
\mathbb{P}(E_j)=\pi_r(q)=q^r+(1-q)^r,
\]
so $J_{\mathrm{tot}}\sim\mathrm{Binomial}(J,\pi_r(q))$ and $\mathbb E[J_{\mathrm{tot}}]=J\pi_r(q)$.
Moreover, $J_{\mathrm{tot}}/J\to \pi_r(q)$ almost surely, and hence $J_{\mathrm{tot}}\to\infty$ almost surely as $J\to\infty$.
\end{lemma}

\begin{proof}
Each pooled section $j$ depends on the disjoint block set $\{W^{((j-1)r+1)},\dots,W^{(jr)}\}$.
The event $E_j$ occurs iff all $r$ blocks equal $1$ or all equal $0$, hence
$\mathbb P(E_j)=q^r+(1-q)^r=\pi_r(q)$. Disjointness implies independence across $j$, so $J_{\mathrm{tot}}$ is binomial.
The almost sure limit follows from the strong law of large numbers applied to $\sum_{j=1}^J \mathbf 1\{E_j\}$.
\end{proof}

\begin{lemma}[Conditional pooled-section law]\label{lem:pool-pre-law}
Conditional on the realized constancy indicators, the labels of the constant pooled sections are mutually independent Bernoulli variables with success probability $p(r;q)$.
\end{lemma}

\begin{proof}
Fix pooled section $j$ and define
$A_j:=\{W^{((j-1)r+1)}=\cdots=W^{(jr)}=1\}$ and
$B_j:=\{W^{((j-1)r+1)}=\cdots=W^{(jr)}=0\}$.
Then $E_j=A_j\dot\cup B_j$ and by independence,
$\mathbb P(A_j)=q^r$, $\mathbb P(B_j)=(1-q)^r$, $\mathbb P(E_j)=q^r+(1-q)^r$.
Let $Z_j$ be the common value on $E_j$. Then
\[
\mathbb P(Z_j=1\mid E_j)=\frac{\mathbb P(A_j)}{\mathbb P(E_j)}
=\frac{q^r}{q^r+(1-q)^r}=p(r;q).
\]
Independence across disjoint sections follows because the underlying block sets are disjoint across $j$.
\end{proof}

\subsection{Proof of Proposition~\ref{prop:power-total}}\label{app:power-total}

We work on a product probability space $(\Omega,\mathcal F,\mathbb P)$ supporting two independent sources of randomness:
(i) the assignment mechanism $\mathbf W$ and (ii) the superpopulation error process $\{\varepsilon_t\}$ in
\eqref{eq:ar1-main}. Formally, let
\[
(\Omega,\mathcal F,\mathbb P)
=
(\Omega_\varepsilon\times\Omega_W,\ \mathcal F_\varepsilon\otimes\mathcal F_W,\ \mathbb P_\varepsilon\otimes\mathbb P_W),
\]
where $\mathbb P_W$ generates the block assignments in~\eqref{eq:design-constq} and $\mathbb P_\varepsilon$ generates the
stationary AR(1) process~\eqref{eq:ar1-main}.

Throughout, we condition on the realized pooled-section structure used by the CRT, i.e., on $\mathcal J(\mathbf W)$ and
hence on $J_{\mathrm{tot}}=|\mathcal J(\mathbf W)|$.

\paragraph{Notation.}
For each selected pooled section $j\in\mathcal J(\mathbf W)$, let $\mathcal U_j$ denote its $n=rL-m$ focal time points and define
\[
\bar Y_j(z)
:=
\frac{1}{n}\sum_{t\in\mathcal U_j} Y_t(\mathbf z_T),
\qquad z\in\{0,1\},
\qquad
\bar Y_j^{\mathrm{obs}}=\bar Y_j(Z_j),
\]
where $Z_j\in\{0,1\}$ is the pooled-section label (constant on the section when $j\in\mathcal J(\mathbf W)$).
Conditional on $\mathcal J(\mathbf W)$, the labels $\{Z_j\}_{j\in\mathcal J(\mathbf W)}$ are independent
$\mathrm{Bernoulli}(p)$ with $p=p(r;q)$ (Lemma~\ref{lem:pool-pre-law}).

Define the finite-population pooled-section estimand
\begin{equation}\label{eq:tauJ-def}
\tau_{J_{\mathrm{tot}}}
:=
\frac{1}{J_{\mathrm{tot}}}\sum_{j\in\mathcal J(\mathbf W)}\{\bar Y_j(1)-\bar Y_j(0)\}.
\end{equation}

\subsubsection*{Step 1: Identify the finite-population target under the DGP}

Assume $m\ge m_0$. Fix any selected pooled section $j$ and any focal time $t\in\mathcal U_j$.
By construction of $\mathcal U_j$, the last $m_0$ lags of $t$ remain within the same pooled section, and because the
section is constant in time under $\mathbf z_T$ we have $w_{t-\ell}=z$ for all $\ell\in\{0,\dots,m_0\}$.
Therefore, under~\eqref{eq:dgp-main},
\[
Y_t(\mathbf 0_T)=\mu+\varepsilon_t,
\qquad
Y_t(\mathbf 1_T)=\mu+\tau_{\mathrm{tot}}+\varepsilon_t,
\qquad
\tau_{\mathrm{tot}}=\sum_{\ell=0}^{m_0}\beta_\ell.
\]
Averaging over $t\in\mathcal U_j$ yields
\begin{equation}\label{eq:barY-model}
\bar Y_j(0)=\mu+\bar\varepsilon_j,
\qquad
\bar Y_j(1)=\mu+\tau_{\mathrm{tot}}+\bar\varepsilon_j,
\qquad
\bar\varepsilon_j:=\frac{1}{n}\sum_{t\in\mathcal U_j}\varepsilon_t.
\end{equation}
In particular, $\bar Y_j(1)-\bar Y_j(0)=\tau_{\mathrm{tot}}$ for every selected section $j$, and hence
\begin{equation}\label{eq:tauJ-equals-tau}
\tau_{J_{\mathrm{tot}}}=\tau_{\mathrm{tot}}\qquad\text{deterministically.}
\end{equation}

\subsubsection*{Step 2: Conditional (randomization) CLT given realized potential outcomes}

Let $\mathcal H$ be the $\sigma$-field generated by (i) the realized potential outcomes
$\{(\bar Y_j(1),\bar Y_j(0))\}_{j\in\mathcal J(\mathbf W)}$ and (ii) the realized pooled-section structure $\mathcal J(\mathbf W)$.
Conditional on $\mathcal H$, the only randomness in $(\hat\tau_{\mathrm{HT}},\widehat V_{\mathrm{up}})$ comes from the independent Bernoulli labels $\{Z_j\}_{j\in\mathcal J(\mathbf W)}$. Define $T_{\mathrm{tot}}:=\hat\tau_{\mathrm{HT}}/\sqrt{\widehat V_{\mathrm{up}}}$.

Write the section-level HT estimator as
\[
\hat\tau_{\mathrm{HT}}
=
\frac{1}{J_{\mathrm{tot}}}\sum_{j\in\mathcal J(\mathbf W)} \xi_{j,J_{\mathrm{tot}}},
\qquad
\xi_{j,J_{\mathrm{tot}}}
:=
\frac{Z_j\bar Y_j(1)}{p}-\frac{(1-Z_j)\bar Y_j(0)}{1-p}.
\]
Then, conditional on $\mathcal H$, the summands $\{\xi_{j,J_{\mathrm{tot}}}\}$ are independent with
$\mathbb E[\xi_{j,J_{\mathrm{tot}}}\mid \mathcal H]=\bar Y_j(1)-\bar Y_j(0)=\tau_{\mathrm{tot}}$, so
$\mathbb E[\hat\tau_{\mathrm{HT}}\mid\mathcal H]=\tau_{\mathrm{tot}}$.

Define centered summands
\[
\psi_{j,J_{\mathrm{tot}}}
:=
\xi_{j,J_{\mathrm{tot}}}-\{\bar Y_j(1)-\bar Y_j(0)\}.
\]
A direct computation gives
\[
\operatorname{Var}\!\left(\psi_{j,J_{\mathrm{tot}}}\mid\mathcal H\right)
=\frac{1-p}{p}\bar Y_j(1)^2+\frac{p}{1-p}\bar Y_j(0)^2
+2\bar Y_j(1)\bar Y_j(0)
=\left\{\sqrt{\frac{1-p}{p}}\bar Y_j(1)
+\sqrt{\frac{p}{1-p}}\bar Y_j(0)\right\}^2.
\]

Let
\[
S_{J_{\mathrm{tot}}}^2
:=
\sum_{j\in\mathcal J(\mathbf W)}\mathrm{Var}(\psi_{j,J_{\mathrm{tot}}}\mid\mathcal H),
\qquad
\sigma_{\tau}^2(J_{\mathrm{tot}})
:=
\frac{S_{J_{\mathrm{tot}}}^2}{J_{\mathrm{tot}}}.
\]

\paragraph{Conditional CLT for $\hat\tau_{\mathrm{HT}}$.}
Because $Z_j$ are independent given $\mathcal H$ and the section lengths are fixed, a Lyapunov condition holds under finite
fourth moments of $\bar Y_j(1)$ and $\bar Y_j(0)$, implying the conditional CLT
\begin{equation}\label{eq:cond-clt-HT}
\frac{\sqrt{J_{\mathrm{tot}}}\,(\hat\tau_{\mathrm{HT}}-\tau_{\mathrm{tot}})}{\sqrt{\sigma_{\tau}^2(J_{\mathrm{tot}})}}
\ \Rightarrow\ \mathcal N(0,1),
\qquad\text{conditional on }\mathcal H,
\end{equation}
where the convergence holds in $\mathbb P$-probability (and in fact almost surely along typical realizations under the
superpopulation model).

\paragraph{Conditional LLN for $\widehat V_{\mathrm{up}}$.}
Recall
\[
\widehat V_{\mathrm{up}}
=
\frac{1}{J_{\mathrm{tot}}^2}\sum_{j\in\mathcal J(\mathbf W)}(\bar Y_j^{\mathrm{obs}})^2
\left(\frac{Z_j}{p^2}+\frac{1-Z_j}{(1-p)^2}\right),
\qquad
\bar Y_j^{\mathrm{obs}}=\bar Y_j(Z_j).
\]
Define the finite-population upper-bound functional
\begin{equation}\label{eq:vupJ-def}
v_{\mathrm{up}}(J_{\mathrm{tot}})
:=
\frac{1}{J_{\mathrm{tot}}}\sum_{j\in\mathcal J(\mathbf W)}
\left\{\frac{\bar Y_j(1)^2}{p}+\frac{\bar Y_j(0)^2}{1-p}\right\}.
\end{equation}
Then a conditional law of large numbers yields
\begin{equation}\label{eq:cond-LLN-Vup}
J_{\mathrm{tot}}\,\widehat V_{\mathrm{up}}\ \xrightarrow{\ \mathbb P\ }\ v_{\mathrm{up}}(J_{\mathrm{tot}}),
\qquad\text{conditional on }\mathcal H.
\end{equation}

\paragraph{Conditional limit for $T_{\mathrm{tot}}$.}
Combining \eqref{eq:cond-clt-HT} and \eqref{eq:cond-LLN-Vup} with Slutsky's theorem (conditionally on $\mathcal H$) gives
\begin{equation}\label{eq:cond-limit-Ttot}
\frac{
T_{\mathrm{tot}}
-\tau_{\mathrm{tot}}\sqrt{J_{\mathrm{tot}}}/\sqrt{v_{\mathrm{up}}(J_{\mathrm{tot}})}
}{
\{\sigma_{\tau}^2(J_{\mathrm{tot}})/v_{\mathrm{up}}(J_{\mathrm{tot}})\}^{1/2}
}
\ \Rightarrow\ \mathcal N(0,1),
\qquad \text{conditional on }\mathcal H.
\end{equation}

\subsubsection*{Step 3: Superpopulation limits of random variance functionals}

We now show that the random finite-population functionals $\sigma_{\tau}^2(J_{\mathrm{tot}})$ and $v_{\mathrm{up}}(J_{\mathrm{tot}})$
converge in probability to deterministic limits under the superpopulation model (indeed, almost surely).

\begin{lemma}[Ergodic ratio LLN under independent thinning]\label{lem:ratio-LLN}
Let $\{X_j\}_{j\ge 1}$ be a stationary ergodic sequence with $\mathbb E|X_1|<\infty$, and let $\{E_j\}_{j\ge 1}$ be i.i.d.\
$\mathrm{Bernoulli}(\pi)$, independent of $\{X_j\}$, with $\pi\in(0,1)$.
Define $J_N:=\sum_{j=1}^N E_j$. Then on the event $\{J_N\to\infty\}$,
\[
\frac{1}{J_N}\sum_{j=1}^N E_j X_j\ \xrightarrow{\ a.s.\ }\ \mathbb E[X_1].
\]
\end{lemma}

\begin{proof}
The independent Bernoulli sequence is mixing, and its product with the stationary ergodic process $\{X_j\}$ is stationary and ergodic. Therefore $\frac{1}{N}\sum_{j=1}^N E_j X_j \to \mathbb E[E_1X_1]=\pi\,\mathbb E[X_1]$ almost surely, while $\frac{1}{N}\sum_{j=1}^N E_j\to\pi$ almost surely. Taking ratios yields the claim.
\end{proof}

\paragraph{Apply Lemma~\ref{lem:ratio-LLN} to section means.}
Let $\bar\varepsilon_j$ be the within-section focal mean error in \eqref{eq:barY-model}, computed on the predetermined pooled partition.
The sequence $\{\bar\varepsilon_j\}$ is stationary ergodic under the stationary AR(1) model and has finite moments.
The selection indicators $\{E_j\}$ are i.i.d.\ Bernoulli with $\pi=\pi_r(q)>0$ and are independent of $\{\bar\varepsilon_j\}$.
Therefore, by Lemma~\ref{lem:ratio-LLN} and Lemma~\ref{lem:J-binomial},
\[
\frac{1}{J_{\mathrm{tot}}}\sum_{j\in\mathcal J(\mathbf W)} \bar\varepsilon_j
\ \xrightarrow{\ a.s.\ }\ 0,
\qquad
\frac{1}{J_{\mathrm{tot}}}\sum_{j\in\mathcal J(\mathbf W)} \bar\varepsilon_j^2
\ \xrightarrow{\ a.s.\ }\ \mathbb E[\bar\varepsilon_1^2]=\sigma_{\bar\varepsilon}^2(n,\rho),
\]
and similarly for higher moments. These limits also hold in conditional probability given
$J_{\mathrm{tot}}$ along the typical sequences considered in the
proposition. Indeed, because the indicators $\{E_j\}$ are i.i.d.\ and
independent of the error process, conditional on
$J_{\mathrm{tot}}=k$ the selected index set is uniform over the
subsets of size $k$ and remains independent of
$\{\bar\varepsilon_j\}$. A standard sampling-without-replacement law
of large numbers therefore gives the same ratio limits whenever
$k/J\to\pi_r(q)$. Thus the conditional approximation given
$\mathcal H$ may be averaged over $\mathcal H$ given
$J_{\mathrm{tot}}$, yielding the count-conditioned formulation in
Proposition~\ref{prop:power-total}.

\paragraph{Limits of $\sigma_{\tau}^2(J_{\mathrm{tot}})$ and $v_{\mathrm{up}}(J_{\mathrm{tot}})$.}
Using \eqref{eq:barY-model} and the above ratio LLNs, we obtain (in probability, indeed a.s.)
\[
\frac{1}{J_{\mathrm{tot}}}\sum_{j\in\mathcal J(\mathbf W)} \bar Y_j(0)^2
\ \to\ \mu^2+\sigma_{\bar\varepsilon}^2(n,\rho),
\qquad
\frac{1}{J_{\mathrm{tot}}}\sum_{j\in\mathcal J(\mathbf W)} \bar Y_j(1)^2
\ \to\ (\mu+\tau_{\mathrm{tot}})^2+\sigma_{\bar\varepsilon}^2(n,\rho),
\]
and
\[
\frac{1}{J_{\mathrm{tot}}}\sum_{j\in\mathcal J(\mathbf W)} \bar Y_j(1)\bar Y_j(0)
\ \to\ \mu(\mu+\tau_{\mathrm{tot}})+\sigma_{\bar\varepsilon}^2(n,\rho).
\]
Substituting these limits yields the deterministic quantities used in
Proposition~\ref{prop:power-total}.

\subsubsection*{Step 4: Local alternatives and power}

Let $\tau_{\mathrm{tot}}=\tau_{\mathrm{tot},J}$ vary along a triangular array such that
$\sqrt{J}\tau_{\mathrm{tot},J}$ is bounded. Because
$J_{\mathrm{tot}}/J\to\pi_r(q)>0$ almost surely, this is equivalent to
$\sqrt{J_{\mathrm{tot}}}\tau_{\mathrm{tot},J}=O_{\mathbb P}(1)$. Then
$\tau_{\mathrm{tot},J}\to0$. The ratio laws in Step 3 apply along this array and give
\[
v_{\mathrm{up}}(J_{\mathrm{tot}})
-v_{\mathrm{up}}\{\tau_{\mathrm{tot},J}\}=o_{\mathbb P}(1),
\qquad
\sigma_\tau^2(J_{\mathrm{tot}})
-\sigma_\tau^2\{\tau_{\mathrm{tot},J}\}=o_{\mathbb P}(1),
\]
where the deterministic functions on the right are obtained by substituting
\eqref{eq:barY-model} into~\eqref{eq:vupJ-def} and the conditional variance in
\eqref{eq:cond-clt-HT}. Because the local noncentrality
$\sqrt{J_{\mathrm{tot}}}\tau_{\mathrm{tot},J}$ is bounded, replacing the random
variance functionals in~\eqref{eq:cond-limit-Ttot} by these deterministic limits
changes the centering and scaling by $o_{\mathbb P}(1)$. Direct simplification gives
the quantities $\mu_{\mathrm{tot}}(J_{\mathrm{tot}})$ and
$\sigma_{\mathrm{tot}}^2$ in Proposition~\ref{prop:power-total}. Hence
\[
\frac{T_{\mathrm{tot}}-\mu_{\mathrm{tot}}(J_{\mathrm{tot}})}
{\sigma_{\mathrm{tot}}}
\ \Rightarrow\ \mathcal N(0,1).
\]

\paragraph{Randomization critical value and power.}
Let $T_{\mathrm{tot}}^\ast$ be the statistic recomputed after resampling
$Z_j^\ast\stackrel{ind}{\sim}\mathrm{Bernoulli}(p)$ on $\mathcal J(\mathbf W)$ while holding
$\{\bar Y_j^{\mathrm{obs}}\}_{j\in\mathcal J(\mathbf W)}$ fixed.
Conditional on $\{\bar Y_j^{\mathrm{obs}}\}$, a randomization CLT yields $T_{\mathrm{tot}}^\ast\Rightarrow\mathcal N(0,1)$, so the
conditional $(1-\alpha)$ quantile $c_{1-\alpha}$ satisfies $c_{1-\alpha}\to z_{1-\alpha}$ in probability.
Rejecting when $T_{\mathrm{tot}}\ge c_{1-\alpha}$ yields~\eqref{eq:power-total}
under the stated local alternatives. Under a fixed positive total effect,
$\mu_{\mathrm{tot}}(J_{\mathrm{tot}})\to\infty$, so power converges to one; the
same display is then a finite-$J$ Gaussian approximation rather than a
nondegenerate limiting power function. \qed

\subsection{Tail-signal derivation and proof of Proposition~\ref{prop:power-carry}}\label{app:power-carry}

We analyze the paired predetermined pooled-section statistic used in~\eqref{eq:carry-stat}.
Let $J_e=\lfloor J/2\rfloor$ and index adjacent odd--even pooled-section pairs by $j=1,\dots,J_e$.
For the final-block label $Z_{2j-1}$ of odd section $2j-1$ and the focal mean $\bar Y_{2j}^{\mathrm{obs}}$ in even section $2j$, define
\begin{align}
\widehat\delta_m
&=\frac{1}{J_e}\sum_{j=1}^{J_e}
\left\{\frac{Z_{2j-1}\bar Y_{2j}^{\mathrm{obs}}}{q}
-\frac{(1-Z_{2j-1})\bar Y_{2j}^{\mathrm{obs}}}{1-q}\right\},\label{eq:ht-carry-power-tech}\\
\widehat V_{\mathrm{up}}^{(m)}
&=\frac{1}{J_e^2}\sum_{j=1}^{J_e}(\bar Y_{2j}^{\mathrm{obs}})^2
\left\{\frac{Z_{2j-1}}{q^2}+\frac{1-Z_{2j-1}}{(1-q)^2}\right\}.
\label{eq:vup-carry-power-tech}
\end{align}
Define the studentized statistic
\begin{equation}\label{eq:tstud-carry-power-tech}
T_m^{\mathrm{stud}}=\widehat\delta_m/\sqrt{\widehat V_{\mathrm{up}}^{(m)}}.
\end{equation}

\paragraph{Step 1: Tail-signal under $m_0\le rL$ (last-block label).}
Fix an even pooled section $2j$ with start time $s$ and consider a focal time $t=s+a$ with offset
$a\in\{m,m+1,\dots,rL-1\}$. The randomized label in~\eqref{eq:ht-carry-power-tech} is the assignment on the \emph{last design
block} of the preceding odd pooled section, which applies to the $L$ time points $\{s-L,\dots,s-1\}$.
For a given lag $\ell\in\{0,\dots,m_0\}$, the lagged exposure $w_{t-\ell}$ equals this label if and only if
\[
t-\ell\in\{s-L,\dots,s-1\}
\quad\Longleftrightarrow\quad
s-L\le s+a-\ell\le s-1
\quad\Longleftrightarrow\quad
a+1\le \ell\le a+L.
\]
Fix $\ell>m$. The number of focal offsets $a\in\{m,\dots,rL-1\}$ such that $a+1\le \ell\le a+L$ is
\[
w_\ell
:=
\left|\,[m,rL-1]\cap[\ell-L,\ell-1]\,\right|.
\]
Under the condition $m_0\le rL$ in Proposition~\ref{prop:power-carry}, the upper truncation does not bind and
\[
w_\ell=\ell-\max\{m,\ell-L\}=\min\{L,\ell-m\},
\qquad \ell=m+1,\dots,m_0.
\]
Therefore, toggling the last-block label from $0$ to $1$ shifts the even-section focal mean by
\[
\bar Y_{2j}(1)-\bar Y_{2j}(0)
=
\frac{1}{n}\sum_{\ell=m+1}^{m_0} w_\ell\,\beta_\ell
=
\frac{1}{n}\sum_{\ell=m+1}^{m_0}\min\{L,\ell-m\}\,\beta_\ell
=:\delta_m(n),
\qquad n=rL-m,
\]
which is~\eqref{eq:delta-tail-main}. Here $\bar Y_{2j}(z)$ denotes the even-section focal mean under the
counterfactual intervention setting the last design block of the preceding odd pooled section to $z$, with all other
block assignments and the error process held fixed.

\paragraph{Step 2: Conditional CLT for the HT estimator.}
Let $\mathcal H_e$ be the $\sigma$-field generated by the collection of even-section potential means
$\{(\bar Y_{2j}(1),\bar Y_{2j}(0))\}_{j=1}^{J_e}$ (and any additional randomness unrelated to the last-block labels,
such as the remaining block assignments and the error process).
Under~\eqref{eq:design-constq}, the last-block labels $\{Z_{2j-1}\}$ are independent across $j$ with
$Z_{2j-1}\sim\mathrm{Bernoulli}(q)$, and are independent of $\mathcal H_e$.

Define the per-pair HT summand
\[
\xi_{j,J_e}
:=
\frac{Z_{2j-1}\,\bar Y_{2j}(1)}{q}
-
\frac{(1-Z_{2j-1})\,\bar Y_{2j}(0)}{1-q}.
\]
Then $\hat\delta_m=J_e^{-1}\sum_{j=1}^{J_e}\xi_{j,J_e}$ and, conditional on $\mathcal H_e$,
\[
\mathbb E[\xi_{j,J_e}\mid \mathcal H_e]=\bar Y_{2j}(1)-\bar Y_{2j}(0)=\delta_m(n).
\]
Moreover, conditional on $\mathcal H_e$,
\[
\mathrm{Var}(\xi_{j,J_e}\mid \mathcal H_e)
=
\left(\frac{1}{q}-1\right)\bar Y_{2j}(1)^2
+\left(\frac{1}{1-q}-1\right)\bar Y_{2j}(0)^2
+2\,\bar Y_{2j}(1)\bar Y_{2j}(0).
\]
Let $\sigma_\delta^2(J_e):=J_e^{-1}\sum_{j=1}^{J_e}\mathrm{Var}(\xi_{j,J_e}\mid \mathcal H_e)$.
Under fixed $n$ and the AR(1) model, a conditional Lyapunov CLT yields
\[
\frac{\sqrt{J_e}\big(\hat\delta_m-\delta_m(n)\big)}{\sqrt{\sigma_\delta^2(J_e)}}\ \Rightarrow\ \mathcal N(0,1),
\qquad\text{conditional on }\mathcal H_e,
\]
with convergence in $\mathbb P$-probability.

\paragraph{Step 3: Studentization.}
Define the finite-population upper-bound functional
\[
v_{\mathrm{up}}^{(m)}(J_e)
:=
\frac{1}{J_e}\sum_{j=1}^{J_e}\left\{\frac{\bar Y_{2j}(1)^2}{q}+\frac{\bar Y_{2j}(0)^2}{1-q}\right\}.
\]
A conditional LLN gives
\[
J_e\,\widehat V_{\mathrm{up}}^{(m)}\ \xrightarrow{\ \mathbb P\ }\ v_{\mathrm{up}}^{(m)}(J_e),
\qquad\text{conditional on }\mathcal H_e,
\]
where $\widehat V_{\mathrm{up}}^{(m)}$ is defined in~\eqref{eq:vup-carry-power-tech}. Therefore, by Slutsky's theorem,
\[
\frac{
T_m^{\mathrm{stud}}-\delta_m(n)\sqrt{J_e}/\sqrt{v_{\mathrm{up}}^{(m)}(J_e)}
}{
\{\sigma_\delta^2(J_e)/v_{\mathrm{up}}^{(m)}(J_e)\}^{1/2}
}
\Rightarrow \mathcal N(0,1),
\qquad\text{conditional on }\mathcal H_e.
\]

\paragraph{Step 4: Deterministic limits and power.}
Along the triangular array in Proposition~\ref{prop:power-carry},
$\boldsymbol\beta_J\to\boldsymbol\beta_\infty$. Because
$(L,r,m,m_0)$ are fixed, the even-section focal means have uniformly
bounded fourth moments and converge to their stationary counterparts
under $\boldsymbol\beta_\infty$. A standard triangular-array version
of the ergodic law of large numbers therefore gives
\[
v_{\mathrm{up}}^{(m)}(J_e)
\ \xrightarrow{\ \mathbb P\ }\
v_{\mathrm{up}}^{(m)},
\qquad
\sigma_\delta^2(J_e)
\ \xrightarrow{\ \mathbb P\ }\
\sigma_\delta^2,
\]
where $v_{\mathrm{up}}^{(m)}$ and $\sigma_\delta^2$ are the
corresponding stationary expectations evaluated at
$\boldsymbol\beta_\infty$.

Along the local alternatives in Proposition~\ref{prop:power-carry},
$\sqrt{J_e}\delta_m(n)$ is bounded. Hence replacing the random variance
functionals in the preceding standardized limit by their deterministic limits
changes its centering and scaling by $o_{\mathbb P}(1)$. This yields
\[
\frac{T_m^{\mathrm{stud}}-\mu_m(J)}{\sigma_m}\Rightarrow\mathcal N(0,1),
\]
with $\mu_m(J)$ and $\sigma_m$ defined in~\eqref{eq:mu-sigma-carry}.

Under the CRT resampling scheme, the final labels of the odd pooled
sections are redrawn independently from $\mathrm{Bernoulli}(q)$ while
holding the even-section outcomes fixed. Conditional on those outcomes,
the corresponding resampled statistic satisfies
\[
(T_m^{\mathrm{stud}})^\ast\Rightarrow\mathcal N(0,1),
\]
so the $(1-\alpha)$ critical value is asymptotically
$z_{1-\alpha}$. This reference-distribution limit does not require
$\delta_m(n)=0$.
Substituting the above normal approximation for $T_m^{\mathrm{stud}}$ yields~\eqref{eq:power-carry}.
For a fixed positive tail signal, $\mu_m(J)\to\infty$ and power converges to one;
the same formula is then interpreted as a finite-$J$ Gaussian approximation. \qed

\section{Implementation of the randomization procedures}
\label{app-online:implementation}

\subsection{Design representation and analytic history probabilities}

The implementation represents a switchback design by a contiguous partition of
$[T]$ into assignment blocks $B_1,\ldots,B_K$ and block-specific treatment
probabilities $q_1,\ldots,q_K$. A draw first samples independent labels
$Z_k\sim\mathrm{Bernoulli}(q_k)$ and then sets $W_t=Z_k$ for every $t\in B_k$.
Before any analysis, the code verifies that the observed path is constant inside
every declared design block. Failure of this check indicates that the outcome
file and the design object do not describe the same experiment.

The probabilities of an all-treatment or all-control local history are computed
analytically. Let $a_t$ and $b_t$ be the indices of the first and last assignment
blocks intersecting $[t-m,t]$. Then
\begin{equation}
\label{eq:online-exact-pi}
\pi_{1t}=\pr(W_{t-m:t}=\one_{m+1})=\prod_{k=a_t}^{b_t}q_k,
\qquad
\pi_{0t}=\pr(W_{t-m:t}=\zero_{m+1})=\prod_{k=a_t}^{b_t}(1-q_k).
\end{equation}
The code evaluates these products on the log scale using cumulative sums. Thus
no preliminary simulation is needed to estimate $\pi_{1t}$ or $\pi_{0t}$. This
replaces the earlier implementation that repeatedly drew treatment paths to
approximate the same design probabilities.

\subsection{Predetermined pooling and the section CRT}

For a proposed horizon $m$, the default rule greedily pools consecutive design
blocks until a section has length at least $m+1$, closes that section, and
continues with the next unused block. The rule depends only on the design and
$m$. In particular, it is valid when $m$ exceeds the operational block length
$L$: several assignment blocks are then pooled into one predetermined section.
This does not force the treatment to be constant for $m+1$ periods. A pooled
section is retained only if its independently randomized block labels happen to
be identical in the realized path.

The implementation of Procedure~\ref{proc:section-crt} is:
\begin{enumerate}
\item Construct the predetermined sections and retain the sections that are
constant under $\Wobs$.
\item In each retained section $[s_j,e_j]$, discard $s_j,\ldots,s_j+m-1$ and
form the focal mean over $s_j+m,\ldots,e_j$.
\item If retained section $j$ contains design blocks $a_j,\ldots,b_j$, compute
\begin{equation}
\label{eq:online-pj}
p_j^\star=
\frac{\prod_{k=a_j}^{b_j}q_k}
{\prod_{k=a_j}^{b_j}q_k+\prod_{k=a_j}^{b_j}(1-q_k)}.
\end{equation}
\item Draw the retained-section labels independently from
$\mathrm{Bernoulli}(p_j^\star)$, hold the assignments outside the retained
sections fixed, and recompute the statistic in~\eqref{eq:section-ht}.
\item For $B$ randomization draws, report the plus-one p-value
\begin{equation}
\label{eq:online-plus-one}
\widehat p=
\frac{1+\sum_{b=1}^B\ind{T^{(b)}\ge T^{\mathrm{obs}}}}{B+1}
\end{equation}
for the upper-tailed alternative.
\end{enumerate}
The code reports the number of retained sections, the number of focal periods,
and an explicit reason whenever a p-value is unavailable. For example, ``no
realized focal sections'' means that none of the predetermined pooled sections
was constant in the realized path. This diagnostic is important when $m/L$ is
large: choosing $m$ does not alter the design, so long constant runs may be rare.

\subsection{Benchmarks and timing diagnostics}

The sharp-null FRT holds the observed outcomes fixed, redraws the complete block
assignment path, and recomputes the pure-history Horvitz--Thompson statistic.
Under the partially sharp total-effect null, mixed-history outcomes are not
imputable, so this FRT is a deliberately misspecified benchmark. The pairwise
imputation-based randomization test (PIRT) benchmark compares the statistic in both directions for an observed and a
randomly drawn schedule, using only histories that are pure under both paths
\citep{zhong2024unconditional}. All figures labeled ``PIRT (heur.)'' reject when
the raw pairwise probability is at most $0.05$. This remains a power benchmark,
not a formal level-$0.05$ rule under the assumptions of this paper. The
finite-sample-valid rule applies the factor-of-two calibration in
Theorem~\ref{thm:valid-pirt}.

For the carryover diagnostic, predetermined sections are paired. Outcomes after
the $m$-period burn-in in each even section are associated with the final block
label of the preceding odd section, yielding the statistic in
\eqref{eq:carry-stat}. The reference distribution fixes the even sections and
redraws the relevant odd-section labels. Its fixed-$m$ validity is conditional
on non-anticipation, as stated in Theorem~\ref{thm:valid-carryover}. For the non-anticipation diagnostic, the implementation uses the conditioning map
$\Gamma_h$ in~\eqref{eq:prefix-map}. Once $\Wobs$ is realized, it conditions on
$\mathcal C_h^{\mathrm{obs}}=\{\W_{1:h^+}=\Wobs_{1:h^+}\}$, fixes every complete
assignment block through $h^+$, and draws all later block labels independently
with their original probabilities. This is exactly the design law conditional
on the observed prefix cell. Fixing whole blocks is essential: overwriting only
part of a design block would create a path outside the support of the
switchback design. The confirmatory p-value is
$\min\{1,2p_{\mathrm{pair}}\}$; equivalently, a level-$0.05$ test rejects when
the raw pair probability is at most $0.025$.

For the centered next-assignment score, the boundary identity
in~\eqref{eq:Tna-boundary-identity} yields a simple exact calculation. Let $k_0$
be the final block fixed by the realized prefix, and let $s_j$ be the start of
post-prefix block $j$. Define $r_j=q_j$ when the observed label of block $j$ is
one and $r_j=1-q_j$ otherwise; thus $r_j$ is the conditional probability that
the reference label matches the observed label. If $j$ is the first block whose
reference label differs from the observed label, the two paths agree through
$s_j-1$, and the difference between the two directional scores has the sign of
\[
2\{1-2Z_j^{\mathrm{obs}}\}
\left(Y_{s_j-1}^{\mathrm{obs}}-\bar Y_{1:(s_j-1)}^{\mathrm{obs}}\right).
\]
With $K_{\mathrm{blk}}$ denoting the final block index, the exact raw pair probability is
\begin{align}
p_{\mathrm{pair}}
&=
\prod_{\ell=k_0+1}^{K_{\mathrm{blk}}}r_\ell\notag\\
&\quad+
\sum_{j=k_0+1}^{K_{\mathrm{blk}}}
\ind{
2\{1-2Z_j^{\mathrm{obs}}\}
\left(Y_{s_j-1}^{\mathrm{obs}}
-\bar Y_{1:(s_j-1)}^{\mathrm{obs}}\right)
\ge0
}
\left(\prod_{\ell=k_0+1}^{j-1}r_\ell\right)(1-r_j).
\label{eq:online-prefix-exact}
\end{align}
The first term is the conditional probability, within the realized prefix cell, that the two schedules never diverge. The $j$th summand is the conditional probability that block $j$ is the first mismatch and that the boundary-local directional comparison favors the reference schedule. The replication code evaluates this finite sum directly and checks it against exhaustive enumeration over all post-prefix block paths.

\subsection{Regression studentization for weak nulls}

For the weak-null analysis, a regression version of the studentized statistic
can reduce sensitivity to a large outcome level. Let $Z_j$ be the focal-session
label, $p_j=\pr(Z_j=1)$, and $\bar Y_j$ the observed focal mean. Define
\begin{equation}
\label{eq:online-hajek}
\widehat\mu_1=
\frac{\sum_j Z_j\bar Y_j/p_j}{\sum_j Z_j/p_j},
\qquad
\widehat\mu_0=
\frac{\sum_j(1-Z_j)\bar Y_j/(1-p_j)}
{\sum_j(1-Z_j)/(1-p_j)},
\qquad
\widehat\tau_{\mathrm{reg}}=\widehat\mu_1-\widehat\mu_0.
\end{equation}
The corresponding HC0 variance is
\begin{equation}
\label{eq:online-reg-var}
\widehat V_{\mathrm{reg}}=
\frac{\sum_{j:Z_j=1}(\bar Y_j-\widehat\mu_1)^2/p_j^2}
{(\sum_j Z_j/p_j)^2}
+
\frac{\sum_{j:Z_j=0}(\bar Y_j-\widehat\mu_0)^2/(1-p_j)^2}
{(\sum_j(1-Z_j)/(1-p_j))^2}.
\end{equation}
The statistic $\widehat\tau_{\mathrm{reg}}/\sqrt{\widehat V_{\mathrm{reg}}}$ is recomputed for every randomization draw. If either normalized-arm denominator is zero, the statistic is set to zero for both the observed assignment and any reference draw; this event has probability at most $2(1-\underline p)^J$ under the assumed assignment bounds. Its asymptotic validity under the
weak null is established in Section~\ref{app:weak-null} for one-sided conventional
levels $\alpha\le1/2$, using the exact reference law or a Monte Carlo
approximation with the number of draws increasing with the number of sessions.
The result separately requires nondegenerate limiting residual variation in the
conditional reference distribution, the sampling contrast, and the HC0
studentizer.
The regression model is a device for constructing the statistic, not an
identifying assumption.

\section{Stylized simulation design}
\label{app-online:toy}

\subsection{Common data-generating process}

The common potential-outcome model is equation~\eqref{eq:dgp_twosided} in the main paper. All stylized experiments set $\mu=0$, $\alpha_t=\log(t)$, and use the history-dependent scale
\begin{equation}
\label{eq:online-noise-scale}
\sigma_t(\w)
=
\begin{cases}
0, & t=1,2,\\
\alpha_t\ind{w_{t-2:t}\in\{\zero_3,\one_3\}}, & t\ge3.
\end{cases}
\end{equation}
The innovations are independent standard Gaussian or standard Cauchy draws.
Consequently, the total-effect, carryover, and anticipation simulations use the
same trend and disturbance mechanism; only the distributed-lag coefficients and
the test statistic change.

Assignments follow the regular switchback design of \citet{Bojinov2023} with
design horizon two and treatment probability $1/2$. The switch points are
$1,5,7,9,\ldots,T-3$; the first and last assignment blocks have length four and
the interior blocks have length two. We use
\[
T\in\{60,120,180,240,300\},
\qquad T/2\in\{30,60,90,120,150\}.
\]
The second set reports the design-scale horizon plotted on the horizontal
axis. Because the endpoint blocks are longer, the literal number of independently randomized assignment blocks is $K_{\mathrm{blk}}=T/2-2$.
Table~\ref{tab:online-toy-grid} records the full grid.

\begin{table}[t]
\centering
\caption{Stylized simulation grid}
\label{tab:online-toy-grid}
\begin{threeparttable}
\begin{tabular}{p{0.23\textwidth}p{0.68\textwidth}}
\toprule
Component & Specification\\
\midrule
Experimental horizons & $T\in\{60,120,180,240,300\}$\\
Design & Regular switchback; design horizon $2$; block probability $1/2$\\
Disturbances & Independent $N(0,1)$ or $t_1$ innovations\\
Total effect & $\delta_0=\delta_1=\delta_2=\delta$, $\delta\in\{0,1,2,3\}$\\
Carryover tails & Near, delayed, decaying, or concentrated; $\delta\in\{0,1.5,3\}$\\
Anticipation & $\delta_0=\delta_{-1}=\delta_{-2}=\delta$, $\delta\in\{0,1,2,3\}$\\
Prefix targets & $3,5,9,17$ periods, expanded to whole assignment blocks\\
Simulation replications & $1{,}000$ per configuration\\
Randomization draws & $1{,}000$ for total effects; $999$ for carryover; exact conditional probability for anticipation\\
Nominal level & $0.05$ for every confirmatory procedure; raw PIRT cutoff reported only as a supplement\\
\bottomrule
\end{tabular}
\begin{tablenotes}[flushleft]
\footnotesize
\item With 1,000 simulation replications, the Monte Carlo standard error of a
rejection rate is approximately $0.0069$ at a true rate of $0.05$ and is at most
$0.0158$ over all possible rates.
\end{tablenotes}
\end{threeparttable}
\end{table}

\subsection{Total-effect and carryover alternatives}

The total-effect experiment analyzes the correctly specified horizon $m=2$.
The null has $\delta_0=\delta_1=\delta_2=0$; under alternatives all three
coefficients equal $\delta\in\{1,2,3\}$. The section CRT, FRT-sharp, PIRT
heuristic, and the conservative Horvitz--Thompson asymptotic test all use
Horvitz--Thompson-style pure-history contrasts, but not identical focal sets.
The CRT uses post-burn-in periods in realized constant predetermined sections;
the FRT, PIRT, and asymptotic benchmark use their method-specific pure-history
sets. Their power can therefore differ because of both the focal observations
and the reference distribution.

For the carryover experiment, the tested null is $H_0^2$. Every alternative
places its signal strictly beyond lag two:
\begin{equation}
\label{eq:online-tail-shapes}
\begin{array}{ll}
\text{near:} & (\delta_3,\delta_4,\delta_5)=\delta(1,1,1),\\
\text{delayed:} & (\delta_6,\delta_7,\delta_8)=\delta(1,1,1),\\
\text{decaying:} & (\delta_3,\ldots,\delta_8)
=\delta(1,0.6,0.36,0.216,0.13,0.078),\\
\text{concentrated:} & \delta_3=3\delta,
\end{array}
\end{equation}
with all other post-horizon coefficients set to zero and
$\delta\in\{0,1.5,3\}$. The concentrated and near alternatives share early
timing but differ in dispersion; the near and delayed alternatives have equal
total coefficient mass but different locations. This decomposition shows that
the diagnostic's power is determined not only by effect magnitude but by where
the residual memory reaches the even-section focal outcomes.

\subsection{Anticipation alternatives and prefix information}

The anticipation experiment sets
\begin{equation}
\label{eq:online-anticipation-dgp}
Y_t=\log(t)+\delta W_t+\delta W_{t+1}+\delta W_{t+2}
+\sigma_t(\W)\varepsilon_t,
\end{equation}
omitting future terms outside $[T]$. Thus $\delta=0$ is the non-anticipation
null, while $\delta\in\{1,2,3\}$ introduces dependence on the next two
assignments. The contemporaneous term is included at the same scale and is a
nuisance for the non-anticipation null.

The target prefix lengths are $3$, $5$, $9$, and $17$. Because the first design block
has length four and subsequent blocks have length two, the implementation fixes
the complete block that contains each target endpoint. Reference paths can also
match the observed path for one or more subsequent blocks. Averaging over this
geometric first-mismatch event, the mean number of prefix-compatible periods is
approximately $6$, $8$, $12$, and $20$ and is essentially invariant to $T$, as shown
in Figure~\ref{fig:online-prefix-count}.
The four curves describe four separately prespecified tests; they are not a
rule for choosing a prefix after examining the outcomes.
Equation~\eqref{eq:Tna-boundary-identity} explains why power need not increase
smoothly with the full experimental horizon. Conditional on the first
post-prefix mismatch, the directional contrast is determined by the final
pre-mismatch outcome relative to the mean of the compatible prefix. A longer
experiment does not add further direction-specific score terms unless it
changes the location of that first mismatch.

\begin{figure}[t]
\centering
\includegraphics[width=0.86\textwidth]{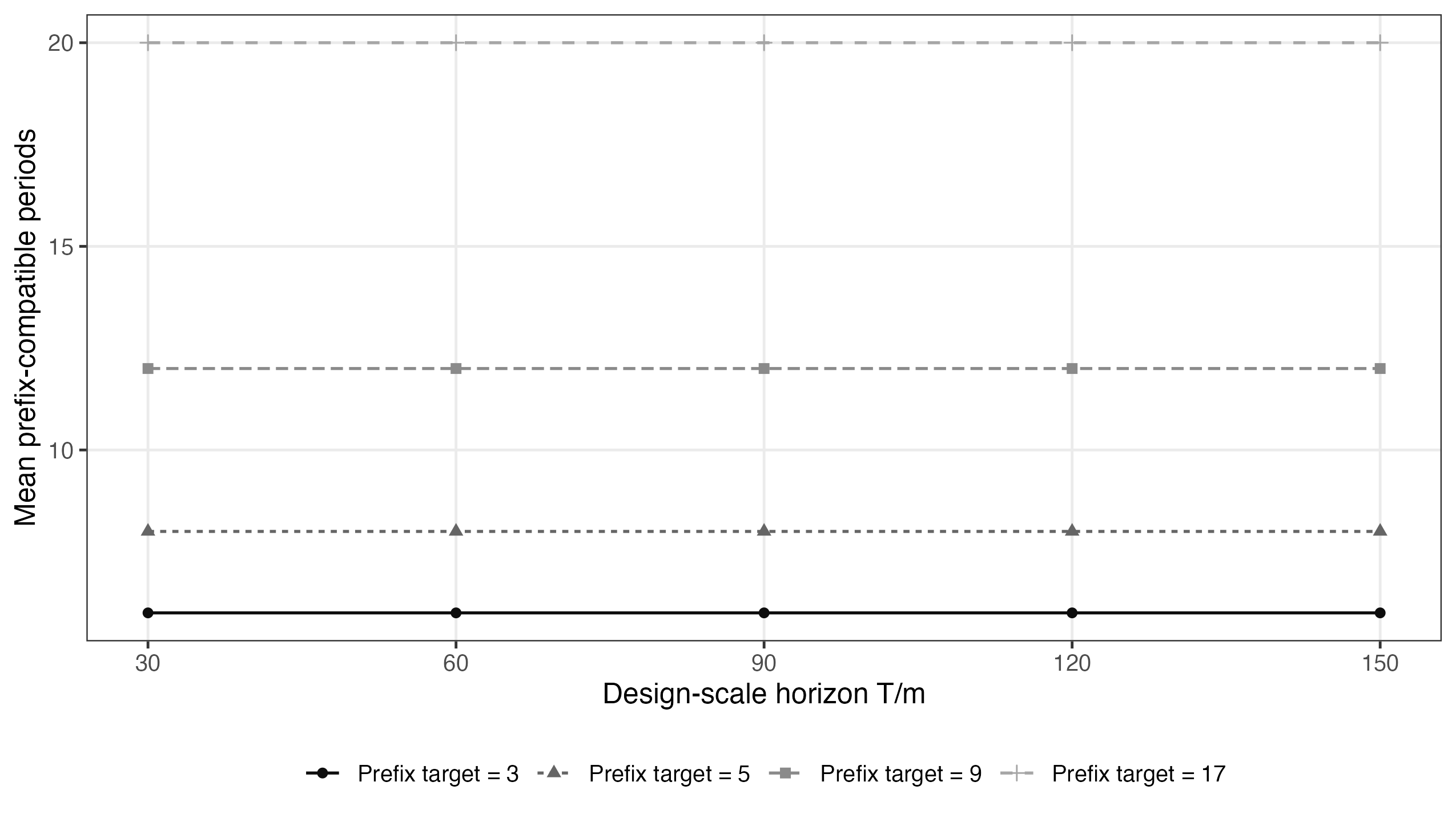}
\caption{Prefix-compatible periods in the anticipation simulation.}
\label{fig:online-prefix-count}
\notes{\textbf{Notes:} Each target is expanded to include every complete
assignment block that overlaps the prefix. The plotted quantity averages the
number of periods shared by the observed and resampled schedules.}
\end{figure}

For transparency, Figure~\ref{fig:online-ant-raw} reports the same experiment
using the uncalibrated rule $p_{\mathrm{pair}}\le0.05$. Its null rejection rates
are close to $0.05$ in this design, but that empirical pattern does not make the
raw probability a valid p-value. We therefore use the factor-of-two calibrated
result in the main paper.

\begin{figure}[t]
\centering
\includegraphics[width=\textwidth]{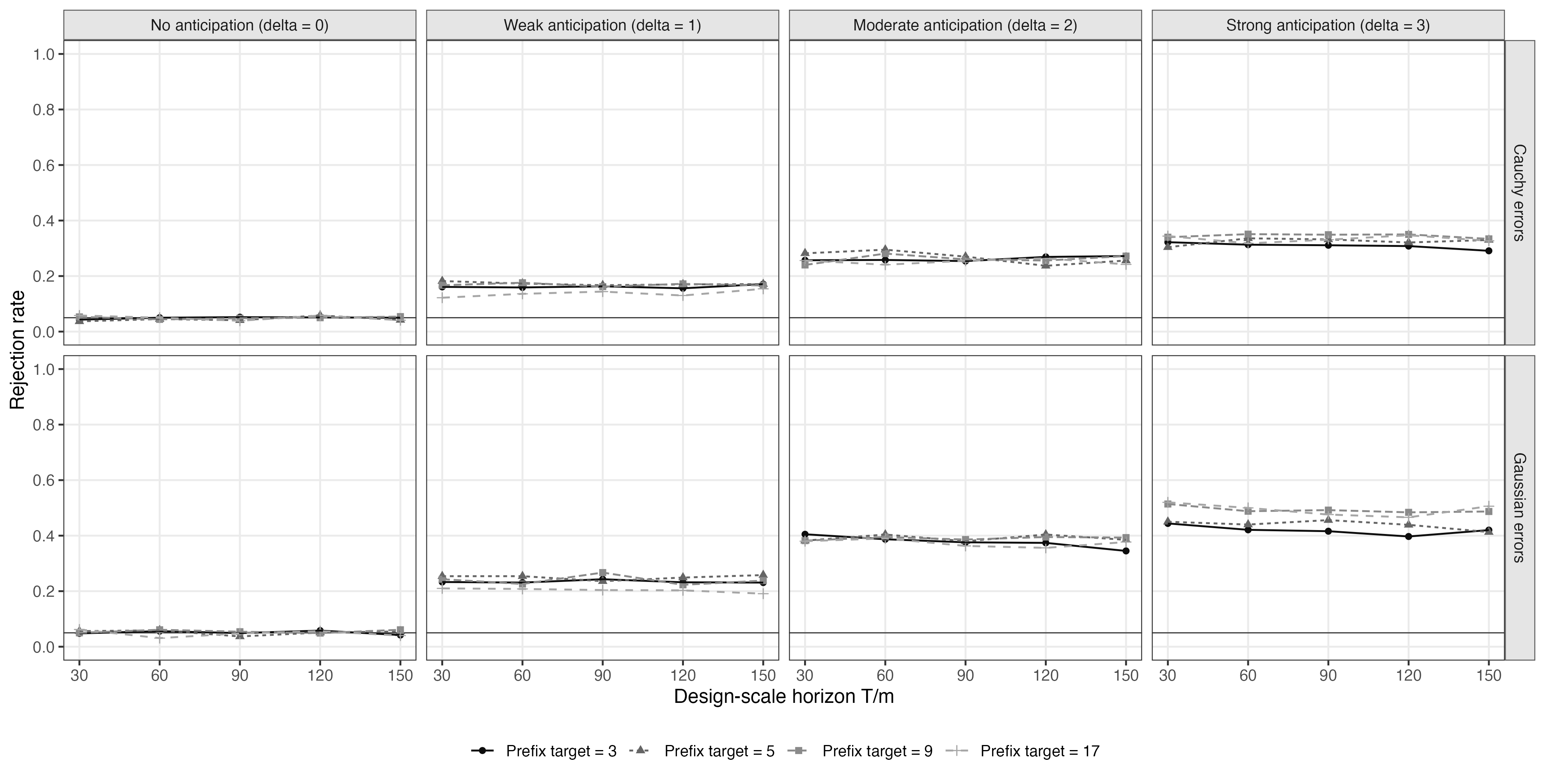}
\caption{Uncalibrated raw-PIRT anticipation diagnostic.}
\label{fig:online-ant-raw}
\notes{\textbf{Notes:} Curves reject when the raw conditional pairwise
probability is at most $0.05$. This figure is a heuristic power comparison; the
confirmatory figure in the main paper uses
$\min\{1,2p_{\mathrm{pair}}\}\le0.05$.}
\end{figure}

\section{Additional rideshare simulator results}
\label{app-online:rideshare}

\subsection{Environment, estimand, and benchmark effects}

The rideshare analysis uses the public \texttt{xp\_gym} pricing environment at
commit \texttt{fe89be31cf93}. Each trajectory has
300 cars and 300,000 request events. The outcome is platform reward per request,
not request-arrival rate. The treatment changes
\texttt{price\_per\_distance} from the baseline $0.010$ to $0.010$, $0.015$, or
$0.020$. The first contrast is a sharp placebo because both actions invoke the
same price rule.

The simulator benchmark is computed from paired always-baseline and
always-comparison trajectories. Within each pair, the initialization and event
random-number streams are held common. This common-random-number construction
reduces simulation noise but does not enter the switchback tests. Across 50
independent pairs, Table~\ref{tab:online-ate} reports the difference in average
reward per request.

\begin{table}[t]
\centering
\caption{Always-policy simulator benchmarks}
\label{tab:online-ate}
\begin{threeparttable}
\begin{tabular}{ccrrr}
\toprule
Baseline price & Comparison price & Price increase & Total effect & Monte Carlo SE\\
\midrule
0.010 & 0.010 & 0.000 & 0.000 & 0.000\\
0.010 & 0.015 & 0.005 & 0.730 & 0.003\\
0.010 & 0.020 & 0.010 & 1.292 & 0.005\\
\bottomrule
\end{tabular}
\begin{tablenotes}[flushleft]
\footnotesize
\item The benchmark is a simulator estimand: the average reward under sustained
comparison pricing minus the average reward under sustained baseline pricing.
It is known for the numerical environment because both counterfactual policy
trajectories can be generated; it would not be observed in a single field
experiment.
\end{tablenotes}
\end{threeparttable}
\end{table}

The switchback trajectories use fixed block lengths
$L\in\{1000,2500,5000,10000\}$, corresponding to $300$, $120$, $60$, and $30$
independently randomized blocks. We generate 50 independent trajectories for
each of the 12 price-by-block-length cells. Every trajectory is analyzed at
$m\in\{250,500,1000,2500,5000\}$ with 999 randomization draws per test. The resulting analysis grid contains $12\times50\times5=3{,}000$ trajectory--horizon combinations.

\subsection{Data-driven horizon selection}

For each trajectory, the exploratory two-stage analysis tests the candidate horizons in increasing order. When at least one null is not rejected, it sets
\begin{equation}
\label{eq:online-select-m}
\widehat m=\min\{m\in\{250,500,1000,2500,5000\}:
p_{\mathrm{carry}}(m)>0.05\}.
\end{equation}
The formal sequential diagnostic would report ``horizon $>5000$'' if every candidate were rejected. To produce a total-effect statistic for every simulated trajectory, the exploratory workflow instead uses $5000$ as a fallback in that event and records that the grid was exhausted. This is an outcome-dependent analysis choice. The finite-sample theorem applies to each
prespecified fixed $m$; it does not automatically imply conditional validity
given $\widehat m$. The main paper therefore treats the two-stage procedure as
an operational workflow and evaluates it directly in the simulator. A
confirmatory application can preserve fixed-$m$ validity by prespecifying the
horizon, using an independent pilot trajectory, or separating selection and
testing with a prespecified washout sufficient to remove cross-sample carryover.
An arbitrary split of a dependent trajectory is not enough.

Figure~\ref{fig:online-selected-m} shows that the diagnostic usually stops at the
smallest candidate, $m=250$. Larger selections occur most often for the shortest
price blocks and the larger price contrast. This pattern should not be read as
point identification of a true finite memory length: nonrejection can reflect
limited power, and the rideshare state process need not have exactly finite
memory.

\begin{figure}[t]
\centering
\includegraphics[width=\textwidth]{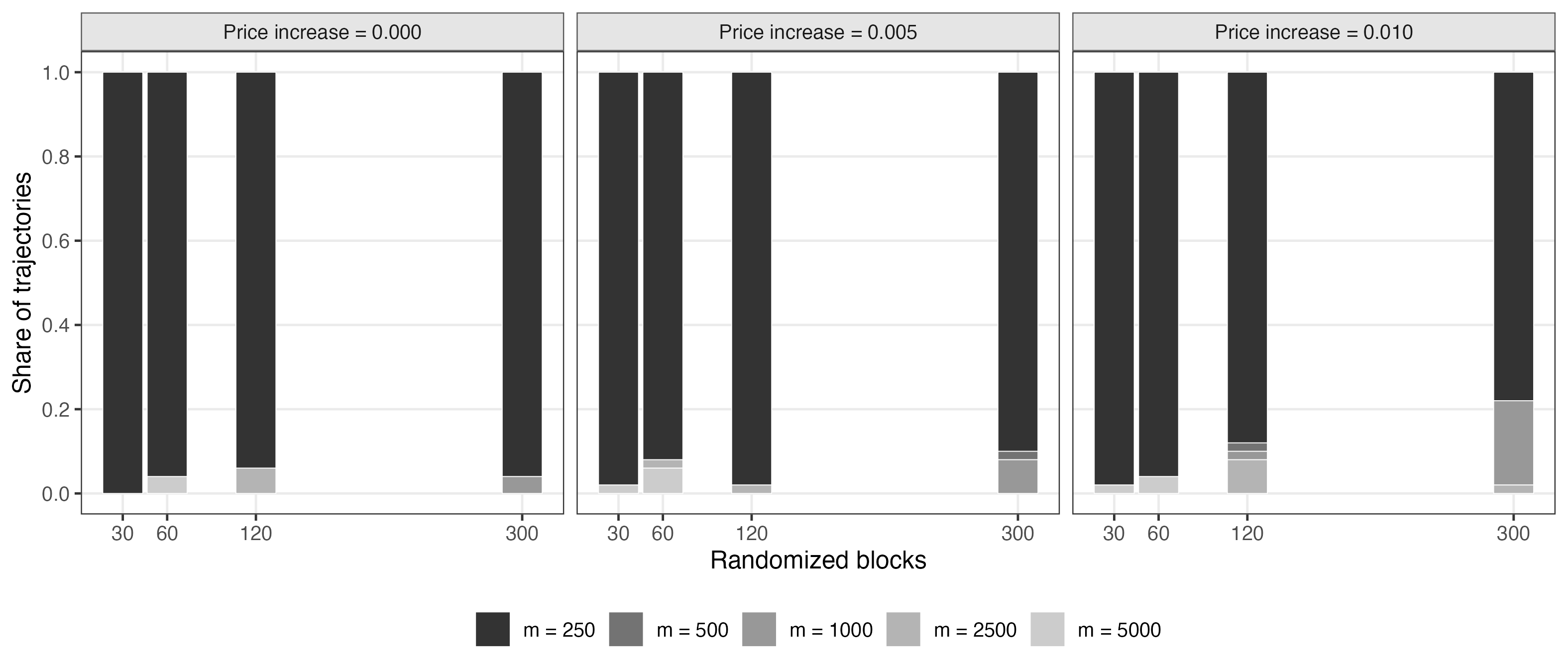}
\caption{Distribution of the selected analysis horizon in the rideshare experiment.}
\label{fig:online-selected-m}
\notes{\textbf{Notes:} Each bar summarizes 50 independent switchback
trajectories. Candidate horizons are tested from smallest to largest according
to~\eqref{eq:online-select-m}.}
\end{figure}

Table~\ref{tab:online-selected-results} gives the numeric source for the
paper-facing rideshare figure. Power rises with the price contrast and the
number of randomized blocks. With only 50 trajectories per cell, the Monte
Carlo standard error is $0.031$ at a rejection rate of $0.05$ and at most
$0.071$.

\begin{table}[t]
\centering
\caption{Rejection rates after horizon selection}
\label{tab:online-selected-results}
\begin{threeparttable}
\begin{tabular}{rrrrrr}
\toprule
Price increase & Blocks & Mean $\widehat m$ & CRT & FRT-sharp & PIRT (heur.)\\
\midrule
0.000 & 300 & 280 & 0.02 & 0.02 & 0.00\\
      & 120 & 385 & 0.00 & 0.00 & 0.00\\
      &  60 & 440 & 0.02 & 0.02 & 0.00\\
      &  30 & 250 & 0.04 & 0.06 & 0.04\\[0.05cm]
0.005 & 300 & 315 & 0.72 & 0.60 & 0.58\\
      & 120 & 295 & 0.54 & 0.52 & 0.54\\
      &  60 & 580 & 0.30 & 0.30 & 0.32\\
      &  30 & 345 & 0.18 & 0.16 & 0.16\\[0.05cm]
0.010 & 300 & 445 & 0.98 & 0.96 & 0.92\\
      & 120 & 450 & 0.84 & 0.82 & 0.82\\
      &  60 & 440 & 0.72 & 0.70 & 0.68\\
      &  30 & 345 & 0.48 & 0.44 & 0.46\\
\bottomrule
\end{tabular}
\begin{tablenotes}[flushleft]
\footnotesize
\item Each entry uses 50 independent trajectories. The PIRT column uses the raw
$0.05$ cutoff as a heuristic benchmark. The FRT imposes a sharp null that is not
implied by the total-effect null under dynamic state dependence.
\end{tablenotes}
\end{threeparttable}
\end{table}

\subsection{Placebo calibration exercise}

The original 50-trajectory placebo panel is too small to distinguish systematic
conservatism from simulation noise. We therefore hold each of the 50 placebo
reward paths fixed and generate 20 fresh switchback assignments, giving 1,000
pseudo-experiments per block count. The exercise repeats both the fixed-horizon CRT
and the full horizon-selection analysis. To limit computation, it uses 499
randomization draws per p-value; the outcome paths and selection rule are
otherwise unchanged.

\begin{table}[t]
\centering
\caption{Placebo re-randomization calibration exercise}
\label{tab:online-placebo-audit}
\begin{threeparttable}
\begin{tabular}{rrrrrr}
\toprule
Blocks & Fixed $m=250$ & Selected $m$ & Mean $\widehat m$ & $\pr(\widehat m>250)$ & Clustered SE\\
\midrule
300 & 0.040 & 0.041 & 282.8 & 0.039 & 0.0060\\
120 & 0.060 & 0.059 & 378.3 & 0.057 & 0.0084\\
 60 & 0.052 & 0.051 & 419.0 & 0.046 & 0.0074\\
 30 & 0.043 & 0.041 & 367.0 & 0.036 & 0.0071\\
\bottomrule
\end{tabular}
\begin{tablenotes}[flushleft]
\footnotesize
\item The standard error clusters the 20 re-randomizations by the 50 underlying
reward paths. The selected-$m$ rejection rates are close to their fixed-$m=250$
counterparts and to the nominal level. Among pseudo-experiments with a fixed
$m=250$ rejection, the horizon moves upward in roughly 33--48\% of cases, but
this screening does not produce a material downward shift in aggregate size.
\end{tablenotes}
\end{threeparttable}
\end{table}

The calibration exercise also reveals a boundary of the design. With 300 assignment blocks and
$m=5000$, only about 1.6 selected sections are available on average, and 227 of
the 1,000 pseudo-experiments have no usable CRT\@. This sparsity is a structural consequence of the design: a
5,001-request focal history must be assembled from blocks of length 1,000 and
must happen to be constant. The selected-horizon workflow almost never reaches
this sparse cell, but the diagnostic is a warning against interpreting very
large $m/L$ settings without reporting focal counts.

\subsection{CRT inversion under a working constant-effect model}

As a supplementary sensitivity analysis, consider the working restriction
\begin{equation}
\label{eq:online-constant-effect}
Y_t^{[m]}(\one_{m+1})
-
Y_t^{[m]}(\zero_{m+1})
=\tau
\qquad\text{for every selected focal time }t.
\end{equation}
For a candidate $\tau_0$, subtract $\tau_0W_t^{\mathrm{obs}}$ from each selected focal outcome and apply the section CRT to the adjusted outcomes. This inversion imposes a constant additive sustained-history effect on the focal experiment; it does not impose a constant-effect model on arbitrary mixed histories. The reported set is the convex hull of the values $\tau_0$ not rejected by equal-tailed tests at level $0.05$. The acceptance set is connected in all 600 reported trajectories, so the hull adds no values in this application. Coverage against the always-policy benchmark measures compatibility of this working restriction with the simulator. It is not a general coverage theorem for heterogeneous dynamic effects, and the outcome-selected horizon provides an additional reason to treat the intervals as diagnostics rather than exact confidence sets.

\begin{figure}[t]
\centering
\includegraphics[width=\textwidth]{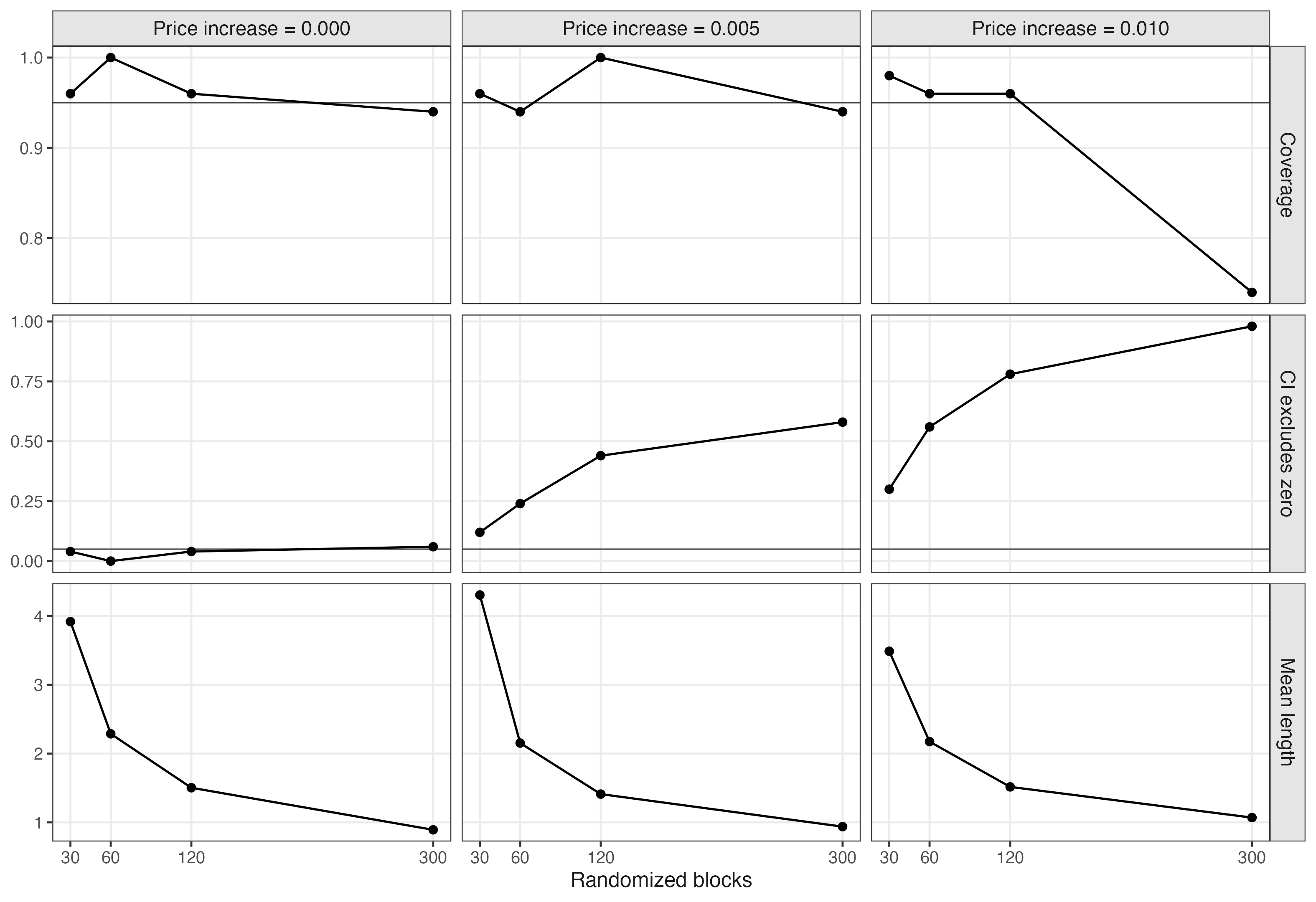}
\caption{CRT inversion under a constant-effect working model.}
\label{fig:online-ci}
\notes{\textbf{Notes:} Reported intervals are convex hulls of equal-tailed CRT
nonrejection sets. The horizontal line in the coverage row marks $0.95$;
the line in the zero-exclusion row marks $0.05$. Each point uses 50 trajectories
and the horizon selected on that trajectory. Mean interval length is reported in
reward units per request.}
\end{figure}

Figure~\ref{fig:online-ci} illustrates the expected block-length tradeoff.
Shorter blocks provide more randomized labels and shorter intervals; longer
blocks leave fewer randomized units and produce wider intervals. Placebo
coverage ranges from $0.94$ to $1.00$. For the $0.005$ price increase, coverage
ranges from $0.94$ to $1.00$. For the $0.010$ increase it is $0.74$ with 300
blocks and $0.96$--$0.98$ with 30--120 blocks. The short-block failure is
consistent with treatment-effect heterogeneity around state transitions: a
single additive shift does not reproduce the always-policy benchmark even when
the total-effect test has high power. We therefore use these intervals only as a model diagnostic, not as the main inferential result.

\FloatBarrier
\section{Replication map and computational details}
\label{app-online:replication}

The replication package has two layers. The cached layer includes compact CSV
files and regenerates every paper-facing figure without rerunning the expensive
rideshare environment. The full layer regenerates the simulator trajectories,
analyzes the five-horizon grid, and then rebuilds the summaries. The full
rideshare trajectories occupy approximately 8.3 GB and are not required for the
cached reproduction.

\begin{longtable}{p{0.34\textwidth}p{0.58\textwidth}}
\caption{Replication map}\label{tab:online-replication-map}\\
\toprule
Output & Primary source\\
\midrule
\endfirsthead
\toprule
Output & Primary source\\
\midrule
\endhead
Total-effect toy figure & Cached canonical summaries and the publication plotting script\\
Carryover-tail figure & Full carryover simulation script; seed 7711; 10 cores\\
Anticipation-prefix figures & Exact conditional-PIRT simulation; seed 8843; 10 cores; exhaustive-enumeration validator\\
Rideshare trajectories & \texttt{xp\_gym} trajectory exporter; design seed base 820000\\
Always-policy benchmarks & Paired policy runner; seed base 4242\\
Rideshare five-horizon grid & R analysis runner; seed base 910000; 999 draws\\
Selected-horizon results & Application summary script applied to the compact grid\\
CRT confidence intervals & Test-inversion script applied to one row per trajectory\\
Placebo calibration & Fixed reward paths, 20 re-randomizations per path, 499 draws\\
\bottomrule
\end{longtable}

Monte Carlo randomization p-values use the plus-one correction; the stylized
anticipation probability is evaluated analytically. Parallel workers receive
pregenerated independent seeds, so changing the number of cores does not change
the intended experiment grid. The package records the R session, Python package
versions, \texttt{xp\_gym} commit, command-line arguments, and checksums for the
cached numerical inputs. The primary reproduction command rebuilds the figures
from cached summaries; a separate full-run command is provided for users who
want to regenerate the request-level simulator data.

\printbibliography

\end{refsection}
\end{document}